\newif\ifsharedrootabove
\IfFileExists{./main.tex}
  {\sharedrootabovefalse}
  {\IfFileExists{../main.tex}{\sharedrootabovetrue}{\sharedrootabovefalse}}
\documentclass[11pt]{article}
\usepackage[margin=1in]{geometry}

\usepackage[utf8]{inputenc} % allow utf-8 input
\usepackage[T1]{fontenc}    % use 8-bit T1 fonts
\usepackage{xcolor}         % colors
\definecolor{arxivlinkblue}{rgb}{0,0,1}
\usepackage[colorlinks,
linkcolor=red,
anchorcolor=blue,
citecolor=blue
]{hyperref}
\makeatletter
\newcommand*{\rom}[1]{\expandafter\@slowromancap\romannumeral #1@}
\makeatother

\usepackage{makecell}
\usepackage{url}            % simple URL typesetting
\usepackage{booktabs}       % professional-quality tables
\usepackage{nicefrac}       % compact symbols for 1/2, etc.
\usepackage{microtype}      % microtypography
\usepackage{amsmath,amsfonts,amsthm,amssymb,bm,verbatim,dsfont,mathtools}
\usepackage{algorithmic}
\usepackage[ruled,vlined,linesnumbered]{algorithm2e}
\usepackage{color,graphicx,appendix}

\makeatletter
\ifsharedrootabove
  \def\input@path{{../}{./}}
\else
  \def\input@path{{./}{arxiv/}}
\fi
\makeatother
\ifsharedrootabove
  \graphicspath{{../}{./}}
\else
  \graphicspath{{./}{arxiv/}}
\fi
\usepackage{bbm}
\usepackage{subfig}
\usepackage{etoolbox}
\usepackage{tikz}
\usetikzlibrary{arrows.meta,shapes.geometric}
\usetikzlibrary{decorations.pathreplacing}
\usetikzlibrary{positioning, fit, backgrounds}
\tikzset{every picture/.append style={execute at begin picture={\colorlet{blue}{arxivlinkblue}}}}
\usepackage{xr,xspace}
\usepackage{todonotes}
\usepackage{paralist}
\usepackage{tabularx}
\usepackage[labelfont=bf,format=plain,justification=raggedright,singlelinecheck=false]{caption}
\usepackage{soul}
\usepackage{appendix}
\usepackage{style}
\hypersetup{anchorcolor=arxivlinkblue,citecolor=arxivlinkblue}
\usepackage{threeparttable}

\theoremstyle{plain}
\newtheorem{theorem}{Theorem}[section] 
\newtheorem{lemma}[theorem]{Lemma}

\newtheorem{proposition}[theorem]{Proposition}
\newtheorem{corollary}[theorem]{Corollary}
\newtheorem{claim}[theorem]{Claim}
\theoremstyle{definition}
\newtheorem{definition}[theorem]{Definition}

\newtheorem{remark}{Remark}
\newtheorem*{remark*}{Remark}

\usepackage{xspace,prettyref}

\newcommand{\naturals}{{\mathbb{N}}}

\newcommand{\supp}{{\rm supp}}

\newcommand{\red}{\color{red}}
\newcommand{\blue}{}
\colorlet{blue}{black}
\newcommand{\nb}[1]{{\sf\blue[#1]}}
\newcommand{\nbr}[1]{{\sf\red[#1]}}
\newcommand{\Expect}{\mathbb{E}}

\newcommand{\expect}[1]{\mathbb{E}\left[ #1 \right]}

\newcommand{\expects}[2]{\mathbb{E}_{#2}\left[ #1 \right]}

\newcommand{\prob}[1]{ \mathbb{P}\left\{ #1 \right\} }
\newcommand{\probs}[2]{ \mathbb{P}_{#2}\left\{ #1 \right\} }

\newcommand{\var}{\text{Var}}

\newcommand{\Bern}{{\rm Bern}}
\newcommand{\Binom}{{\rm Bin}}
\newcommand{\Pois}{{\rm Pois}}

\newcommand{\ie}{i.e.\xspace}

\newrefformat{eq}{(\ref{#1})}
\newrefformat{chap}{Chapter~\ref{#1}}
\newrefformat{sec}{Section~\ref{#1}}
\newrefformat{algo}{Algorithm~\ref{#1}}
\newrefformat{fig}{Fig.~\ref{#1}}
\newrefformat{tab}{Table~\ref{#1}}
\newrefformat{rmk}{Remark~\ref{#1}}
\newrefformat{clm}{Claim~\ref{#1}}
\newrefformat{def}{Definition~\ref{#1}}
\newrefformat{cor}{Corollary~\ref{#1}}
\newrefformat{lmm}{Lemma~\ref{#1}}
\newrefformat{prop}{Proposition~\ref{#1}}
\newrefformat{app}{Appendix~\ref{#1}}
\newrefformat{hyp}{Hypothesis~\ref{#1}}
\newrefformat{thm}{Theorem~\ref{#1}}

\newcommand{\indc}[1]{{\mathbf{1}_{\left\{{#1}\right\}}}}

\newcommand{\calA}{{\mathcal{A}}}
\newcommand{\calB}{{\mathcal{B}}}

\newcommand{\calD}{{\mathcal{D}}}
\newcommand{\calE}{{\mathcal{E}}}

\newcommand{\calG}{{\mathcal{G}}}

\newcommand{\calT}{{\mathcal{T}}}

\newcommand{\odd}{\mathrm{odd}}

\title{Optimality of Random Regular Graphs in Sparse Network Designs}
\author{Weijia Li \and Xiaochun (Nora) Niu \and Yehua Wei \and Jiaming Xu\thanks{W. Li is with the Department of Statistics and Data Science, Tsinghua University, \texttt{liwj25@mails.tsinghua.edu.cn}. X. Niu, Y. Wei, and J. Xu are with
 the Fuqua School of Business, Duke University, \texttt{\{xiaochun.niu,yehua.wei, jx77\}@duke.edu}. }}
\date{}

\begin{document}

\maketitle
% \nb{Shall we discuss the numerical experiments???}
% \nbr{comments from Itai: How to use random graph design in practice? Shall we state the results in terms of high probability? Shall we also comment on the derandomization??? }
\begin{abstract}
% Enter your abstract
The problems of designing sparse networks arise frequently in resource allocation and operations research. In production systems, for example, sparse process flexibility designs are used to handle uncertain demand effectively: the goal is to construct the sparsest bipartite graph between supply and demand that still achieves an expected fulfilled demand comparable to that of a fully flexible system. In middle-mile transportation, sparse delivery-route subgraphs that sustain large matchings after random node deletions help reduce delivery costs; here, the goal is to design the sparsest graph whose maximum matching size remains comparable to that of the fully connected graph under node deletions.

The design of sparse networks has been studied extensively, with state-of-the-art results providing order-wise optimal designs for both bipartite and unipartite networks \citep{chen2015optimal, feng2024designing}. 
However, identifying designs that achieve the sharp theoretical limit, where the average degree $d$ asymptotically matches the lower bound of any graph to achieve a given loss level, has remained open. In this paper, we prove that the random $d$-regular graph achieves this sharp optimal condition in both bipartite and unipartite settings. Numerical experiments further validate this optimality. Our results highlight a practical guideline for sparse flexibility networks: designs that combine degree regularity with dispersed edge placement can achieve optimal performance under uncertainty. 
\end{abstract}

% \nbr{Summary of main results:
% \begin{itemize}
% \item Flexibility Design
% \begin{itemize}
% \item Random $d$-regular bipartite graph: $d\ge (1+\delta)\log (1/\epsilon) / [\log (1/1-1/q)]$ for both a small constant $\epsilon>0$ and $\omega(1/(n\log n)) \le \epsilon \le o(1/n)$.
% \end{itemize}
% \item Maximum Matching
% \begin{itemize}
% \item Random $d$-regular graph: 
% $\expect{L} \le C np q^d$ 
% and $\expect{L}=o(1)$ if $d \ge [(1+\epsilon) /\log \frac{1}{q} ]\log n$;
% \item $d$-circular graph:
% $\expect{L} =\Theta(np q^{d/2})$ and $\expect{L}=O(1)$ if $d \ge[ 2/\log \frac{1}{q} ]\log n$;
% \item random $d$-geometric graph (one-dimensional): 
% $\expect{L} =\Theta(np\exp(-dp/2))$ and $\expect{L}=O(1)$ if $d \ge (2/p) \log n$
% \end{itemize}
% \end{itemize}

% Note that $\log(1/q)=-\log(1-p) \ge p$, where the equality is achieved asymptotically when $p \to 0.$}

% Why the second problem is challenging, and it's not a natural extension from the first one?

\section{Introduction}
Flexibility network design under uncertainty is an important problem in operations research. Such problems arise in %a wide range of applications, including 
manufacturing systems \citep{jordan1995principles}, service operations \citep{wallace2005staffing}, e-commerce fulfillment networks \citep{devalve2023understanding,feng2024designing}, kidney exchange \citep{blum2020ignorance}, and modern platforms such as ride-hailing and online labor markets \citep{freund2024two}. In these settings, a planner must commit ex ante to a limited set of feasible supply-demand connections that constrain operational decisions after uncertainty is realized. A key insight from this literature is that a carefully chosen set of such ``flexibility edges'' can achieve performance close to that of a fully flexible (complete) network while using only a small fraction of all possible connections \citep{jordan1995principles}. This insight has made sparse network design a cornerstone of modern service and fulfillment system engineering.

%Flexibility network design under uncertainty is an important problem in operations research, providing the structural foundation for resource pooling in various manufacturing, e-commerce and service domains \citep[see, e.g.,][]{jordan1995principles, wallace2005staffing, tsitsiklis2017flexible, devalve2023understanding, feng2024designing}. By adding ``flexibility edges'' between demand and supply nodes, firms can significantly improve their operational performance in the face of uncertainty, often capturing nearly all the benefits of a fully connected network with only a fraction of the edges \citep{jordan1995principles}. This remarkable efficacy of limited flexibility has made sparse network design a cornerstone of modern supply chain and service system engineering.

To date, research in this area has primarily focused on a \emph{sufficient design} perspective. These studies have been transformative, progressing from early structural designs to the establishment of rigorous \emph{order-wise optimality} results \citep[e.g.,][]{chen2015optimal, feng2024designing}. However, while these scaling laws identify the correct functional relationship between sparsity and performance, they do not characterize the sharp theoretical limit, i.e., the precise minimum average degree required to reach the efficiency frontier.

This paper shifts the focus toward the \emph{optimal design} perspective by investigating the following question: given a performance target $\epsilon$ (representing the fractional loss relative to full flexibility), how should a network be designed to meet this target with the minimum number of edges? We address this question in two network design settings, process flexibility and middle-mile transportation, corresponding to the fundamental cases of bipartite and unipartite networks, respectively. 
{\color{blue} As we shall see, random regular networks combine two complementary design principles, degree regularity and dispersed connections, and achieve optimal sparsity in both settings.}

%As we shall see, our study reveals that two simple design principles, degree regularity and randomization, 
%{\color{red} YW: I want to flag here because maybe we shouldn't say randomization is a principle; it's really about using randomization to achieve edge dispersion.}are the essential and \emph{sufficient} ingredients for achieving theoretically optimal sparsity. 

\subsection{Related Work}
The study of sparse flexibility network design began with the seminal work of \citet{jordan1995principles} in the context of manufacturing. More recently, flexibility design has attracted renewed interest in various contexts arising in online platforms and marketplaces. Within the vast flexibility design literature, our work distinguishes itself by shifting the focus from sufficiency (identifying sparse designs with at most $\epsilon$-fractional loss relative to full flexibility) to optimality (identifying the sparsest possible designs that satisfy a particular guarantee). Below, we focus on the most relevant papers to illustrate this distinction and refer readers to \citet{wang2021review} for a more comprehensive overview of the literature on flexibility network design.

\paragraph{Sufficient Sparse Designs.} The majority of the flexibility literature focuses on identifying sparse structures that are sufficient to approximate the performance of full flexibility. In the context of two-stage flexibility design, the foundational work is \citet{jordan1995principles}. They introduced the ``long chain'' flexibility design and demonstrated that it is often sufficient to capture most of the benefits of total flexibility. Subsequent work has focused on designing graph structures that achieve only an $\epsilon$-fractional loss relative to the fully flexible system for arbitrarily small $\epsilon$. These include $K$-chains, probabilistic expanders, and \ER random graphs \citep{chou2011process, wang2015process, chen2015optimal, chen2019optimal, shen2019reliable, feng2024designing}. Notably, \citet{chen2015optimal, chen2019optimal, feng2024designing} prove that their designs achieve order-wise optimality in various settings, as they achieve the best possible scaling of degree with respect to system size or loss. %More recently, \citep{freund2024two} study a two-sided platform setting, focusing on which nodes should be made flexible to maximize the expected maximum matching size under a randomized network design. Our results also relate to \cite{ameen2026uniformity}, who identify a uniformity principle in spatial matching markets: when allocating a fixed budget of service range (flexibility), distributing this capacity uniformly across supply nodes yields a larger expected matching size.
More recently, researchers have started to study alternative node-level flexibility decisions to maximize the expected maximum matching size, including which nodes should be made flexible under a randomized network design \citep{freund2024two} and how to allocate a fixed service-range budget across locations in spatial matching markets \citep{ameen2026uniformity}.

In flexibility design for dynamic systems, most of the literature similarly adopts a sufficiency perspective, showing that limited flexibility can guarantee strong performance. For instance, \cite{tsitsiklis2017flexible} demonstrate that random sparse architectures achieve near-optimal delay and throughput in flexible queueing systems, while \cite{asadpour2020online, xu2020online} show that sparse networks under appropriate dynamic policies %are sufficient to 
achieve strong guarantees under mild conditions. Recent research in online marketplaces has extended this notion of sufficiency to include non-structural dimensions and more complex resource dynamics. In the context of reusable resources, \citet{dong2024value} show that sparse, chain-like structures (modeling service upgrades) with an average degree less than three achieve the same heavy-traffic performance scaling as fully flexible systems. Related studies \citep{elmachtoub2015retailing, elmachtoub2019value, freund2025power} investigate the amount of opaque product, a form of demand flexibility, needed to achieve inventory balancing for e-retailers selling multiple horizontally differentiated products. Moreover, moving beyond the optimization of a single performance metric under limited flexibility, \citet{afeche2022optimal} study the design of bipartite queueing networks that explicitly balance the trade-off between two objectives: minimizing customer waiting times and maximizing matching rewards.

\paragraph{Optimally Sparse Designs.}
While some of the aforementioned works were successful in designing systems achieving order-wise optimality, the precise minimum average degree required to meet a specific target is often unknown for large systems. In fact, for two-stage flexibility designs, even when we restrict attention to networks with an average degree of two, the classical long-chain can be suboptimal \citep{desir2016sparse}. As a result, most of the literature here focuses on numerical methods for finding optimal or near-optimal networks \citep[e.g.,][]{santoso2005stochastic, feng2017process, yan2018design}. However, such methods are computationally intensive and, more importantly, yield instance-specific solutions rather than general structural design principles. 

%[Mark here. Read these.]

%\subsection{Our Contribution}
\subsection{Achieving Optimal Sparsity}
%We address this gap \nbr{It is a bit unclear which gap we refer to} in the context of two-stage stochastic systems, for both bipartite and unipartite networks, corresponding to process flexibility and middle-mile transportation network design, under the canonical i.i.d.\ scaled Bernoulli demand setting with some parameter $q$. Rather than proposing a heuristic and testing its sufficiency, we derive a lower bound on the number of edges required to achieve a target loss in large systems, then show that random regular networks are not merely sufficient but, in fact, match this theoretical lower bound and hence achieving optimal sparsity, as illustrated in Figure~\ref{fig:two_simple}. Specifically, we prove that the random regular networks (orange stars) reach this theoretical limit (within an arbitrarily small constant) in the high-performance regime where $\epsilon \to 0$, whereas existing constructions (blue dots) remain strictly to the right of this boundary.

%\nbr{Are we going to keep or remove the following paragraph?}\nb{fixed, I kept the paragraph below}
%\nbr{It is a bit unclear which gap we refer to}

\blue{We study optimally sparse designs for two representative two-stage flexibility problems: process flexibility on bipartite networks and middle-mile transportation on unipartite networks. Under a canonical i.i.d.\ scaled Bernoulli demand model, we first derive a fundamental lower bound on the average degree required to achieve a target fractional loss $\epsilon$. 
We then show that random regular networks asymptotically attain this bound in both problems in the high-performance regime as $\epsilon \to 0$.
For process flexibility, this result extends far beyond the Bernoulli model: the same random regular designs guarantee the desired performance uniformly over a broad class of bounded demand distributions, even when demands are heterogeneous and correlated. We also consider the stronger constant-loss criterion, under which the total loss remains bounded as the system grows, and show that random regular networks remain optimally sparse in this regime.} %The models and precise results are presented in Section~\ref{sect:overview}.}

Figure~\ref{fig:two_simple} illustrates how our results improve upon existing designs. Specifically, while existing designs (blue dots) remain strictly to the right of this boundary, we prove that random regular networks (orange stars) reach this frontier asymptotically.

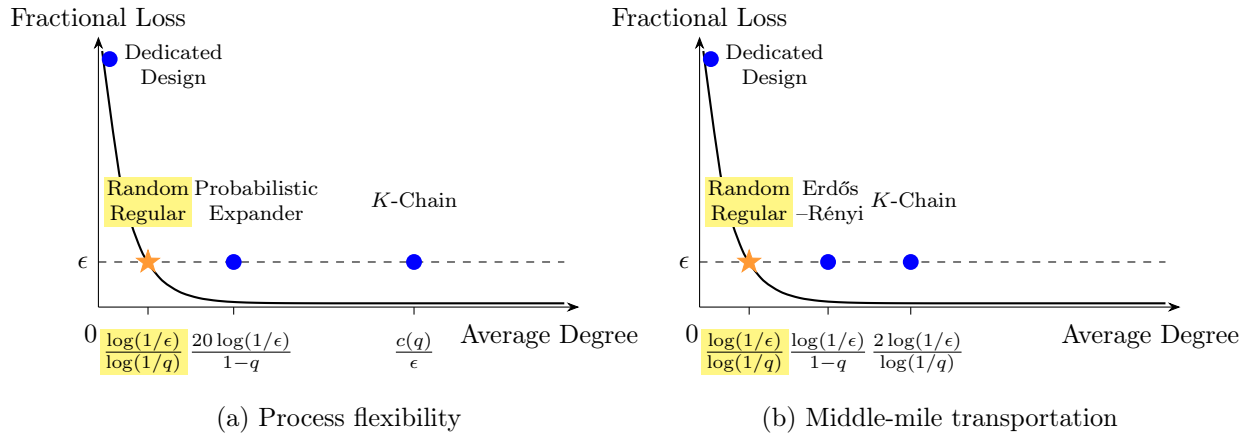
\begin{figure}[htbp]
    \centering
\resizebox{\textwidth}{!}{\begin{tikzpicture}[>=Stealth, every node/.style={font=\small}]

% ================= LEFT: Process flexibility =================
\begin{scope}
  % title
  % axes
  \draw[->] (0,0) -- (6.4,0) node[below] at (6,-0.1) {Average Degree};
  \draw[->] (0,0) -- (0,3.6) node[above] {Fractional Loss};

  % loss curve (schematic)
  \draw[thick,smooth,domain=0.05:6.2] plot (\x,{3.9*exp(-3*\x)+0.05});

  % epsilon n line
  \draw[dashed] (0, 0.6) -- (6.2,0.6);
  \node[left] at (0,0.6) {$\epsilon$};

  % dedicated & long chain
  \fill[blue] (0.15,3.3) circle (3pt);
  \node[inner sep=1pt,align=center,font=\scriptsize] at (1,3.2) {Dedicated\\Design};

  % random regular (star)
  \node[
    star,star points=5,star point ratio=2.5,
    fill=orange!80,draw=orange!80,
    minimum size=9pt,inner sep=1pt
  ] at (0.66,0.6) {};
  \node[fill=yellow!60,inner sep=1pt,align=center,font=\scriptsize] at (0.66,1.4) {Random\\Regular};

  % other designs on curve
  \fill[blue] (1.8,0.6) circle (3pt);
  \node[inner sep=1pt,align=center,font=\scriptsize] at (2.1,1.4) {Probabilistic\\Expander};

  \fill[blue] (4.2,0.6) circle (3pt);
  \node[above] at (4.2,1.2) {\scriptsize $K$-Chain};

  %\fill[blue] (6.4,0.1) circle (3pt);
  %\node[above] at (5.5,0.8) {Fully flexible};

  % x-axis ticks and labels
  \node[below] at (-0.1,-0.1) {$0$};

  \draw (0.66,0) -- (0.66,-0.1);
  \node[fill=yellow!60,inner sep=1pt,align=center,below=4pt]
        at (0.6,-0.15)
        {$\frac{\log(1/\epsilon)}{\log(1/q)}$};

  \draw (1.8,0) -- (1.8,-0.1);
  \node[inner sep=1pt,below=4pt] at (1.9,-0.15)
        {$\frac{20\log(1/\epsilon)}{1-q}$};

  \draw (4.2,0) -- (4.2,-0.1);
  \node[inner sep=1pt,below=4pt] at (4.2,-0.15)
        {$\frac{c(q)}{\epsilon}$};

        \node[font=\small] at (3.2,-1.5) {(a) Process flexibility};
\end{scope}

% ================= RIGHT: Middle-mile transportation ==========
\begin{scope}[xshift=8cm]
  % axes
  \draw[->] (0,0) -- (6.4,0) node[below] at (6,-0.1) {Average Degree};
  \draw[->] (0,0) -- (0,3.6) node[above] {Fractional Loss};

  % loss curve (schematic)
  \draw[thick,smooth,domain=0.05:6.2] plot (\x,{3.9*exp(-3*\x)+0.05});

  % epsilon n line
  \draw[dashed] (0, 0.6) -- (6.2,0.6);
  \node[left] at (0,0.6) {$\epsilon$};

  % dedicated & long chain
  \fill[blue] (0.15,3.3) circle (3pt);
  \node[inner sep=1pt,align=center,font=\scriptsize] at (1,3.2) {Dedicated\\Design};

  % random regular (star)
  \node[
    star,star points=5,star point ratio=2.5,
    fill=orange!80,draw=orange!80,
    minimum size=9pt,inner sep=0pt
  ] at (0.66,0.6) {};
  \node[fill=yellow!60,inner sep=1pt,align=center,font=\scriptsize] at (0.66,1.4) {Random\\Regular};

  % other designs on curve
  \fill[blue] (1.71,0.6) circle (3pt);
  \node[inner sep=1pt,align=center,font=\scriptsize] at (1.74,1.4) {Erd\H{o}s\\--R\'enyi};

  \fill[blue] (2.81,0.6) circle (3pt);
  \node[above] at (2.85,1.2) {\scriptsize $K$-Chain};

  %\fill[blue] (6.4,0.1) circle (3pt);
  %\node[above] at (5.5,0.8) {Fully flexible};

  % x-axis ticks and labels
  \node[below] at (-0.1,-0.1) {$0$};

  \draw (0.66,0) -- (0.66,-0.1);
  \node[fill=yellow!60,inner sep=1pt,align=center,below=4pt]
        at (0.6,-0.15)
        {$\frac{\log(1/\epsilon)}{\log(1/q)}$};

  \draw (1.71,0) -- (1.71,-0.1);
  \node[inner sep=1pt,below=4pt] at (1.71,-0.15)
        {$\frac{\log(1/\epsilon)}{1-q}$};

  \draw (2.81,0) -- (2.81,-0.1);
  \node[inner sep=1pt,below=4pt] at (2.9,-0.15)
{$\frac{2\log(1/\epsilon)}{\log(1/q)}$};

  % title
\node[font=\small] at (3.2,-1.5) {(b) Middle-mile transportation};
\end{scope}
\end{tikzpicture}}
\caption{Degree requirements for $\epsilon$-fractional loss. ``Dedicated'' denotes the degree-$1$ graph in process flexibility and the empty graph in transportation. In process flexibility, the $K$-chain is deterministic and regular, whereas the probabilistic expander introduced by \citet{chen2015optimal} is random and nonregular. In transportation, \citet{feng2024designing} derive degree requirements for the deterministic, regular $K$-chain and the random, nonregular Erdős--Rényi design (for constant loss); 
%\nbr{I remember they do not have a rigorous result for \ER, only a conjecture?}
we sharpen these requirements, as shown in the figure.}
% \caption{Tradeoff between average degree requirement ($x$-axis) and fractional demand loss ($y$-axis). ``Dedicated'' denotes the degree-$1$ graph in process flexibility, and empty graph in transportation. 
% In the process flexibility setting, $K$-chain is a deterministic and regular design, and probabilistic expander is a random and non-regular design from \cite{chen2015optimal}. 
% In the transportation setting,
% $K$-chain is a deterministic and regular design, and Erdős-Rényi is a classical random and non-regular design studied by \citet{feng2024designing}.}
%The corresponding results for the transportation setting are obtained in this work. , with results from
%Notably, the performance of $K$-chains in the two settings differs significantly, highlighting a fundamental difference between the problems.
\label{fig:two_simple}
\end{figure}

\subsection{Organization of the Paper}
The rest of the paper is organized as follows.
Section~\ref{sect:overview} provides an overview of the process flexibility and middle-mile transportation models, along with our main results and key insights. Section~\ref{sec:proc-flex-main} analyzes the process flexibility problem, and Section~\ref{sect:trans} analyzes the middle-mile transportation problem. Section~\ref{sect:exp} presents our numerical experiments, and Section~\ref{sect:conclude} concludes the paper. We defer some proof details to the supplementary file. Specifically, Sections~\ref{sect:app-process} and \ref{sect:app-transportation} complete the proofs of the main results in Sections~\ref{sec:proc-flex-main} and \ref{sect:trans}, respectively. Section \ref{sect:app-chain} proves the tight bounds on the matching loss for $K$-chains in the transportation setting.
%\yw{YW: Not sure where would be an appropriate place. But we should state something like: Throughout, we suppress floor and ceiling notation: whenever a quantity represents a cardinality, degree, or number of trials, it is understood to be rounded to the appropriate integer. Such rounding affects our bounds only by lower-order terms and does not alter any asymptotic conclusions.}

\section{Models and Overview of Main Results}\label{sect:overview}
This section introduces the two models and their residual-graph representations under Bernoulli demand. We then present common degree lower bounds across both settings, summarize our main results and design implications, and outline the key proof ideas.

%\subsection{Two Network Design Problems}
\paragraph{Process Flexibility.}
We begin with the classic process flexibility design problem.
%, first introduced in \cite{jordan1995principles}. 
In this setting, a firm manages a set of supply nodes, each with fixed capacity, and a set of demand nodes, each with random demand. 
After the demand is realized, the firm seeks to maximize total fulfilled demand by matching supply nodes to demand nodes. However, the firm must commit ex ante to a process flexibility network $G$, which specifies the feasible supply–demand connections.

We focus on the balanced system commonly studied in the literature, with $n$ supply nodes and $n$ demand nodes.
Each supply node has equal and deterministic capacity normalized to one. Demand is represented by a nonnegative random vector $D=(D_1,\ldots,D_n)$ with distribution $\mathcal D$. 
%Unless otherwise specified, the demand coordinates need not be identically distributed or independent. 
A process flexibility network is modeled as a bipartite graph $G = ([n], [n], E)$, where one partition consists of the $n$ supply nodes and the other of the $n$ demand nodes. An edge $(i, j) \in E$ indicates that supply node $i \in [n]$ can serve demand node $j \in [n]$. The maximum total fulfilled demand under network $G$ and demand realization $D$ %is denoted as $\E_{\cD}[z(G, D)]$, where $z(G, D)$ 
is defined by
\#\label{prob:max-flow}
z(G, D)  \triangleq & \max_{x_{ij} \ge 0} \sum_{(i,j)\in E} x_{ij} \notag\\
\text{s.t. } & \sum_{j\in [n]\colon (i,j)\in E} x_{ij} \le 1, i\in [n] \notag\\
 & \sum_{i\in [n]\colon (i,j)\in E} x_{ij} \le D_j,  j \in [n].
\#
In this formulation, $x_{ij}$ represents the amount of demand fulfilled at node $j$ using the capacity from supply node $i$. The first set of constraints ensures that the capacity at each supply node is not exceeded, while the second set ensures that the demand at each demand node is not oversatisfied. %Thus,  $z(G, D)$ can be interpreted as the (maximum) total demand satisfied under network $G$ given demand instance $D$.

%The formulation of $z(G, D)$ is also known as fractional bipartite matching problem or transportation problem
%In addition to $\E[z(G, D)]$, t
The performance of a process flexibility network $G$ is benchmarked against the complete bipartite network $K_{n,n}$, also known as the full-flexibility network. For a demand realization $D$, define the realized loss $L(G, D) \triangleq z(K_{n,n}, D)-z(G, D)$, which measures the shortfall in fulfilled demand under $G$ relative to full flexibility. For a demand distribution $\mathcal D$, define the \emph{expected loss} as
\$
L_{\cD}(G) \triangleq \E_{D\sim \cD}[L(G, D)].
\$ 
When the demand distribution is fixed and clear from context, we write $L(G)$ in place of $L_{\mathcal D}(G)$.
% We are particularly interested in designing a graph $G$ that achieves either $L(G) \le\epsilon n$ for an arbitrarily small constant $\epsilon$ (i.e., $\epsilon$-fractional loss),\footnote{This requirement is equivalent to the $(1-\epsilon)$-optimality criterion, $L(G) \le \E_D[z(K_{n,n}, D)]$, studied in the literature \citep[see, e.g.,][]{chen2015optimal}, since the expected fulfilled demand under full flexibility scale linearly with network size $n$.} or the stronger constant loss $L(G) \le O(1)$ independent of $n$. 
% It has been observed that well-designed sparse networks $G$ can achieve small losses in both theory and practice \citep[see, e.g.,][]{jordan1995principles,chen2015optimal}.
% \nb{I find the term ``small-loss regime'' imprecise. I would like to replace above with the paragraph below.} \blue{I rewrote it a little bit.}
We aim to design networks that achieve \emph{$\epsilon$-fractional loss}, defined by the requirement that $L_{\cD}(G) \le \epsilon n$ for large $n$.
%Since the expected fulfilled demand under full flexibility scales linearly with $n$, this is equivalent---up to a constant multiplier---to the $(1-\epsilon)$-fractional full flexibility criterion commonly used in the literature \citep[see, e.g.,][]{chen2015optimal}.  
%widely studied in the literature 
We also study the stronger \emph{constant-loss} guarantee $L_{\mathcal D}(G)=O(1)$, under which the additive performance gap remains bounded as the system grows.
%\blue{Looks great.}
%It has been observed that well-designed sparse networks can achieve such performance in both theory and practice \citep[see, e.g.,][]{jordan1995principles,chen2015optimal}. 

% and further developed in \cite{chou2011process, chen2015optimal, chen2019optimal}

%Although $z(G, D)$ is given by a simple linear optimization problem, the design of the underlying network $G$ is considerably more challenging. Specifically, identifying an optimal network that either (i) minimizes $|E|$ subject to a lower-bound requirement on $\E[z(G, D)]$, or (ii) maximizes $\E[z(G, D)]$ subject to a budget constraint on $|E|$, constitutes a difficult combinatorial optimization problem \cite[see, e.g.,][]{desir2016sparse, devalve2023approximate}.

\paragraph{Middle-Mile Transportation.}
Next, we turn to the middle-mile transportation network design problem studied by \cite{feng2024designing}. In this model, a firm serves stochastic demand at a set of demand nodes, called stations. \blue{The model assumes no split deliveries: each station's entire demand must be served by a single truck. This assumption may be appropriate when serving the same station with multiple trucks would entail substantial setup or coordination costs.} Before demand is realized, the firm commits to a transportation network $G$ specifying the feasible routes between stations. After observing demand, it serves all stations using the minimum possible number of trucks. Each truck has a fixed capacity and can serve at most two stations. Thus, stations $i$ and $j$ can share a truck only if $(i,j)$ is included in $G$ and their combined demand does not exceed the truck capacity. Minimizing the number of trucks is therefore equivalent to maximizing the number of station pairs served jointly.

Formally, suppose that all trucks have equal capacity normalized to one, and let $D$ be the realized demands drawn from distribution $\calD$ on $[0,1]^n$. 
A transportation network is modeled as a unipartite graph $G=([n],E)$, where each node represents a station. Given a transportation network $G$ and a demand realization $D$, we define $\mu(G, D)$ as the maximum number of trucks that serve two stations: %is denoted by $\mu(G, D)$, defined as 
%the optimal value of the following optimization problem:  
\#\label{prob:transportation}
\mu(G, D)  \triangleq & \max_{y_{e} \in \{0,1\}} \sum_{e \in E} y_e \notag\\
\text{s.t. } 
 & \sum_{e \sim j} y_e \le 1, \ j\in [n] \notag\\
 & y_{(i,j)} \le \indc{D_i+D_j \le 1}, \ (i,j) \in E.
\#
In this formulation, the decision variable $y_{(i,j)}$, defined for each edge, indicates whether stations $i$ and $j$ are served by a single truck ($y_{(i,j)}=1$) or not ($y_{(i,j)}=0$). 

Analogous to process flexibility, the performance of a transportation network $G$ is benchmarked against the complete network $K_n$. For a demand realization $D$, define the matching loss
$L(G,D)\triangleq
2\bigl(\mu(K_n,D)-\mu(G,D)\bigr)$.
The factor of two converts the loss in matched pairs into the corresponding number of stations that cannot be served by a single truck. For a demand distribution $\mathcal D$, define
\$
L_{\cD}(G) \triangleq \E_{D\sim \cD}[L(G,D)].
\$ 
%and the \emph{fractional loss}, given by the expected loss normalized by $2\E_{D}[\mu(K_{n}, D)]$. 
As in process flexibility, we write $L(G)$ when the demand distribution is clear from context. We also consider both $\epsilon$-fractional and constant loss levels for the middle-mile transportation problem.

\subsection{\blue{Bernoulli Demand and the Residual-Graph Connection}}  %Next, we compare the network performance metrics for process flexibility and middle-mile transportation. 
The two problems have distinct combinatorial structures. Process flexibility \eqref{prob:max-flow} is a \emph{fractional bipartite} $b$-matching problem, whereas middle-mile transportation \eqref{prob:transportation} is an \emph{integer unipartite} matching problem whose natural linear programming relaxation can have a nonzero integrality gap. Despite these differences, under the Bernoulli demand models studied in the literature, a connection emerges through their residual-graph representations.

Next, we describe the canonical Bernoulli demand models studied in the literature. In process flexibility, independently for each demand node $j$, demand equals $1/p$ with probability $p$ and $0$ with probability $q\triangleq 1-p$, as in \citet{chen2015optimal}. The scaling ensures that each node has expected demand one. 
In middle-mile transportation, independently for each station $j$, demand equals a half-truckload (i.e., $0.5$ units) with probability $p$ and a full truckload with probability $q$, following \citet{feng2024designing}.
Although the demand values differ, both models induce independent node deletion with probability $q$: 
zero-demand nodes are deleted in process flexibility, while full-truckload station nodes are removed in middle-mile transportation. The remaining nodes form a residual bipartite or unipartite graph, respectively, in which every isolated surviving node creates unavoidable loss. Figure~\ref{fig:node_deletion} illustrates this common mechanism. When the setting is clear from context, we use $\cD_p$ to denote the corresponding Bernoulli demand distribution in either model.

\begin{figure}[ht]
\centering
\begin{tikzpicture}[
    every node/.style={inner sep=1pt, minimum size=5mm, font=\scriptsize,thick},
    supplyNode/.style={rectangle, draw=red!70, fill=red!30},
    demandNode/.style={circle, draw=blue!70, fill=blue!30},
    ghostNode/.style={circle, draw=gray!70, densely dashed},
    isolateNode/.style={rectangle, draw=red!70},
    isolateNode1/.style={circle, draw=blue!70},
    solidEdge/.style={thick, blue!60!red!60},
    ghostEdge/.style={thick, gray!70, densely dashed},
]

%--------------------------------------------------
% LEFT: bipartite supply–demand graph
%--------------------------------------------------
\begin{scope}[xshift=0cm]

  % Titles
  \node at (-1.7,1.6) {Supply};
  \node at (1.7,1.6)  {Demand};

  % Supply nodes (squares)
  \node[supplyNode] (s1) at (-1.7, 1.0) {};
  \node[supplyNode] (s2) at (-1.7, 0.4) {};
  \node[supplyNode] (s3) at (-1.7,-0.2) {};
  \node[supplyNode] (s4) at (-1.7,-1.4) {};
  \node[isolateNode] (s5) at (-1.7,-2.0) {};

  % Vertical dots for more supply
  \node at (-1.7,-0.7) {\Large$\vdots$};

  % Demand nodes (solid circles)
  \node[demandNode] (d1) at ( 1.7, 1.0) {};
  \node[demandNode] (d2) at ( 1.7, -0.2) {};
  \node[demandNode] (d3) at ( 1.7,-1.4) {};

  % Ghost demand nodes (dashed circles)
  \node[ghostNode] (g1) at ( 1.7, 0.4) {};
  \node[ghostNode] (g2) at ( 1.7,-2.0) {};

  % Vertical dots for more demand
  \node at ( 1.7,-0.7) {\Large$\vdots$};

  % Solid edges (example pattern)
  \draw[solidEdge] (s1) -- (d1);
  \draw[solidEdge] (s2) -- (d1);
  \draw[solidEdge] (s3) -- (d3);
  \draw[solidEdge] (s3) -- (d2);
  \draw[solidEdge] (s4) -- (d3);

  % Ghost edges (some dashed connections to suggest "possible" links)
  \draw[ghostEdge] (s1) -- (g1);
  \draw[ghostEdge] (s2) -- (g1);
  \draw[ghostEdge] (s3) -- (g1);
  \draw[ghostEdge] (s5) -- (g2);
  \draw[ghostEdge] (s4) -- (g2);
\node at (0.0, -3.0) {\small (a) Process flexibility under $\cD_p$};
\end{scope}

%--------------------------------------------------
% RIGHT: general (projected) network
%--------------------------------------------------
\begin{scope}[xshift=6.5cm]

  % Solid nodes
  \node[demandNode] (a1) at (-1.0, 0.5) {};
  \node[demandNode] (a2) at (-1.0,-0.8) {};
  \node[demandNode] (a3) at (0.2,-1.8) {};
  \node[demandNode] (a4) at ( 1.4,-0.6) {};
  \node[demandNode] (a5) at ( 2.4, 0.5) {};
  \node[isolateNode1] (a6) at ( 3.2,-0.8) {};

  % Vertical dots on the left chain
  \node at (-1.0, -0.1) {\Large$\vdots$};

  % Ghost nodes
  \node[ghostNode] (gA) at ( 0.8, 1.2) {};
  \node[ghostNode] (gB) at ( 2.2,-1.8) {};

  % Solid edges
  \draw[solidEdge] (a2) -- (a3);
  \draw[solidEdge] (a1) -- (a3);
  \draw[solidEdge] (a3) -- (a4);
  \draw[solidEdge] (a4) -- (a5);

  % Ghost edges
  \draw[ghostEdge] (a1) -- (gA);
  \draw[ghostEdge] (gA) -- (a5);
  \draw[ghostEdge] (gA) -- (a3);
  \draw[ghostEdge] (gB) -- (a5);

  \draw[ghostEdge] (a6) -- (gB);
  \draw[ghostEdge] (gB) -- (a3);
\node at (0.9, -3.0) {\small (b) Middle-mile transportation under $\cD_p$};
\end{scope}
\end{tikzpicture}

\caption{Sparse network design under Bernoulli demand. In (a), the gray dashed demand nodes have zero demand, and the unfilled supply node is isolated in the residual bipartite graph. In (b), the gray dashed nodes have full-truckload demand and do not enter the matching problem; the unfilled surviving node is isolated in the residual unipartite graph. In both models, the dashed nodes occur independently with probability $q=1-p$, and the residual graph is obtained by deleting them and their incident edges.}

\label{fig:node_deletion}
\end{figure}
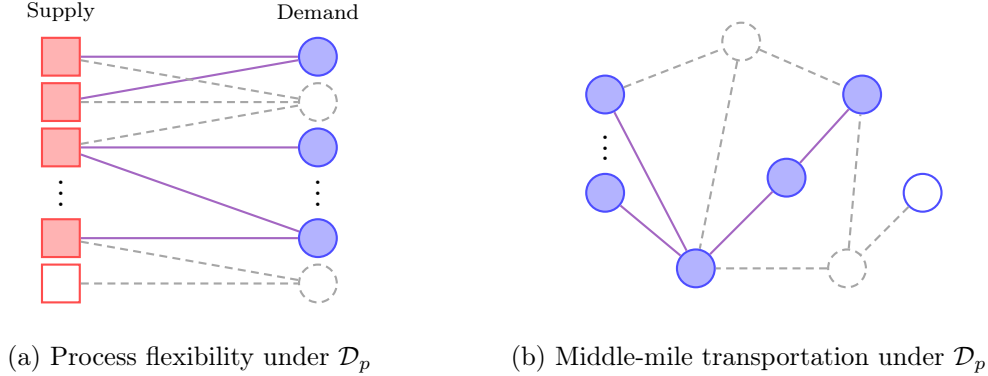

This shared deletion mechanism yields the same leading-order degree lower bound in the two problems. Let $d_i$ denote the degree of supply node $i$ and station $i$ for the process flexibility and transportation setting, respectively. Let $d$ be the corresponding average degree. In process flexibility, node $i$ becomes isolated with probability $q^{d_i}$; in middle-mile transportation, it must first survive and then have all its neighbors deleted, which occurs with probability $pq^{d_i}$. Because each isolated node creates unavoidable loss and \(\sum_{i=1}^n q^{d_i}\ge nq^d\) 
by Jensen's inequality, the
loss lower bounds in both settings scale as $nq^d$. Consequently, in both settings, achieving $\epsilon$-fractional loss requires
\(
d\ge \frac{\log(1/\epsilon)}{\log(1/q)}-O(1),
\)
whereas achieving constant loss requires
\(
d\ge \frac{\log n}{\log(1/q)}-O(1).
\)
As we show in this paper, random regular bipartite and unipartite graphs attain these bounds asymptotically in process flexibility and middle-mile transportation, respectively.

\subsection{\blue{Summary of Main Results and Design Implications}}
%{\color{red}YW: We probably want to change this sentence a bit as we have more general result.} 
%Our goal is to design the sparsest network with $n$ supply and $n$ demand nodes (or $n$ nodes in the unipartite setting) that achieves either $\epsilon$-fractional loss, $L_{\cD}(G)\le \epsilon n$, or constant loss, $L_{\cD}(G) = O(1)$ (independent of $n$). We begin with the Bernoulli demand model $\cD_p$.

The lower bounds above identify, to leading order, the minimum average degree required for \(\epsilon\)-fractional and constant loss. 
We call a design asymptotically \emph{optimally sparse} in a given loss regime if the average degree required to attain the target loss matches the corresponding lower bound to leading order. This sharp notion of sparsity strengthens the order-wise guarantees in prior work, which identify the correct scaling in $\epsilon$ or $n$ but not the leading constant.

Our main results can be summarized as follows.
For process flexibility, we show that, with high probability over the graph construction, a random regular bipartite graph is (asymptotically) optimally sparse in both $\epsilon$-fractional and constant loss regimes, uniformly over all demand distributions supported on $[0,1/p]^n$.
Compared with \citet{chen2015optimal}, our result identifies the sharp leading-order degree and extends the guarantee to bounded demands that may be non-identically distributed and arbitrarily correlated.
For middle-mile transportation under $\cD_p$, we likewise show that, with high probability, a random regular graph is optimally sparse in either $\epsilon$-fractional or constant loss regimes, sharpening the sparsity guarantees obtained in \citet{feng2024designing}.

%by identifying the optimal leading-order degree.

%Because $\cD_p^{\mathrm{PF}}$ belongs to this class and establishes the lower bounds, random regular graphs are asymptotically minimax-optimal over the entire bounded-support class. 

%\paragraph{Design Implications.}
Together, our results provide a theoretical justification for combining two longstanding principles in flexibility design: (i) \emph{regularity}, which distributes the edge budget evenly across nodes, and (ii) \emph{dispersion}, which spreads connections throughout the network. These ideas trace back to the flexibility principles articulated by \citet{jordan1995principles} and have informed much of the subsequent literature. What has remained unresolved, however, is whether combining them yields an optimally sparse design. Our results answer this question affirmatively in both settings studied here.

Specifically, degree regularity minimizes, for a fixed edge budget, the expected number of nodes that become isolated after node deletion. Dispersion, often implemented through randomized edge placement with limited dependence among edges, mitigates the loss caused by local demand imbalances and structural bottlenecks. 
Either principle alone can be insufficient to attain optimal sparsity. Random but nonregular designs, such as probabilistic expanders and \ER graphs, can leave low-degree nodes particularly vulnerable to isolation, whereas deterministic regular designs, such as $K$-chains, concentrate their connections locally. Random regular graphs combine both principles and thereby attain the minimum possible average degree to leading order.

\begin{figure}[htbp]
\centering
\begin{tikzpicture}[
    every node/.style={inner sep=1pt, minimum size=3.8mm, font=\scriptsize, thick},
    leftNode/.style={rectangle, draw=red!70, fill=red!30},
    rightNode/.style={circle, draw=blue!70, fill=blue!30},
    fullEdge/.style={thick, blue!60!red!60},
    chainEdge/.style={thick, blue!60!red!60},
]

% =====================================================
% (a) Full bipartite K_{8,8}
% =====================================================
\begin{scope}[xshift=0cm, yshift = -0.5cm]
\node at (0.1, 0) {Supply};
\node at (2.8, 0) {Demand};
    
    % Left nodes
    \foreach \i in {1,...,8} {
        \node[leftNode] (aL\i) at (0, -0.46*\i) {$\i$};
    }
    % Right nodes
    \foreach \j in {1,...,8} {
        \node[rightNode] (aR\j) at (3, -0.46*\j) {$\j$};
    }

    % K_{8,8} edges
    \foreach \i in {1,...,8} {
        \foreach \j in {1,...,8} {
            \draw[fullEdge] (aL\i) -- (aR\j);
        }
    }

    \node at (1.5, -4.8) {\small (a) Full bipartite $K_{8,8}$};
\end{scope}

% =====================================================
% (b) 8×8 Long Chain
% =====================================================
\begin{scope}[xshift=4.1cm, yshift =-0.5cm]
\node at (0.1, 0) {Supply};
\node at (2.8, 0) {Demand};
    % Left nodes
    \foreach \i in {1,...,8} {
        \node[leftNode] (bL\i) at (0, -0.46*\i) {$\i$};
    }
    % Right nodes
    \foreach \j in {1,...,8} {
        \node[rightNode] (bR\j) at (3, -0.46*\j) {$\j$};
    }

    % chain edges
    \foreach \i in {1,...,7} {
        \pgfmathtruncatemacro{\ip}{\i+1}
        \draw[chainEdge] (bL\i) -- (bR\i);
        \draw[chainEdge] (bL\i) -- (bR\ip);
    }
    \draw[chainEdge] (bL8) -- (bR8);
    \draw[chainEdge] (bL8) -- (bR1);

    \node at (1.5, -4.8) {\small (b) Long chain};
\end{scope}

% =====================================================
% (c) Probabilistic expander (matrix 1)
% =====================================================
\begin{scope}[xshift=8.2cm, yshift = -0.5cm]
\node at (0.1, 0) {Supply};
\node at (2.8, 0) {Demand};
% ORDERED BY DEGREE (high → low)
  % Left nodes, ordered by degree:
  % order: 1, 4, 5, 7, 3, 2, 6, 8
  %-------------------------------------------------------
  \node[leftNode] (cL1) at (0, -0.46*1) {$1$};
  \node[leftNode] (cL4) at (0, -0.46*2) {$2$};
  \node[leftNode] (cL5) at (0, -0.46*3) {$3$};
  \node[leftNode] (cL7) at (0, -0.46*4) {$4$};
  \node[leftNode] (cL3) at (0, -0.46*5) {$5$};
  \node[leftNode] (cL2) at (0, -0.46*6) {$6$};
  \node[leftNode] (cL6) at (0, -0.46*7) {$7$};
  \node[leftNode] (cL8) at (0, -0.46*8) {$8$};

  %-------------------------------------------------------
  % Right nodes, ordered by degree:
  % order: 2, 1, 6, 5, 3, 4, 7, 8
  %-------------------------------------------------------
  \node[rightNode] (cR2) at (3, -0.46*1) {$1$};
  \node[rightNode] (cR1) at (3, -0.46*2) {$2$};
  \node[rightNode] (cR6) at (3, -0.46*3) {$3$};
  \node[rightNode] (cR5) at (3, -0.46*4) {$4$};
  \node[rightNode] (cR3) at (3, -0.46*5) {$5$};
  \node[rightNode] (cR4) at (3, -0.46*6) {$6$};
  \node[rightNode] (cR7) at (3, -0.46*7) {$7$};
  \node[rightNode] (cR8) at (3, -0.46*8) {$8$};
% -------------------------------------------------------
% Edges from matrix E_{ij} = 1
% -------------------------------------------------------

% Row 1: L1
\draw[chainEdge] (cL1) -- (cR1);
\draw[chainEdge] (cL1) -- (cR2);
\draw[chainEdge] (cL1) -- (cR3);
\draw[chainEdge] (cL1) -- (cR5);
\draw[chainEdge] (cL1) -- (cR6);
\draw[chainEdge] (cL1) -- (cR8);

% Row 2: L2
\draw[chainEdge] (cL2) -- (cR1);
\draw[chainEdge] (cL2) -- (cR5);
\draw[chainEdge] (cL2) -- (cR7);

% Row 3: L3
\draw[chainEdge] (cL3) -- (cR1);
\draw[chainEdge] (cL3) -- (cR2);
\draw[chainEdge] (cL3) -- (cR4);
\draw[chainEdge] (cL3) -- (cR6);

% Row 4: L4
\draw[chainEdge] (cL4) -- (cR2);
\draw[chainEdge] (cL4) -- (cR3);
\draw[chainEdge] (cL4) -- (cR4);
\draw[chainEdge] (cL4) -- (cR5);
\draw[chainEdge] (cL4) -- (cR6);

% Row 5: L5
\draw[chainEdge] (cL5) -- (cR1);
\draw[chainEdge] (cL5) -- (cR2);
\draw[chainEdge] (cL5) -- (cR4);
\draw[chainEdge] (cL5) -- (cR6);
\draw[chainEdge] (cL5) -- (cR8);

% Row 6: L6
\draw[chainEdge] (cL6) -- (cR2);
\draw[chainEdge] (cL6) -- (cR6);

% Row 7: L7
\draw[chainEdge] (cL7) -- (cR1);
\draw[chainEdge] (cL7) -- (cR2);
\draw[chainEdge] (cL7) -- (cR3);
\draw[chainEdge] (cL7) -- (cR5);
\draw[chainEdge] (cL7) -- (cR7);

% Row 8: L8
\draw[chainEdge] (cL8) -- (cR1);
\draw[chainEdge] (cL8) -- (cR2);

    \node at (1.5, -4.8) {\small (c) Probabilistic expander};
\end{scope}

% =====================================================
% (d) Random regular graph (matrix 2)
% =====================================================
\begin{scope}[xshift=12.3cm, yshift = -0.5cm]
\node at (0.1, 0) {Supply};
\node at (2.8, 0) {Demand};

% Left nodes L1..L8, but swap
\node[leftNode] (dL7) at (0, -0.46*1) {$1$}; 
\node[leftNode] (dL5) at (0, -0.46*2) {$2$};
\node[leftNode] (dL3) at (0, -0.46*3) {$3$};
\node[leftNode] (dL4) at (0, -0.46*4) {$4$};
\node[leftNode] (dL2) at (0, -0.46*5) {$5$};
\node[leftNode] (dL6) at (0, -0.46*6) {$6$};
\node[leftNode] (dL1) at (0, -0.46*7) {$7$}; 
\node[leftNode] (dL8) at (0, -0.46*8) {$8$};
    % Right nodes
\node[rightNode] (dR1) at (3, -0.46*1) {$1$}; 
\node[rightNode] (dR2) at (3, -0.46*2) {$2$};
\node[rightNode] (dR7) at (3, -0.46*3) {$3$};
\node[rightNode] (dR4) at (3, -0.46*4) {$4$};
\node[rightNode] (dR5) at (3, -0.46*5) {$5$};
\node[rightNode] (dR6) at (3, -0.46*6) {$6$};
\node[rightNode] (dR3) at (3, -0.46*7) {$7$}; 
\node[rightNode] (dR8) at (3, -0.46*8) {$8$};

    % --- matrix edges ---
    % Row 1: [1,1,0,0,1,0,1,0]
\draw[chainEdge] (dL1) -- (dR1);
\draw[chainEdge] (dL1) -- (dR2);
\draw[chainEdge] (dL1) -- (dR5);
\draw[chainEdge] (dL1) -- (dR7);

% Row 2: [1,1,0,1,0,0,0,1]
\draw[chainEdge] (dL2) -- (dR1);
\draw[chainEdge] (dL2) -- (dR2);
\draw[chainEdge] (dL2) -- (dR4);
\draw[chainEdge] (dL2) -- (dR8);

% Row 3: [0,1,1,0,0,1,0,1]
\draw[chainEdge] (dL3) -- (dR2);
\draw[chainEdge] (dL3) -- (dR3);
\draw[chainEdge] (dL3) -- (dR6);
\draw[chainEdge] (dL3) -- (dR8);

% Row 4: [1,0,0,1,0,1,0,1]
\draw[chainEdge] (dL4) -- (dR1);
\draw[chainEdge] (dL4) -- (dR4);
\draw[chainEdge] (dL4) -- (dR6);
\draw[chainEdge] (dL4) -- (dR8);

% Row 5: [0,0,1,1,0,1,1,0]
\draw[chainEdge] (dL5) -- (dR3);
\draw[chainEdge] (dL5) -- (dR4);
\draw[chainEdge] (dL5) -- (dR6);
\draw[chainEdge] (dL5) -- (dR7);

% Row 6: [0,0,0,1,1,0,1,1]
\draw[chainEdge] (dL6) -- (dR4);
\draw[chainEdge] (dL6) -- (dR5);
\draw[chainEdge] (dL6) -- (dR7);
\draw[chainEdge] (dL6) -- (dR8);

% Row 7: [0,1,1,0,1,0,1,0]
\draw[chainEdge] (dL7) -- (dR2);
\draw[chainEdge] (dL7) -- (dR3);
\draw[chainEdge] (dL7) -- (dR5);
\draw[chainEdge] (dL7) -- (dR7);

% Row 8: [1,0,1,0,1,1,0,0]
\draw[chainEdge] (dL8) -- (dR1);
\draw[chainEdge] (dL8) -- (dR3);
\draw[chainEdge] (dL8) -- (dR5);
\draw[chainEdge] (dL8) -- (dR6);

    \node[draw=none, font=\small] at (1.5, -4.8) {(d) Random regular};
\end{scope}

\begin{scope}[xshift=1.5cm, yshift = -8cm]
% radius of the circle
\def\r{1.5}

% Angles for a regular octagon with *two top nodes*
% (regular octagon rotated by +22.5 degrees)
\foreach \i in {1,...,12} {
    \node[rightNode] (N\i) at ({90 - 360/12*(\i-1)}:\r) {$\i$};
}

% Fully connected K8
\foreach \i in {1,...,12} {
    \foreach \j in {1,...,12} {
        \ifnum\i<\j
            \draw[chainEdge] (N\i) -- (N\j);
        \fi
    }
}
\node[draw=none, font=\small] at (0, -2.4) {(e) Full $K_{12}$};
\end{scope}

\begin{scope}[xshift=5.6cm, yshift = -8cm]
    
% radius chosen so width = 2.6cm (same as bipartite figs)
\def\r{1.5}

% 8 nodes on a circle, rotated so two are on top & two on bottom
% angles: regular octagon rotated by 22.5 degrees
\foreach \i in {1,...,12} {
    \node[rightNode] (b\i) at ({90 - 360/12*(\i-1)}:\r) {$\i$};
}

% 4-chain: each node connected to its 4 nearest neighbors (±1, ±2)
\foreach \i in {1,...,12} {
    % neighbor i+1 (mod 12)
    \pgfmathtruncatemacro{\jone}{mod(\i,12) + 1}
    % neighbor i+2 (mod 12)
    \pgfmathtruncatemacro{\jtwo}{mod(\i+1,12) + 1}

    \draw[chainEdge] (b\i) -- (b\jone);
    \draw[chainEdge] (b\i) -- (b\jtwo);
}
\node[draw=none, font=\small] at (0, -2.4) {(f) $K$-chain};
\end{scope}

\begin{scope}[xshift=9.7cm, yshift = -8cm]
\def\r{1.5}
% Degree-sorted circular placement:
% positions:  1   2   3    4    5   6   7   8   9   10   11   12
% nodes:     c6, c3, c12, c10, c1, c7, c4, c5, c8,  c2,  c11, c9

\node[rightNode] (c6)  at ({90-30*0}:\r) {1};
\node[rightNode] (c3)  at ({90-30*1}:\r) {2};
\node[rightNode] (c12) at ({90-30*2}:\r) {3};
\node[rightNode] (c10) at ({90-30*3}:\r) {4};
\node[rightNode] (c1)  at ({90-30*4}:\r) {5};
\node[rightNode] (c7)  at ({90-30*5}:\r) {6};
\node[rightNode] (c4)  at ({90-30*6}:\r) {7};
\node[rightNode] (c5)  at ({90-30*7}:\r) {8};
\node[rightNode] (c8)  at ({90-30*8}:\r) {9};
\node[rightNode] (c2)  at ({90-30*9}:\r) {10};
\node[rightNode] (c11) at ({90-30*10}:\r) {11};
\node[rightNode] (c9)  at ({90-30*11}:\r) {12};

% Edges from adjacency matrix (i<j)
\draw[chainEdge] (c1) -- (c3);
\draw[chainEdge] (c1) -- (c6);
\draw[chainEdge] (c1) -- (c10);
\draw[chainEdge] (c1) -- (c12);

\draw[chainEdge] (c2) -- (c6);
\draw[chainEdge] (c2) -- (c12);

\draw[chainEdge] (c3) -- (c4);
\draw[chainEdge] (c3) -- (c6);
\draw[chainEdge] (c3) -- (c7);
\draw[chainEdge] (c3) -- (c10);
\draw[chainEdge] (c3) -- (c11);
\draw[chainEdge] (c3) -- (c12);

\draw[chainEdge] (c4) -- (c6);
\draw[chainEdge] (c4) -- (c10);

\draw[chainEdge] (c5) -- (c6);
\draw[chainEdge] (c5) -- (c7);
\draw[chainEdge] (c5) -- (c10);

\draw[chainEdge] (c6) -- (c7);
\draw[chainEdge] (c6) -- (c8);
\draw[chainEdge] (c6) -- (c12);

\draw[chainEdge] (c7) -- (c8);

\draw[chainEdge] (c8) -- (c12);

\draw[chainEdge] (c9) -- (c12);

\draw[chainEdge] (c10) -- (c11);
\node[draw=none, font=\small] at (0, -2.4) {(g) \ER};
\end{scope}

\begin{scope}[xshift=13.8cm, yshift = -8cm]
    \def\r{1.5}

% 12 nodes on a circle, d1 at top, then clockwise
% Node placement with d10 and d12 swapped:
\node[rightNode] (d11)  at ({ 90}:\r) {1};
\node[rightNode] (d10)  at ({ 60}:\r) {2};
\node[rightNode] (d4)  at ({ 30}:\r) {3};
\node[rightNode] (d9)  at ({  0}:\r) {4};
\node[rightNode] (d5)  at ({-30}:\r) {5};
\node[rightNode] (d6)  at ({-60}:\r) {6};
\node[rightNode] (d2)  at ({-90}:\r) {7};
\node[rightNode] (d8)  at ({-120}:\r) {8};
\node[rightNode] (d1)  at ({-150}:\r) {9};
\node[rightNode] (d3) at ({-180}:\r) {10};  
\node[rightNode] (d7) at ({-210}:\r) {11};
\node[rightNode] (d12) at ({-240}:\r) {12}; 

% Row 1
\draw[chainEdge] (d1) -- (d4);
\draw[chainEdge] (d1) -- (d5);
\draw[chainEdge] (d1) -- (d7);
\draw[chainEdge] (d1) -- (d11);

% Row 2
\draw[chainEdge] (d2) -- (d3);
\draw[chainEdge] (d2) -- (d5);
\draw[chainEdge] (d2) -- (d6);
\draw[chainEdge] (d2) -- (d8);

% Row 3
\draw[chainEdge] (d3) -- (d6);
\draw[chainEdge] (d3) -- (d9);
\draw[chainEdge] (d3) -- (d11);

% Row 4
\draw[chainEdge] (d4) -- (d6);
\draw[chainEdge] (d4) -- (d9);
\draw[chainEdge] (d4) -- (d12);

% Row 5
\draw[chainEdge] (d5) -- (d7);
\draw[chainEdge] (d5) -- (d12);

% Row 6
\draw[chainEdge] (d6) -- (d9);

% Row 7
\draw[chainEdge] (d7) -- (d8);
\draw[chainEdge] (d7) -- (d10);

% Row 8
\draw[chainEdge] (d8) -- (d9);
\draw[chainEdge] (d8) -- (d10);

% Row 9
% (d9,d3), (d9,d4), (d9,d6), (d9,d8) already drawn

% Row 10
\draw[chainEdge] (d10) -- (d11);
\draw[chainEdge] (d10) -- (d12);

% Row 11
\draw[chainEdge] (d11) -- (d12);

\node[draw=none, font=\small] at (0, -2.4) {(h) Random regular};
\end{scope}
\end{tikzpicture}
\caption{Comparison of graph designs.  For random graphs, we show one realization drawn from the corresponding distribution, with nodes ordered by degree. Panels (a)--(d) are bipartite graphs of size $8\times 8$, where the long chain is the deterministic design studied in \cite{jordan1995principles}, and the probabilistic expander is the random design introduced by \cite{chen2015optimal}. Panels (e)--(h) are unipartite graphs on $12$ nodes. Graphs (c), (d), and (f)--(h) all have average degree $4$. %Here the bipartite graphs are of size $8\times 8$, and the unipartite graphs each have $12$ nodes.
}
\label{fig:graph}
\end{figure}
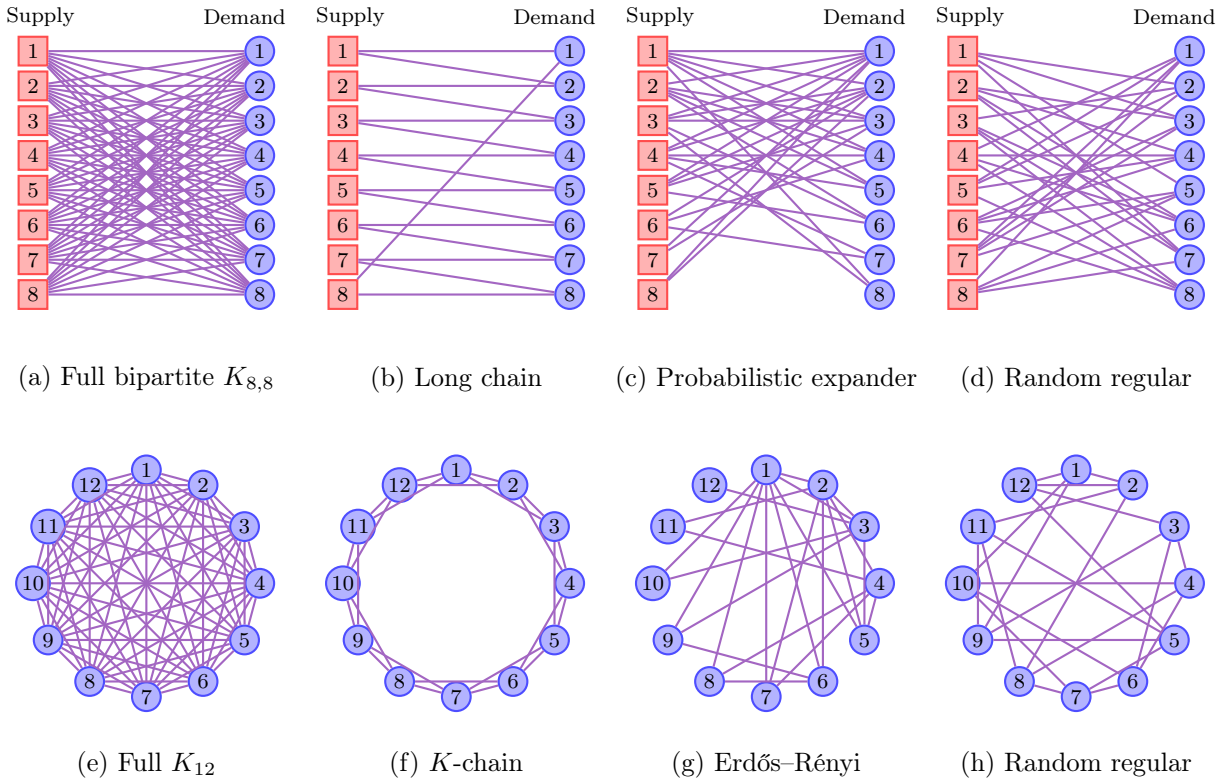

Figure~\ref{fig:graph} illustrates these complementary principles. The probabilistic expander in panel (c) and the \ER graph in panel (g) randomize their edges but have unequal node degrees. The long chain in panel (b) and the $K$-chain in panel (f) are regular but concentrate their connections locally. The random regular designs in panels (d) and (h) distribute edges evenly across nodes while dispersing connections throughout the network.

\subsection{\blue{Key Proof Ideas}}

Degree regularity and dispersed edge placement motivate random regular designs, but these features do not by themselves establish optimal sparsity. The central challenge is to show that, beyond the unavoidable loss from isolated nodes, larger residual-graph bottlenecks contribute only lower-order loss. 
We first describe a common connectivity-based proof strategy that underlies the Bernoulli-demand analysis for process flexibility and the constant-loss regime for middle-mile transportation. We then discuss two parts of our analysis that fall outside this common strategy.

\paragraph{Connectivity Properties as a Common Proof Strategy.}
We first consider the Bernoulli demand instances. 
In both settings, achieving the target loss is equivalent to the residual graph satisfying certain connectivity conditions. In process flexibility, the standard max-flow min-cut characterization requires every subset of supply nodes to have sufficiently many neighbors among the surviving demand nodes. In middle-mile transportation, the Tutte--Berge formula requires that removing any subset of nodes not fragment the residual graph into too many connected components. We call any subset of nodes that violates the corresponding condition a \emph{bottleneck}.

The analyses in the two settings (underlying Propositions~\ref{prop:flex-const} and~\ref{prop:perfect matching}) follow the same high-level strategy. First, we identify structural properties of the residual graph that are sufficient to rule out bottlenecks of any size. Second, we show that random $d$-regular graphs satisfy these properties with high probability when the degree is just above the lower-bound threshold. 
The key properties for the two settings, formalized respectively in Lemmas \ref{lmm:epsilon_connectivity} and \ref{lmm:match-prop}, also exhibit close parallels. For example, in process flexibility, every supply set with size exceeding $\epsilon n$ must have sufficiently many edge connections to the surviving demand nodes (see \eqref{eq:no_waste}). 
In transportation, the residual graph must contain only a small number of connected components of small sizes (see \eqref{eq:no_small_comp}). Both properties can be viewed as natural extensions of the basic requirement that the residual graph contain few isolated nodes, from individual nodes to larger subsets of nodes.

\paragraph{General Bounded Demand Distributions for Process Flexibility.} For process flexibility, the preceding connectivity analysis establishes the desired guarantee for Bernoulli-type demand vectors with support size $np$  (see Proposition~\ref{prop:flex-const}). We then extend this guarantee to every demand distribution supported on $[0,1/p]^n$. On each slice with fixed total demand, an appropriate excess-loss function is convex and symmetric in the demand coordinates. Its maximum is therefore attained at a Bernoulli-type extreme point whose coordinates lie in $\{0,1/p\}$ except possibly one fractional coordinate. Monotonicity properties then allow us to compare every such extreme point with a Bernoulli-type demand vector of support size $np$ covered by Proposition~\ref{prop:flex-const}, yielding the desired performance guarantees for every bounded demand vector.

\paragraph{$\epsilon$-Fractional Loss for Middle-Mile Transportation.}
For the fractional-loss regime in middle-mile transportation, the degree remains constant as the system grows for any fixed loss target $\epsilon$. We hence exploit the locally tree-like structure of sparse random regular graphs. In particular, the residual graphs converge locally to an appropriate Galton--Watson tree. This allows us to apply the matching characterization of \citet{bordenave2013matchings} to compute the limiting fraction of unmatched nodes and show that the resulting degree threshold matches the isolated-node lower bound to leading order.

\section{Process Flexibility Design}\label{sec:proc-flex-main}
In this section, we establish asymptotically sharp degree thresholds for process flexibility design. We first derive necessary degree lower bounds under the scaled Bernoulli demand model $\mathcal D_p$, and then show that random $d$-regular bipartite graphs attain the resulting thresholds for the broader class of $(1/p)$-bounded demand distributions.

\subsection{Lower Bound for Any Design}\label{sect:pf-lb}

The following theorem provides a lower bound on the expected loss of an arbitrary graph design under $\mathcal D_p$ in terms of its average degree.

\begin{theorem}\label{thm:lower_bound_bp} Suppose $D \sim \calD_p$ and recall $q = 1-p$.
Then the expected loss of any bipartite graph $G$, with $n$ nodes in each partition and average degree $d$,\footnote{Throughout, we suppress floor and ceiling notation: when a quantity represents a cardinality or degree, it is understood to be rounded to an appropriate integer. Such rounding affects our bounds only by lower-order terms and does not alter asymptotic conclusions.} is lower bounded by 
\$L(G) \ge  \prob{\Binom(n-d, p) \ge np} \cdot nq^{d} \ge \left(\frac{1}{2}-\frac{d p}{\sqrt{2\pi(n-d)pq}}  - \frac{c(
1-2pq)}{\sqrt{(n-d)pq}}\right) nq^{d},
\$ 
where $c>0$ is a fixed constant. 
\end{theorem}

\blue{The lower bound in \prettyref{thm:lower_bound_bp} has a natural interpretation in terms of isolated supply nodes. The probability term captures the event that the full-flexibility benchmark can utilize all \(n\) units of supply, so that isolated supply nodes directly contribute to the loss. A supply node \(i\) with degree \(d_i\) becomes isolated when all of its neighbors have zero realized demand, which occurs with probability \(q^{d_i}\). The expected number of isolated supply nodes is therefore \(\sum_{i=1}^n q^{d_i}\), which, by Jensen's inequality, is at least \(nq^d\). Thus, the lower bound can be understood as the loss caused by isolated supply nodes when the fully flexible system is fully utilized.}

\prettyref{thm:lower_bound_bp} provides a direct lower bound on the minimum average degree required to achieve $\epsilon$-fractional loss or constant expected loss. For sufficiently large $n$ and $d\ll \sqrt{n}$, the theorem implies that $L(G) \ge nq^d/4$ for any bipartite graph $G$. Thus, for a graph to achieve an $\epsilon$-fractional loss for some $\epsilon >0$, \ie, $L(G) \le  \epsilon n$, its average degree must satisfy 
\#\label{eq:lb-degree-flex} 
d \geq \frac{\log(1/\epsilon)}{\log(1/q)} - \frac{\log4}{\log(1/q)}.
\#
%\nb{Another strange equation label, (2) is already used earlier.}
Similarly, for a graph to achieve a constant loss, \ie, $L(G) \le \ell$ for some $\ell >0$, it must satisfy  
\#\label{eq:lb-n-flex} 
d \geq \frac{\log n}{\log(1/q)} - \frac{\log(4\ell)}{\log(1/q)}.
\#
These results correspond to the lower bound plotted in Figure \ref{fig:two_simple}.\footnote{We consider small $\epsilon$ and large $n$, so the constant variations $\log4/\log(1/q)$ and $\log(4\ell)/\log(1/q)$ are negligible.}

The isolated-node argument also highlights the importance of degree regularity. By Jensen's inequality,
$\sum_{i=1}^n q^{d_i}\ge nq^d$
with equality only when $d_i\equiv d$. Thus, among graphs with the same average degree, degree regularity minimizes the expected number of isolated supplies, a fundamental source of loss.
However, degree regularity alone is not sufficient to attain the lower-bound degree threshold. For instance, the $K$-chain is regular yet still requires $K = \Omega(1/\epsilon)$ for $\epsilon$-fractional loss \citep{chen2015optimal} and thus remains far from optimal. 
Next, we show that the random regular graph design asymptotically achieves the minimum required degree.

\subsection{Optimality of Random $d$-Regular Bipartite Graphs}\label{sec:flex-design-general}

In this section, we formally define the random $d$-regular bipartite graph and establish its optimality for the process flexibility problem. Our result applies not
only to the i.i.d.\ scaled Bernoulli demand model considered in Section~\ref{sect:pf-lb},
but to a substantially broader class of bounded demand distributions, allowing
the demands at different nodes to be non-identically distributed and arbitrarily
correlated.

We begin with the standard configuration-model construction \cite[Chapter 9]{janson2011random}.

\begin{definition}[Bipartite configuration model]
A random $d$-regular bipartite graph with $n$ nodes in each partition, denoted by $\calG(n,n,d)$, is generated as follows. Attach $d$ half-edges to each vertex in either partition to create a total of $nd$ half-edges in either partition. Then, pair the $nd$ half-edges incident to the left vertices with the $nd$ half-edges incident to the right vertices uniformly at random, and merge each pair into a single edge.  
\end{definition}
A graph generated by the configuration model may contain multi-edges, i.e., multiple edges between the same pair of nodes. Our analysis studies graphs that may contain such multi-edges, while noting that replacing them with a single edge does not impact the performance measure for process flexibility. Thus, the results in our theorems also hold for the corresponding simple graphs.
Next, we define a general class of demand distributions with bounded support considered in our analysis.
{\color{blue}
\begin{definition}[Bounded demand distribution]\label{def:general_demand}
Fix $p\in(0,1)$.
A distribution $\cD$ on $\R^n$ is called $(1/p)$-bounded if, for a random vector $D = (D_1, \ldots, D_n) \sim \cD$, each coordinate $D_i$ is supported on $[0, 1/p]$. Let $\mathfrak{B}(p)$ denote the set of $(1/p)$-bounded demand distributions on $\R^n$.
\end{definition}
This class imposes only a coordinate-wise upper bound on demand. In particular,
the demand coordinates need not be identically distributed or independent. The
i.i.d.\ scaled Bernoulli distribution $\mathcal D_p$ is one member of this class. Our main result shows that the degree thresholds identified by
the scaled Bernoulli lower bound are sufficient for every $(1/p)$-bounded demand
distribution.

\begin{theorem}\label{thm:loss-general}
Fix  $p\in(0,1)$ and let $q=1-p$. 
\begin{enumerate}
\renewcommand{\labelenumi}{\textnormal{(\alph{enumi})}}
\renewcommand{\theenumi}{\thetheorem(\alph{enumi})}
\item\textbf{$\epsilon$-fractional loss.}\label{thm:frac-loss-general}\label{thm:flex-const}
There exist constants $\epsilon_0,c_1>0$ such that for every $0<\epsilon<\epsilon_0$, all sufficiently large $n$, %every $(1/p)$-bounded distribution $\mathcal D$ on $\mathbb R^n$, 
and every
integer $d$ satisfying
\# \label{eq:degree-const-require}
d \ge \frac{\log(1/\epsilon) + c_1\log\log(1/\epsilon)}{\log(1/q)},
\#
a random graph $G\sim\mathcal G(n,n,d)$ satisfies
\[
\probs{L_{\cD}(G)\le\epsilon n}{G} \ge 1-1/n , \qquad \forall \cD\in \mathfrak{B}(p).
\]
%with probability at least $1-1/n$ over the randomness of $G$. \nbr{I feel that the current statement is a bit confusing. Are you taking the uniform statement within probability??? Maybe it is clearer to write out the probability event.}
\item\textbf{Constant loss.}\label{thm:constant-loss-general}\label{thm:flex-1-n}
For every fixed $\delta>0$ and $\ell\ge 64/\delta+2/p$, all sufficiently large $n$, %every $(1/p)$-bounded distribution $\mathcal D$ on $\mathbb R^n$, \nbr{Can we move the common things to the beginning of the theorem statement, so we do not need to repeat???} 
and every integer $d$ satisfying
\#\label{eq:degree-d-1-n}
d \geq \frac{(1+\delta)\log n}{\log(1/q)},
\#
a random graph $G\sim\mathcal G(n,n,d)$ satisfies
\[
\probs{L_{\cD}(G)\le\ell}{G} \ge 1-1/n, \qquad \forall \cD\in \mathfrak{B}(p).
\]
% with probability at least $1-1/n$ over the randomness of $G$.
\end{enumerate}
\end{theorem}

Theorem~\ref{thm:loss-general} shows that random regular graphs are asymptotically minimax-optimal over the class of $(1/p)$-bounded demand distributions. Indeed, up to lower-order terms, the two degree thresholds match the lower bounds \eqref{eq:lb-degree-flex} and \eqref{eq:lb-n-flex}, which are forced by the i.i.d.\ scaled Bernoulli distribution $\cD_p$ within this class. 
Thus, our result indicates that
$\cD_p$ serves as the hard case for the broader class of bounded demand distributions.  
Compared with the more elaborate flow-decomposition argument of \cite{chen2015optimal}, which gives order-wise optimality for i.i.d.\ bounded demands, our simpler convexity reduction in Section \ref{sec:flex-bound} identifies the minimax threshold with a sharp constant over all demand distributions supported on $[0,1/p]^n$, including those with
nonidentically distributed and arbitrarily correlated demand coordinates. Finally, the guarantees hold with high probability over the random graph construction, implying that a single graph drawn from the configuration model typically performs well. If a stronger practical guarantee is desired, one may draw several independent graphs, evaluate them via Monte Carlo simulation, and select the best-performing one.}
%{\color{red} YW: Just want to note that we use notation $L(G)$ for lower bound but $\E_{D\sim\cD}[L(G,D)]$ for upper bound.}

{\color{blue} The proof of Theorem \ref{thm:loss-general} proceeds in two stages. We first establish a performance guarantee for the Bernoulli-type demand vector in Section \ref{sec:flex-dp}, by identifying the connectivity properties required for small loss. We then extend this guarantee to arbitrary $(1/p)$-bounded demand vectors in Section \ref{sec:flex-bound}. The extension uses convexity and permutation symmetry to reduce the analysis on each fixed-total-demand slice to Bernoulli-type extreme points.}
%\nbr{JX. Shall we move this paragraph before Section 3.3???}
%{\color{blue}YW: I would clarify that we are studying deterministic demand realization of $\cD_p$ with randomized construction in this section.}

\subsection{Key Connectivity Properties}\label{sec:flex-dp}

This subsection develops the first stage of the proof of \prettyref{thm:loss-general}. Throughout, we fix a Bernoulli demand vector $D$ drawn from $\cD_p$, so the only randomness arises from the graph construction. Let $\supp(D)= \{j\in \nd\colon D_j =1/p\}$ denote its set of active demand nodes. The next proposition, which provides the main ingredient of our analysis, shows that when $|\supp(D)|=np$,\footnote{
For simplicity, we assume that \(np\) is an integer. The general case follows by replacing \(np\) with \(\lfloor np\rfloor\): the corresponding demand vector then differs in at most one coordinate, so the resulting upper bound on the loss increases by at most \(1/p\) (see \prettyref{lem:mono-frac}).} a random $d$-regular graph serves at least $(1-\epsilon)n$ units of demand with high probability at the desired degree threshold. 

\begin{proposition}\label{prop:flex-const}
There exist  constants $\epsilon_0,c_1 >0$ such that for any $0<\epsilon < \epsilon_0$, %\nbr{can we pin down $\epsilon_0?$} 
any $\delta\ge c_1\log\log(1/\epsilon)/\log(1/\epsilon)$, and any demand realization $D\sim\cD_p$ with $|\supp(D)| = np$, a random $d$-regular graph $G\sim \cG(n,n,d)$ with degree  
\#\label{eq:ddddelta}
d \ge \frac{(1+\delta)\log(1/\epsilon)}{\log(1/q)}
\#
satisfies $\probs{z(G,D) \ge (1-\epsilon) n}{G} > 1- c_n(\epsilon)$ with $c_n(\epsilon) \triangleq 3n\exp\left(-\delta n\epsilon\log(1/\epsilon)/8 \right)$.
\end{proposition} 
{\color{blue} 
We prove Proposition~\ref{prop:flex-const} by identifying three key connectivity properties that collectively rule out all potential bottlenecks. Our main technical novelty lies in uncovering a set of connectivity properties that both hold for random regular graphs and provide the key tools for establishing the optimality of these graph designs.}
%to provide properties of random $d$-regular graphs. %\blue{[Maybe how we prove the conditions is more important? Like novel analysis to deal with bottleneck cases, especially for small demand set. Do we want to put more details like their union bound over the demand subsets? [Tho maybe it's too much and hard to understand.]]}
%In addition to \prettyref{thm:flex-const}, \prettyref{prop:flex-const} is also used in the proof of \prettyref{thm:flex-1-n} as we can take $\epsilon$ to be $1/n$ to bound the loss when demand is sufficiently high. However, when we take $\epsilon$ to be $1/n$, then the probability of $|\supp(D)| \ge (1-\epsilon/2)np$ no longer converges to 1 as $n$ becomes large. Thus, we also need to bound the loss when $|\supp(D)|$ is low, which we formally do in Section XX.
% Key technical lemma of the proposition
%Turning our attention back to establishing \prettyref{prop:flex-const}, we present the following key technical lemma, which establishes three connectivity conditions for random $d$-regular networks. 
For node sets $U$ and $W$, let $e(U,W)$ be the number of edges between them (counting edge multiplicities).  
\begin{lemma}\label{lmm:epsilon_connectivity}
Under the same conditions as in \prettyref{prop:flex-const}, let $V = \supp(D)$ with $|V| = np$. Then there exist constants $\gamma>0$ and $c\in(0,1)$ such that
\begin{align}
&\probs{\exists S \subset [n] \text{ with } |S| > \epsilon^c n\colon |N(S)\cap V| < p(|S|- \epsilon n)}{G} \le  \zeta, %n\exp\big(- {\tfrac{p}{4\log(1/q)}n\epsilon}\log^2(1/{\epsilon})\big), 
\label{eq:flex-d-large-s-frac} \\
&\probs{\exists S \subset [n] \text{ with } |S| \ge  \epsilon n \colon e(S,V) \le \gamma |S|}{G} 
\le  \zeta,
\label{eq:no_waste} \\
&\probs{\exists S \subset [n], T\subset V \text{ with } \epsilon n \le |S| \le \epsilon^c n, |T| = p(|S|- \epsilon n) \colon e(S,T) > \gamma |S|}{G} \le  \zeta,
%& \qquad \le n \exp(- {4 n \epsilon} \log ({1}/{\epsilon})), 
\label{eq:flex-d-no-cong} 
\end{align}
where $\zeta \triangleq n\exp (- \delta n \epsilon\log(1/\epsilon)/8)$.
\end{lemma}

%explain the intuitions of the conditions stated in the lemma.
We defer the formal proof of \prettyref{lmm:epsilon_connectivity} to Section \ref{sect:app-process} and focus here on the key intuitions for why it is helpful. By the max-flow min-cut theorem, for any graph $G$, the total fulfilled demand, $z(G, D)$, can be written as
\#\label{eq:mfmc-zg-s}
z(G, D) = \min_{S\subset [n]} \left\{ n-|S| + |N(S)\cap \supp(D)|/p \right\},
\#
where $N(S)$ is the set of neighbors of a node set $S$ in $G$. 
For the full flexibility network $K_{n,n}$, 
\$
z(K_{n,n}, D) = \min\left\{ n, |\supp(D)|/p\right\}.
\$
Let $V = \supp(D)$. Because $|V| = np$, full flexibility serves $n$ units. Hence, $L(G,D) \le \epsilon n$ is equivalent to $z(G,D) \ge (1-\epsilon)n$. By the max-flow min-cut formulation in \eqref{eq:mfmc-zg-s}, this is equivalent to
\$
|N(S)\cap V| \ge p(|S|-\epsilon n), \qquad \forall S\subset [n].
\$ 
 As a result, any
$S$ for which $|N(S)\cap V| < p(|S|- \epsilon n)$ forms a \emph{bottleneck}, violating the desired condition.  
Figure \ref{fig:bottle} illustrates examples of both large and small bottlenecks.
The first connectivity condition~\eqref{eq:flex-d-large-s-frac} ensures that the probability of having a large bottleneck ($|S| > \epsilon^c n$) is exponentially small. %This condition follows by upper-bounding the probability that $e(S,T)=0$ for any demand subset $T$ of size $|V| - |S||V|/n$ and then applying a union bound. 
Note that if $S$ is a bottleneck, then there is a demand set $T\subset V$ with $|T| = |V| - p(|S|- \epsilon n)$
such that $e(S,T)=0$ (see Figure \ref{fig:bottle} (a)). %there are no edges between $S$ and $T$, 
In the proof, we establish \eqref{eq:flex-d-large-s-frac} by upper-bounding the probability of this event via a union bound over all possible choices of $S$ and $T$.
%that $e(S,T)=0$ for $S$ and $T$ with $|S| > \epsilon^c n$ and $|T| = |V| - p(|S|- \epsilon n)$. 

%[Because $S$ being a bottleneck is equivalent to the existence of a demand subset $T\subset V$ with $|T| = |V| - p(|S|- \epsilon n)$ such that there are no edges between $S$ and $T$, \eqref{eq:flex-d-large-s-frac} follows by upper-bounding the probability that $e(S,T)=0$ for $S$ and $T$ with $|S| > \epsilon^c n$ and $|T| = |V| - p(|S|- \epsilon n)$. and then applying a union bound. 
%One may think of $T$ as the complement of $N(S)$ in $V$. ]
%\blue{[maybe by using the degree regularity on the demand side and a direct union bound over all possible choices of supply subsets $S$ and demand subsets $N(S)\cap V$.]}
%\blue{[The main challenge and our novelty lie in the small-$S$ regime.]} 

%\nbr{Can we include plots, one for potential large bottleneck and one for potential small bottleneck? This might help reader to understand the explanations here in words}

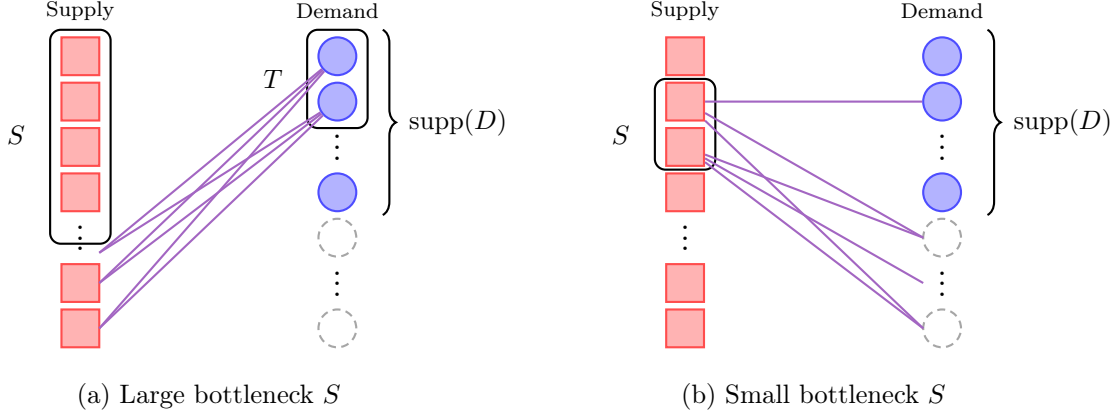
\begin{figure}[ht]
    \centering
\begin{tikzpicture}[
    every node/.style={inner sep=1pt, minimum size=5mm, font=\scriptsize,thick},
    supplyNode/.style={rectangle, draw=red!70, fill=red!30},
    demandNode/.style={circle, draw=blue!70, fill=blue!30},
    ghostNode/.style={circle, draw=gray!70, densely dashed},
    isolateNode/.style={rectangle, draw=red!70},
    isolateNode1/.style={circle, draw=blue!70},
    solidEdge/.style={thick, blue!60!red!60},
    ghostEdge/.style={thick, gray!70, densely dashed},
]
\begin{scope}[xshift=0cm, yshift = 0cm]
% Left set S
\node at (-2.55,-0.04) {\small $S$};

% Titles
  \node at (-1.7,1.6) {Supply};
  \node at (1.7,1.6)  {Demand};

  % Supply nodes (squares)
  \node[supplyNode] (s1) at (-1.7, 1.0) {};
  \node[supplyNode] (s2) at (-1.7, 0.4) {};
  \node[supplyNode] (s3) at (-1.7,-0.2) {};
  \node[supplyNode] (s4) at (-1.7,-0.8) {};
  \node[supplyNode] (s5) at (-1.7,-2.0) {};
  \node[supplyNode] (s6) at (-1.7,-2.6) {};

  % Vertical dots for more supply
  \node at (-1.7,-1.3) {\Large$\vdots$};

  % Demand nodes (solid circles)
  \node[demandNode] (d1) at ( 1.7, 1.0) {};
  \node[demandNode] (d2) at ( 1.7, 0.4) {};

  \node[demandNode] (d4) at ( 1.7, -0.8) {};

\node at (1.7,-0.1) {\Large$\vdots$};
\node at (1.7,-1.9) {\Large$\vdots$};
  % Ghost demand nodes (dashed circles)
  \node[ghostNode] (g1) at ( 1.7,-1.4) {};
\node[ghostNode] (g2) at ( 1.7, -2.6) {};

\draw[rounded corners,thick] (-2.1,1.35) rectangle (-1.3,-1.48);

% Right set T (dashed box)
\draw[rounded corners,thick] (1.3,1.35) rectangle (2.1,0.05);
\node at (0.85,0.7) {\small $T$};
\draw[decorate,thick,decoration={brace,amplitude=5pt}]
  (2.3,1.35) -- (2.3,-1.10)
  node[midway,xshift=1cm] {\small $\supp(D)$};

% Multiple edges from the bottom left square to the bottleneck node
\draw[solidEdge] (-1.45,-1.6)-- (d1);
\draw[solidEdge] (-1.45,-2.0)-- (d1);
\draw[solidEdge] (-1.45,-2.6)-- (d1);
\draw[solidEdge] (-1.45,-1.6)-- (d2);
\draw[solidEdge] (-1.45,-2.0)-- (d2);
\draw[solidEdge] (-1.45,-2.6)-- (d2);

% Caption
\node at (0.0, -3.5) {\small (a) Large bottleneck $S$};
\end{scope}

\begin{scope}[xshift=8cm, yshift = 0cm]
% Left set S
\node at (-2.55,-0.04) {\small $S$};
\draw[rounded corners,thick] (-2.1,0.7) rectangle (-1.3,-0.5);
% Titles
  \node at (-1.7,1.6) {Supply};
  \node at (1.7,1.6)  {Demand};

  % Supply nodes (squares)
  \node[supplyNode] (ss1) at (-1.7, 1.0) {};
  \node[supplyNode] (ss2) at (-1.7, 0.4) {};
  \node[supplyNode] (ss3) at (-1.7,-0.2) {};
  \node[supplyNode] (ss4) at (-1.7,-0.8) {};
  \node[supplyNode] (ss5) at (-1.7,-2.0) {};
  \node[supplyNode] (ss6) at (-1.7,-2.6) {};

  % Vertical dots for more supply
  \node at (-1.7,-1.3) {\Large$\vdots$};

 % Demand nodes (solid circles)
  \node[demandNode] (dd1) at ( 1.7, 1.0) {};
  \node[demandNode] (dd2) at ( 1.7, 0.4) {};

  \node[demandNode] (dd4) at ( 1.7, -0.8) {};

\node at (1.7,-0.1) {\Large$\vdots$};
\node at (1.7,-1.9) {\Large$\vdots$};
  % Ghost demand nodes (dashed circles)
  \node[ghostNode] (gg1) at ( 1.7,-1.4) {};
\node[ghostNode] (gg2) at ( 1.7, -2.6) {};

% Multiple edges from the bottom left square to the bottleneck node
\draw[solidEdge] (1.45,-1.4)-- (ss2);
\draw[solidEdge] (dd2)-- (ss2);
\draw[solidEdge] (1.45,-2.6)-- (ss2);
\draw[solidEdge] (1.45,-1.4)-- (ss3);
\draw[solidEdge] (1.45,-2.0)-- (ss3);
\draw[solidEdge] (1.45,-2.6)-- (ss3);
\draw[decorate,thick,decoration={brace,amplitude=5pt}]
  (2.3,1.35) -- (2.3,-1.10)
  node[midway,xshift=1cm] {\small $\supp(D)$};

% Caption
\node at (0.0, -3.5) {\small (b) Small bottleneck $S$};
\end{scope}

\end{tikzpicture}
\caption{Illustration of bottlenecks. The blue-colored nodes are demand nodes with positive demand; together they form $\supp(D)$. For clarity, panel (a) shows only edges incident to $T$, while panel (b) shows only edges incident to $S$. Panel (a) illustrates a large bottleneck $S$: there exists  $T\subset \supp(D)$ of a certain size whose $|T|d$ edges all connect to $S^c$.  
Such an event is highly unlikely, so large bottlenecks rarely occur.
Panel (b) shows a small bottleneck $S$, which has too few edges connecting to $\supp(D)$. However, since $\supp(D)$ contains a constant fraction of demand nodes (roughly $pn$), it is unlikely that most of the $|S|d$ edges of $S$ fall in $\supp(D)^c$, thereby preventing small bottlenecks as well.}
    \label{fig:bottle}
\end{figure}

Deriving a similar connectivity condition for small $S$, with $\epsilon n \le |S| \le \epsilon^c n$, is much more delicate.
The difficulty arises because when a small $S$ has few connections to $V$, the events $\{e(S, T) = 0\}$  are highly correlated across different demand sets $T\subset V$.
For instance, if $e(S,V)=0$, then $e(S,T)=0$ for all $T \subset V.$ This correlation leads to a severe overestimation of the probability if one directly applies a union bound over all $T$.
To address this challenge, we establish two complementary conditions for small $S$, stated as \eqref{eq:no_waste} and \eqref{eq:flex-d-no-cong} in \prettyref{lmm:epsilon_connectivity}. 
Specifically, \eqref{eq:no_waste} ensures that the probability of a small $S$ having too few edges to $V$ is low (see Figure \ref{fig:bottle} (b)), thereby preventing the ``waste'' of supply nodes, while \eqref{eq:flex-d-no-cong} bounds the probability of excessive edge concentration between any small $S$ and $T$, thus preventing ``congestion''. Together, these two properties guarantee that small bottlenecks occur with vanishing probability, as shown in the proof of \prettyref{prop:flex-const}.

In particular, the degree regularity of the random $d$-regular graph is crucial for proving \prettyref{lmm:epsilon_connectivity}. The regularity of demand nodes ensures that every demand subset receives sufficiently many edges from the supply side, which is essential for \eqref{eq:flex-d-large-s-frac}. Likewise, the regularity of supply nodes ensures that every supply subset has sufficiently many edges to the demand side, which is crucial for \eqref{eq:no_waste} and \eqref{eq:flex-d-no-cong}.
%The regular degree of supply nodes is essential in the proof of \eqref{eq:no_waste}, \blue{[Intuitions on why regularity is crucial, and why the sharp constant is sufficient?]} and the joint regularity of both supply and demand nodes facilitates the proof of \eqref{eq:flex-d-large-s-frac}.  
Specifically, \eqref{eq:no_waste} guarantees that the number of isolated supplies is at most $\epsilon n$. When $d$ is close to the threshold $\log(1/\epsilon)/\log(1/q)$, this bound can be achieved only by regular graphs.
%relies critically on the sharp condition on $d$.}
%This property generalizes the notion of isolated supplies and yields thesharp condition on $d$ given by \eqref{eq:d-const-prop}. \nb{YW:not clear what this sentence means.}

With the connectivity properties in \prettyref{lmm:epsilon_connectivity}, we now complete the proof of \prettyref{prop:flex-const}. %by applying the max-flow–min-cut theorem and showing that both large and small bottlenecks occur only with exponentially small probability.

\begin{proof}[Proof of Proposition \ref{prop:flex-const}]
Let $V = \supp(D) $ with $|V| = np$.
By the min-cut formulation of $z(G,D)$ in \eqref{eq:mfmc-zg-s}, the event of small loss $\{n - z(G,D) \le \epsilon n\}$ is equivalent to ensuring 
\#\label{eq:edge-expansion}
\forall S \subset [n]\colon |N(S)\cap V| \ge p(|S|- \epsilon n).
\#
If $|S| \le \epsilon n,$
then $|N(S)\cap V|\ge 0 \ge p(|S|- \epsilon n)$ trivially holds. Hence, it suffices to show that random $d$-regular graphs satisfy \eqref{eq:edge-expansion} for all $S$ with size $|S| > \epsilon n$ with high probability. We separately consider regimes where $\epsilon n< |S| \le \epsilon^c n$ and $|S| > \epsilon^c n$, for the constant $c< 1$ given by Lemma \ref{lmm:epsilon_connectivity}.

\paragraph{Small $S$: $\epsilon n< |S| \le \epsilon^c n$.}
\iffalse
For any supply subset $S$, the bad event that its demand neighborhood is too small, $\{|N(S)\cap V| < p(|S|- \epsilon n)\}$, implies one of two cases: either the total number of edges from $S$ to $V$ is small, or many edges edges go to $V$ but are concentrated on a small subset of $V$. For any set $U \subset [n]$ and $W \subset [n]$, let $e(U,W)$ denote the number of edges between $U$ and  $W$ (counting edge multiplicities). Formally, for any threshold $\gamma>0$, we have 
\$ 
\{|N(S)\cap V| < p(|S|- \epsilon n)\} & \subset \{ e(S,V) \le \gamma|S|\}  \\
& \qquad \cup \{\exists T\subset V \text{ with } |T| = p(|S|- \epsilon n) \colon e(S,T) > \gamma|S|\}.
\$
\fi
Fix the threshold $\gamma>0$ as given by Lemma \ref{lmm:epsilon_connectivity}.
For any $S$, consider the bad event that 
%its demand neighborhood is too small, 
$|N(S)\cap V| < p(|S|- \epsilon n)$. This event implies that there exists $T\subset V$ with $|T|=p(|S|- \epsilon n)$ and $e(S,T) = e(S,V)$. In turn, this implies either the total number of edges from $S$ to $V$ is small, \ie, $e(S,V) \le \gamma|S|$, or there exists a subset $T$ for which $e(S,T) = e(S,V)>\gamma |S|$. 
Formally,  we have 
\$ 
\{|N(S)\cap V| < p(|S|- \epsilon n)\} & \subset \{ e(S,V) \le \gamma|S|\}  \cup \{\exists T\subset V \text{ with } |T| = p(|S|- \epsilon n) \colon e(S,T) > \gamma|S|\}.
\$
Therefore, by first showing that no supply nodes are wasted via \eqref{eq:no_waste}, and then showing that no edge congestion occurs between small sets via \eqref{eq:flex-d-no-cong}, both with high probability, we have
\$ 
\probs{\exists S \subset [n] \text{ with } \epsilon n \le |S| \le \epsilon^c n: |N(S)\cap V| < p(|S|- \epsilon n)}{G} \le 2\zeta.
\$
\paragraph{Large $S$: $\epsilon^c n < |S| \le n$.} The probability that a large-$S$ bottleneck exists is no greater than $\zeta$, as given by 
\eqref{eq:flex-d-large-s-frac} in Lemma \ref{lmm:epsilon_connectivity}. Combining it with the above small-$S$ result, we conclude that 
\begin{equation*}
\probs{z(G,D) < (1-\epsilon) n}{G} = \probs{\exists S \subset [n]\colon |N(S)\cap V| < p(|S|-\epsilon n)}{G} 
\le 3\zeta. \qedhere
\end{equation*}
\end{proof}

\subsection{\blue{Proof of Theorem~\ref{thm:loss-general}}}
\label{sec:flex-bound}

We now extend the Bernoulli performance guarantee in Proposition~\ref{prop:flex-const} to arbitrary $(1/p)$-bounded demand distributions. The extension relies on two observations. First, given the total demand, an appropriate excess-loss function is convex and symmetric in the demand coordinates. Its maximum is thus attained at a Bernoulli-type extreme point. Second, monotonicity properties allow us to compare every such extreme point with demand vectors of support size $np$ considered in Proposition~\ref{prop:flex-const}.

To formalize the first observation, we fix a total demand and consider the fixed-total-demand slice: 
\$
\Sigma_t = \left\{x\in [0, 1/p]^n\colon \sum_{j\in[n]}x_j = t\right\}, \qquad \forall t\in[0, n/p].
\$ 
Let $k(t)\triangleq \left\lfloor tp\right\rfloor$ and $r(t)\triangleq t-k(t)/p\in[0,1/p)$. Let $e_j \in \R^n$ denote the $j$-th standard basis vector, and $\mathbf{1}_{[k]} \triangleq \sum_{j=1}^k e_j$ denote the vector whose first $k$ entries are one and whose remaining entries are zero.
Define the canonical extreme point of $\Sigma_t$ by
\#\label{eq:extreme-dp}
D^\star(t)\triangleq
\begin{cases}
\mathbf{1}_{[k(t)]}/p +r(t)e_{k(t)+1}, & k(t)<n,\\
\mathbf{1}_{[n]}/p, & k(t)=n.
\end{cases}
\#
Thus, $D^\star(t)$ has entries in $\{0,1/p\}$ except for at most one fractional coordinate. 
%{\color{red} yw: Transition to the next sentence is not clear.} Theorem~\ref{thm:loss-general} establishes a high-probability guarantee for the expected loss under an arbitrary demand distribution. To this end, 
Next, we work with the positive-part excess loss. For $a>0$, define
\begin{equation}
F_a(D)
\triangleq
\E_{G\sim\cG(n,n,d)}
\left[(L(G,D)-a)^+\right],
\qquad
u^+\triangleq\max\{u,0\}.
\label{eq:f_ta}
\end{equation}
The excess loss preserves the convexity used below and will subsequently allow us to obtain a tail bound through Markov's inequality. The following lemma formalizes the convexity reduction.

\begin{lemma}[Convexity reduction]
\label{lmm:lt-convex}
For every $t\in[0,n/p]$ and $a>0$, the function $F_a$ is convex on $\Sigma_t$. Moreover, for every $D\in\Sigma_t$,
\[
F_a(D)\le F_a(D^\star(t)).
\]
\end{lemma}

The convexity-reduction lemma hence reduces the analysis of every demand vector in \(\Sigma_t\) to the canonical extreme point \(D^\star(t)\). It remains to compare these extreme points with the Bernoulli demand vector with $|\supp(D)| = np$ covered in Proposition~\ref{prop:flex-const} using the monotonicity properties.

\begin{lemma}[Demand monotonicity]
\label{lem:monotonicity}
Fix a bipartite graph $G$ with $n$ unit-capacity supply nodes.
\renewcommand{\labelenumi}{\textnormal{(\alph{enumi})}}
\renewcommand{\theenumi}{\thelemma(\alph{enumi})}
\begin{enumerate}
\item\label{lem:mono-frac} For every $D\in\mathbb R^n$, $j\in[n]$, and $r\in[0,1/p]$, it holds that
$L(G,D+r e_j)\le L(G,D)+r$.

\item\label{lem:mono-down} If $D, D'\in\mathbb R^n$ satisfy $0\le D'\le D$ componentwise and $\sum_jD_j\le n$, then
$L(G,D')\le L(G,D)$.
\end{enumerate}
\end{lemma}

Proofs of Lemmas~\ref{lmm:lt-convex} and~\ref{lem:monotonicity} are deferred to Section~\ref{sec:general_demand_proofs}. Together, these lemmas allow us to reduce the general-demand analysis to the extreme point $D^\star(n)$ and prove \prettyref{thm:loss-general}.

\begin{proof}[Proof of \prettyref{thm:loss-general}]
%For simplicity, we assume that $np$ is an integer, so that $D^\star(n)={\mathbf 1_{[np]}}/p$. When \(np\) is not an integer, \(D^\star(n)\) has one additional fractional coordinate relative to \(\mathbf 1_{[\lfloor np\rfloor]}/p\). By Lemma~\ref{lem:monotonicity}(a), this coordinate increases the loss by at most \(1/p\), which can be absorbed into the final loss bound.
Recall that $D^\star(n)={\mathbf 1_{[np]}}/p$.
We first claim that
\[
L(G,D^\star(t))\le L(G, D^\star(n)),
\qquad \forall t\in[0,n/p].
\]
When \(t\le n\), we have \(D^\star(t)\le D^\star(n)\) componentwise, so the inequality follows from Lemma~\ref{lem:mono-down}. When \(t\ge n\), we have \(D^\star(t)\ge D^\star(n)\) componentwise and
$z(K_{n,n},D^\star(t))=z(K_{n,n}, D^\star(n))=n$.
Since \(z(G,\cdot)\) is nondecreasing, the inequality again follows.

Fix $\epsilon',\delta(\epsilon') >0$ as in \prettyref{prop:flex-const}. Then $L(G, D^\star(n))\le\epsilon' n$ with probability at least \(1-c_n(\epsilon')\).
On this event, \(L(G,D^\star(t))\le\epsilon' n\) simultaneously for every \(t\in[0,n/p]\). Therefore, for every \(a\ge\epsilon' n\),
\begin{equation}\label{eq:slice-bd-asd}
F_a(D^\star(t))\le n c_n(\epsilon'),
\qquad \forall t\in[0,n/p].
\end{equation}    
%{\color{red} Here $F_a$ truncates the loss below $a$: the excess loss is zero whenever $L(G,D)\le a$, allowing us to bound $F_a(D)$ directly using the small probability that $L(G,D) > a$.}
%Thus, Proposition~\ref{prop:flex-const} controls $D^\star(t)$ simultaneously for every total-demand level $t$. \nbr{JX. This last sentence may be unnecessary?}

%\nbr{JX. In the following proof, why we use $F_a$ is a bit unclear.}
We next average the bound in \eqref{eq:slice-bd-asd} over the random total demand and convert the resulting joint expectation into a high-probability guarantee over the graph.
Let $D \sim \mathcal{D}$ and $T= T (D)= \sum_j D_j$. For any $a> 0$, since $G$ and $D$ are independent, and $(L(G,D)-a)^+ \ge 0$, Tonelli's theorem gives
\#\label{eq:tadasafg}
\mathbb{E}_{G, D}\left[(L(G,D)-a)^+\right] &= \E_{D}\left[\E_{G}\left[(L(G,D)-a)^+\right]\right] \\ \notag
&= \E_{T}\left[ \E_D\left[F_{a}(D) \mid T\right] \right] \le \E_{T}\left[F_{a}(D^\star(T))\right],
\#
where the last equality follows from $F_{a}$ defined in \eqref{eq:f_ta} and the inequality follows from \prettyref{lmm:lt-convex}.  
Let $a = \epsilon' n$. 
By \eqref{eq:slice-bd-asd}, $F_a(D^\star(t)) \le nc_n(\epsilon')$ for all $t$. Hence, by \eqref{eq:tadasafg}, 
\[
\mathbb{E}_{G, D}\left[(L(G,D)-a)^+\right]
\le \E_{T}\left[F_{a}(D^\star(T))\right] \le
n  c_{n}(\epsilon').
\]
Since $\mathbb{E}_D[L(G,D)]\ge 2a$ implies $\mathbb{E}_D\left[(L(G,D)-a)^+\right]
\ge a$, it follows that 
\$
\probs{\E_D[L(G,D)] \ge 2a }{G}
&\le
\probs{\mathbb{E}_D\left[(L(G,D)-a)^+\right] \ge a }{G}
\\
& \le 
\mathbb{E}_{G, D}\left[(L(G,D)-a)^+\right]/a
\le 
c_{n}(\epsilon')/\epsilon'.
\$
We now choose $\epsilon'$ separately for the fractional-loss and constant-loss regimes.
\paragraph{$\epsilon$-fractional loss.} Setting $\epsilon' = \epsilon/4$ and noting that $c_{n}(\epsilon')/\epsilon' \le 1/n$ conclude \prettyref{thm:frac-loss-general}.

\paragraph{Constant loss.} Fix constants $\delta, \ell >0$ such that $\ell > 64/\delta$. Setting $\epsilon' = \ell/(2n)$ and noting that  $c_{n}(\epsilon')/\epsilon' = 6n^2\exp\left(-\delta \ell\log(2n/\ell)/16 \right)/\ell \le 1/n$ conclude \prettyref{thm:constant-loss-general}.
\end{proof}

%\yw{YW: \prettyref{thm:constant-loss-general} uses $l$, not $C$.}

\section{Middle-Mile Transportation Design}\label{sect:trans}
This section studies the middle-mile transportation problem. Given an input graph $G$, and under the Bernoulli demand $\cD_p$, each node with a full-truckload demand does not enter the matching problem and is independently deleted with probability $q = 1-p$. 
For a demand realization $D\sim\cD_p$, let $\supp(D)$ be the set of remaining nodes with half-truckload demand, and $G[\supp(D)]$
be the residual graph induced by these nodes. 
In this case, the maximum number of station pairs served jointly, $\mu(G,D)$ in \eqref{prob:transportation}, is simply the maximum matching size of $G[\supp(D)]$. Accordingly, we write $\mu(G,D) = \mu(G[\supp(D)])$, and the matching loss 
$L(G) = 2\E_{D}[\mu(K_n[\supp(D)]) - \mu(G[\supp(D)])]$. 

\subsection{Lower Bound for Any Design}
The theorem below establishes a lower bound on the expected loss for any graph design $G$ as a function of its average degree $d$. Compared with the existing bound in~\citet[Theorem~1]{feng2024designing}, which shows that $L(G) \geq (npq^d - 1)^+$, our result provides a slight improvement.
\begin{theorem}\label{thm:lower_bound}
Suppose that $D\sim \cD_p$ and recall that $q=1-p$.
For any graph $G$ with $n$ nodes and average degree $d$, the expected matching loss is lower bounded by
     %$L(G) \geq npq^d.$
     $$
     L(G) \geq \max\left\{ npq^d - \frac{1}{2} 
     \left(1-(1-2p)^n\right), \, \frac{1}{2} np q^d  \left( 1 - (1-2p)^{n-d-1}\right) \right\}.
     $$
     %$L(G)\geq(npq^d-1)^+$.
\end{theorem}

\prettyref{thm:lower_bound} implies that for a graph to achieve an $\epsilon$-fractional loss for some $\epsilon >0$, \ie, $L(G) \le \epsilon n$, its average degree must be $d \geq \log(1/\epsilon)/{\log(1/q)} - O(1)$.
Similarly, achieving a constant loss, \ie, $L(G) = O(1)$, 
requires $d \geq \log n/{\log(1/q)} -O(1)$. 
We derive the lower bound by computing the expected number of \emph{isolated} nodes, since any isolated node cannot be matched and therefore contributes to the matching loss. Similar to the flexibility design setting, a node is isolated with probability $pq^{d_i}$. Applying Jensen's inequality yields a bound $npq^d$, with equality only when $d_i \equiv d$. Thus, any graph design achieving the minimum required average degree must be regular. We defer the formal proof of \prettyref{thm:lower_bound} to Section \ref{sect:app-transportation}.

%In fact, as shown later, the random $d$-regular graph achieves this minimum degree requirement. 

\subsection{Optimality of Random $d$-Regular Graphs}
We now show that the random $d$-regular graph achieves the optimal degree requirement, beginning with its formal definition via the standard configuration model~\cite[Chapter 9]{janson2011random}.

\begin{definition}[Configuration Model]
A random $d$-regular graph, denoted by $\calG(n,d)$, is generated as follows.  
Given a vertex set $[n]$, attach $d$ half-edges to each vertex, where $nd$ is assumed to be even. Then, pair all $nd$ half-edges uniformly at random, and merge each pair into a single edge.  
\end{definition}
The resulting random regular graph may contain self-loops and multi-edges. These do not affect the matching performance, so the results in our theorems still hold for the corresponding simple graphs obtained by removing them. 
\blue{The following theorem shows that random $d$-regular graphs attain the optimal degree thresholds in both the $\epsilon$-fractional-loss and constant-loss regimes.}
{\color{blue}
\begin{theorem}
\label{thm:matching-loss}
Suppose $D\sim\cD_p$ and recall that $q=1-p$.
\begin{enumerate}
\renewcommand{\labelenumi}{\textnormal{(\alph{enumi})}}
\renewcommand{\theenumi}{\thetheorem(\alph{enumi})}

\item \textbf{$\epsilon$-fractional loss.}
\label{thm:constant d}
There exist constants $\epsilon_0,c>0$ such that, for every constant $0<\epsilon<\epsilon_0$ and every integer $d$ satisfying
\begin{equation}
d
\ge
\frac{
\log(1/\epsilon)
+
c\epsilon\log^2(1/\epsilon)
}{
\log(1/q)
},
\label{eq:match-const-d}
\end{equation}
a random $d$-regular graph $G\sim\cG(n,d)$ satisfies
\[
\lim_{n\to\infty}
\probs{L(G)\le\epsilon n}{G}
=
1.
\]

\item \textbf{Constant loss.}
\label{thm:match-1-n}
For every fixed constant $\delta>0$ and all sufficiently large $n$, if
\begin{equation}
d
\ge
\frac{(1+\delta)\log n}{\log(1/q)},
\label{eq:cond_d}
\end{equation}
then a random $d$-regular graph $G\sim\cG(n,d)$ satisfies
$
\E_G[L(G)]=
O\big(n^{-\delta/8}\big)$.
Consequently,
\[
\probs{
L(G)\le n^{-\delta/16}
}{G}
\ge
1-O\left(n^{-\delta/16}\right).
\]
\end{enumerate}
\end{theorem}
}

Theorem~\ref{thm:matching-loss} shows that random $d$-regular graphs attain the optimal degree requirement in both loss regimes. For $\epsilon$-fractional loss, the degree requirement in Theorem~\ref{thm:constant d} matches the lower bound $\log(1/\epsilon)/\log(1/q)$ up to a lower-order term. For constant loss, Theorem~\ref{thm:match-1-n} shows that a degree just above $\log n/\log(1/q)$ yields not merely a constant loss, but vanishing loss with high probability.

\begin{remark}
For comparison, we also provide tight characterizations of the matching loss for \ER random graphs and for $K$-chains. Existing work ~\citep[Theorem~4]{feng2024designing} shows that \ER graphs require an average degree $d\ge (1+\delta)\log n/(1-q)$ to achieve a constant loss. For the $\epsilon$-fractional loss, we confirm their conjecture that achieving an $\epsilon$-fractional loss  requires $d\ge (1+o_\epsilon(1))\log (1/\epsilon)/(1-q)$ (see \prettyref{cor:er-const}). 
These requirements are strictly larger than the optimal ones. This gap stems from the degree fluctuations in \ER graphs, which force a larger average degree even to keep the number of isolated nodes below the target loss level.

For $K$-chains, $C_{n,K}$, existing work ~\citep[Corollary~1]{feng2024designing} establishes the upper bound $L(C_{n,K}) \le 4nq^{d/4}/d$. We improve this result and derive a nearly matching lower bound, obtaining the tight order $L(C_{n,K}) = \Theta(nq^{d/2})$ (\prettyref{thm:KNN}). Consequently, achieving an $\epsilon$-fractional loss requires a $K$-chain to have an average degree  $2\log(1/\epsilon)/\log(1/q)$, which is twice the optimal degree. Interestingly, the suboptimality gap for $K$-chains  in this setting is much smaller than that in the process flexibility setting, highlighting a fundamental difference between the two problems.
\end{remark}

%Although $K$-chains are regular graphs, they still require twice the optimal degree. 
%Interestingly, the pronounced contrast between the performance of $K$-chains in process flexibility design (where $d\ge \Omega(1/\epsilon)$ is necessary, as shown in \cite{chen2015optimal}) and in the middle-mile transportation problem (where our results show that $d\ge 2\log(1/\epsilon)/\log(1/q)$) also highlights a fundamental difference between the two problem settings. 

%\nbr{Also, mention the ER random graph result here???}
\subsection{Proof of \prettyref{thm:matching-loss} }
This section proves the main results. We begin with the constant loss, whose proof follows a structure similar to the analysis for process flexibility. 
%We first establish several key connectivity properties of random $d$-regular graphs, and then complete the proof via the classic Tutte's theorem. 
For the fractional loss, we apply a celebrated result that directly characterizes the maximum matching size of sparse graphs with locally tree-like structure. 
%In this part, we verify the local weak convergence of random regular graphs with constant degree and derive the corresponding limiting matching size.
\subsubsection{Proof of Theorem \ref{thm:match-1-n}}
\blue{In this section, we prove \prettyref{thm:match-1-n} for the constant-loss regime. The proof is guided by two possible sources of matching loss in the residual graph. The first is isolated vertices, which are an unavoidable source of loss. A vertex is isolated in the residual graph if it is retained while all $d$ of its neighbors are deleted, which occurs with probability $pq^d$. Hence, the expected number of isolated residual nodes is $npq^d$. Under the degree condition in \prettyref{eq:cond_d}, this is at most $pn^{-\delta}$ and therefore vanishes asymptotically.
However, the absence of isolated nodes alone does not guarantee a near-perfect matching, as larger structural bottlenecks may still exist. The main challenge is therefore to show that, after removing the isolated vertices, the residual graph has no additional matching bottlenecks and admits a matching covering all remaining non-isolated vertices except possibly one due to parity. The following proposition formalizes these two steps.}

\begin{proposition}\label{prop:perfect matching}
Fix any sufficiently large $n$.
For any constant $\delta>0$, there exists a constant $\eta>0$ such that 
for any demand realization $D\sim\cD_p$ with 
$|\supp(D)| \ge (1-\eta)np$, the random $d$-regular graph 
$G\sim\cG(n,d)$ with $d$ given by~\prettyref{eq:cond_d} satisfies the following.
Let $I$ denote the set of isolated vertices in $G[\supp(D)]$. Then
\begin{align}
\probs{|I| \ge 1}{G}
\le \expects{|I|}{G}
\le n^{-\delta/8}. 
\label{eq:isolated_node_0}
\end{align}
Moreover,
\begin{align}
\probs{
\mu\left(G[\supp(D)]\right)
=
\left\lfloor \frac{|\supp(D)\setminus I|}{2}\right\rfloor
}{G}
\ge 1-n^{-1-\delta/8}.
\label{eq:matching_size}
\end{align}
\end{proposition}
\blue{The two statements in Proposition~\ref{prop:perfect matching} correspond to the two steps above: \eqref{eq:isolated_node_0} controls the unavoidable loss from isolated vertices, while \eqref{eq:matching_size} shows that, once these vertices are removed, no additional matching loss arises except for the unavoidable parity effect. Therefore, \prettyref{thm:match-1-n} follows directly from \prettyref{prop:perfect matching}. We defer the formal proof of \prettyref{thm:match-1-n} to Section \ref{sec:app-trans-cons} and focus here on establishing \prettyref{prop:perfect matching}. The main technical challenge is to establish the second statement, for which we use the classic Tutte's theorem to characterize the possible matching bottlenecks.} In what follows, we condition on $V = \supp(D)$ with $|V| \ge (1-\eta) np$. 
%, \emph{with self-loops and duplicated parallel edges removed}.  
Let 
$I$ be the set of isolated vertices in $G[V]$, and let $G^\ast =G[V] - I$ be the subgraph of $G[V]$ obtained by removing these isolated vertices. %Note that the conclusion \prettyref{eq:isolated_node_0} in~\prettyref{prop:perfect matching} directly follows from~\prettyref{eq:isolated_node} in~\prettyref{lmm:no_small_component}. Thus,  it remains to show that 
Under this notation, the statement \eqref{eq:matching_size} in~\prettyref{prop:perfect matching} asserts that
$G^*$ has a maximum matching of size $\floor{|V(G^*)|/2}$ with probability at least $1-n^{-1-\delta/8}.$
%\nbr{I thought we should use Tutte–Berge formula ?}

\begin{lemma}[Tutte's Theorem  \citep{tutte1947factorization}, Tutte–Berge formula \citep{berge1958couplage}]\label{lmm:tutte}
Let $G$ be a finite undirected graph (simple or multi-graph).
Then $G$ has a perfect matching if and only if for every subset
$S\subset V(G)$, the number of odd components in $G-S$, denoted by $\odd(G-S)$, satisfies
\$
\odd(G-S)\le |S|,\  \forall S\subset V(G) .
\$
More generally, the number of unmatched vertices in a maximum matching is given by 
$$
\max_{S\subseteq V(G)}\left(\odd(G-S)-|S|\right).
$$
\end{lemma}

Note that if $G$ has an even number of vertices, then $\odd(G-S) \equiv |S|  \pmod{2} $. Thus, by \prettyref{lmm:tutte}, a graph $G$ has a maximum matching of size $\lfloor |V(G)|/2\rfloor$ if and only if $\odd(G-S) \le |S|+1$ for all $S \subset V(G). $ 
Hence, to prove \prettyref{prop:perfect matching}, it suffices to show that with probability at least $1-n^{-1-\delta/8}$, the total number of connected components (even or odd) of $G^*-S$, denoted by $c(G^*-S)$, satisfies 
\#\label{eq:to-prove}
c(G^*-S)\le |S|+1, \quad \forall S \subset V(G^*).
\# 
We refer to any
$S$ satisfying $c(G^*-S) \ge |S|+2$ as a \emph{bottleneck}, since it violates the desired condition.  
Our analysis rules out the existence of bottlenecks of any size using key connectivity properties of random $d$-regular graphs, stated in the next lemma. Let $\mathrm{Ind}(G)$ be the maximum size of an independent set in $G$, \blue{where an independent set is a set of vertices with no edges between any pair of them,} and let $X_k$ be the number of connected components of size $k$ in $G[V]$.

\begin{lemma}\label{lmm:match-prop}
Under the same conditions as in \prettyref{prop:perfect matching}, let $V = \supp(D)$. There exists a constant $\gamma>0$ such that, with $\zeta = n^{-1-\delta/4}$,
\begin{align}
&\probs{\mathrm{Ind}(G) \ge 2n\log d/d}{G} \le
\zeta, \label{eq:ind_set} \\
&\probs{ \exists \text{ disjoint } U, W \subset [n] \text{ with } |U|, |W| \ge  {4 n \log d}/{d}\colon e(U, W)= 0}{G} 
\le \zeta, \label{eq:no_big_cut_1}\\
& \probs{ \exists \text{ disjoint } S, T\subset  V \text{ with }  |S|+1 \le |T| \le \gamma n \colon e(S,T) = e(V\setminus T, T) \ge |S| + 1}{G} \le \zeta, \label{eq:congestion} \\
&\probs{X_1 \ge 1}{G} 
\le n^{-\delta/8}, \ \pr_G\Big\{ \sum_{k=2}^{\gamma n} X_k \ge 1 \Big\} 
\le \zeta. \label{eq:no_small_comp} 
\end{align}
\end{lemma}
In words, \prettyref{lmm:match-prop} asserts that, with high probability, a random $d$-regular graph with $d$ given by \eqref{eq:cond_d} satisfies these key structural properties: \emph{no large independent set} \eqref{eq:ind_set}, \emph{no big cut} \eqref{eq:no_big_cut_1}, \emph{little congestion} \eqref{eq:congestion}, and \emph{no small component} \eqref{eq:no_small_comp}. We defer the formal proof of \prettyref{lmm:match-prop} to Section \ref{sect:app-transportation} and focus here on how these properties are used. The first bound in \eqref{eq:no_small_comp} directly yields \prettyref{eq:isolated_node_0}. To establish \prettyref{eq:matching_size}, \eqref{eq:ind_set} rules out large bottlenecks, while \eqref{eq:no_big_cut_1}--\eqref{eq:no_small_comp} together rule out the remaining small bottlenecks. Figure~\ref{fig:trans-bott} illustrates the latter argument.

Our proof strategy is inspired by the classical work of \cite{p__erds__1966}, which proved the existence of perfect matching for \ER random graph with average degree $\log n+\omega(1)$, \blue{where $\omega(1)$ denotes a quantity that diverges as $n\to \infty$.}
Suppose, for the sake of contradiction, there exists a bottleneck $S \subset V(G^*)$
with $|S|=s$ such that $c(G^*-S) \ge s+2$. We first rule out the possibility of a large bottleneck. Since selecting one vertex from each connected component forms an independent set, the \emph{no-large-independent-set} property \eqref{eq:ind_set} ensures that $s+2< 2n \log d/d$. Next, we show that the existence of a small bottleneck also leads to a contradiction, as illustrated in Figure
\ref{fig:trans-bott}. Consider the connected components of $G^*-S$. Any partition of these components into two groups induces a cut; therefore, by the \emph{no-big-cut} property \eqref{eq:no_big_cut_1}, there exists a union of $s+1$ components of $G^*-S$, denoted by $T$, whose total size is small, specifically $|T| \le \gamma n$. Since both $S$ and $T$ are small, by the \emph{little-congestion} property~\eqref{eq:congestion}, we then show that $e(S,T)\le s$. On the other hand, by the \emph{no-small-component} property \eqref{eq:no_small_comp} of $G$, we obtain $e(S, T)\ge  s+1$, leading to a contradiction. 

\begin{figure}[ht]
    \centering
\begin{tikzpicture}[
    % Define styles for consistency
    node distance=1.5cm and 1cm,
    every node/.style={},
    bottleneck/.style={ellipse, draw, thick, minimum width=1.8cm, minimum height=1.35cm, inner sep=2pt},
    t_node/.style={ellipse, draw, thick, fill=white, minimum width=1.0cm, minimum height=0.75cm},
    set_node/.style={ellipse, draw, thick, fill=white, minimum width=3.2cm, minimum height=2.4cm},
    dashed_box/.style={draw=gray, thick, dashed, rounded corners, inner sep=10pt, fill=gray!5},
    inner_node/.style={circle, draw=blue!70, fill=blue!30, inner sep=1.5pt},
    inner_edge/.style={thick, blue!60!red!60}
]

    % 1. Draw the top node (S)
    \node[bottleneck] (S) {\small $S$};

    % 2. Draw the bottom nodes (T set)
    % Center node
    \node[t_node, below=1.6cm of S] (n3) {};
    % Nodes to the right
    \node[t_node, right=0.5cm of n3] (n4) {};
    \node[right=0.3cm of n4] (dots) {\Large $\cdots$};
    \node[t_node, right=0.3cm of dots] (n5) {};
    % Nodes to the left
    \node[t_node, left=0.5cm of n3] (n2) {};
    \node[set_node, left=0.5cm of n2] (n1) {};

    %--- Small Components on the right ---
\foreach \i in {2, 3, 4, 5} {
    % Add internal nodes and edges to small components
    \foreach \j in {1,2,3, 4} {
         \node[inner_node] (cn\i\j) at ($(n\i.center)+({120+(\j-1)*90}:0.2)$) {};
    }
}
\draw[inner_edge] (cn24) -- (cn21) -- (cn22) -- (cn23);
\draw[inner_edge] (cn34) -- (cn31);
\draw[inner_edge] (cn31) -- (cn32);
\draw[inner_edge] (cn31) -- (cn33);
\draw[inner_edge] (cn44) -- (cn41) -- (cn43) -- (cn42);
\draw[inner_edge] (cn51) -- (cn54) -- (cn53) -- (cn52);
\draw[inner_edge] (cn54) -- (cn52);

\foreach \j in {1,..., 9} {
         \node[inner_node] (ln\j) at ($(n1.center)+({120+(\j-1)*40}:0.9)$) {};
}
\draw[inner_edge] (ln2) -- (ln3) -- (ln4) -- (ln5) -- (ln6) -- (ln7) -- (ln8) -- (ln9) -- (ln1);
\draw[inner_edge] (ln1) -- (ln5);
\draw[inner_edge] (ln1) -- (ln3) -- (ln7)-- (ln9);

\foreach \j in {1,..., 6} {
         \node[inner_node] (sn\j) at ($(S.center)+({120+(\j-1)*60}:0.45)$) {};
}
\draw[inner_edge] (sn6) -- (sn1) -- (sn5) -- (sn6);
\draw[inner_edge] (sn2) -- (sn3);

    % 3. Draw the edges
    % We use 'edge' to push them to the background layer automatically if configured, 
    % or just draw lines.
    \draw[inner_edge, dashed] (ln8) -- (sn2);
    \draw[inner_edge] (sn3) -- (cn24);
    \draw[inner_edge] (sn4) -- (cn31);
    \draw[inner_edge] (sn4) -- (cn41);
    \draw[inner_edge] (sn5) -- (cn51);

    % 4. Draw the grouping box (G* - S)
    % The 'fit' library calculates the size automatically
    \begin{scope}[on background layer]
        \node[dashed_box, fit=(n1) (n5), label={above left: \small  Components of $G^* - S$}] (box) {};
    \end{scope}

    % 5. Draw the curly brace
    \draw[decorate, decoration={brace, amplitude=10pt, mirror}, thick]
        ([yshift=-0.24cm]n2.south west) -- ([yshift=-0.24cm]n5.south east)
        node[midway, below=12pt] {\small $T$: Small components};
\end{tikzpicture}
\caption{Illustration of a small bottleneck $S$. The figure shows the graph $G^*$, decomposed into the node set $S$ and the connected components of $G^*-S$. By the {no-big-cut} property \eqref{eq:no_big_cut_1}, there exists a union $T$ of at least $s+1$ small components, whose total size $|T|$ is small. Since both $S$ and $T$ are small, the {little-congestion} property~\eqref{eq:congestion} implies that $e(S,T)\le s$. However, by the {no-small-component} property \eqref{eq:no_small_comp} of $G^*$, each of these components must have at least one edge to $S$ (shown as the solid lines), so $e(S, T) \ge s+1$. This contradiction rules out the existence of such a small bottleneck. We also note that the largest component of $G^*-S$ may have no edge to $S$, hence shown with a dashed line.}
\label{fig:trans-bott}
\end{figure}
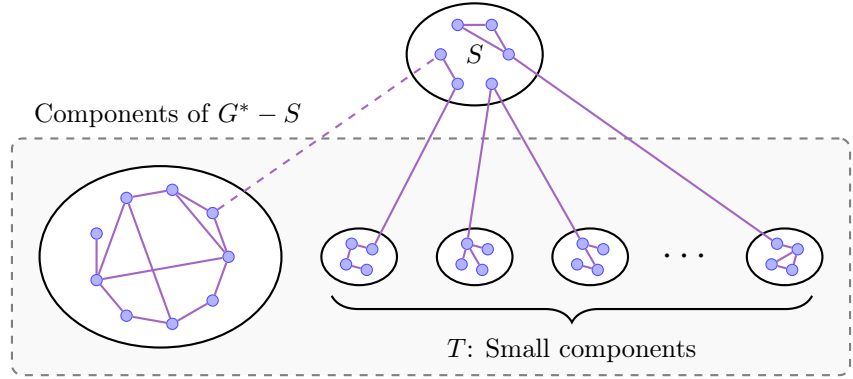

We remark that in \prettyref{lmm:match-prop}, the properties \eqref{eq:congestion} and \eqref{eq:no_small_comp} concern the induced subgraph $G[V]$. Ensuring that these properties hold under the minimum-degree requirement fundamentally relies on the random $d$-regular graph design. %These conditions are crucial for achieving bounded loss, and their validity is specific to random regular graphs. 
In contrast, \eqref{eq:ind_set} and \eqref{eq:no_big_cut_1} are properties of the original graph $G$ that are already present whenever $d$ is a sufficiently large constant. Hence, they do not determine the sharp degree threshold.

\begin{proof}[Proof of~\prettyref{prop:perfect matching}] Assume that none of the low-probability events in \prettyref{lmm:match-prop} occurs. Suppose, for the sake of contradiction, that there exists a set $S \subset V(G^*)$
with $|S|=s$ such that $c(G^*-S) \ge s+2$. 
Selecting one vertex from each connected component in $G^*-S$ forms an independent set in $G$; hence 
$c(G^*-S) \le \mathrm{Ind}(G)$.
By \eqref{eq:ind_set}, $\mathrm{Ind}(G) < 2n\log d/d$. Hence, we have
$$
s+2 \le c(G^*-S) \le \mathrm{Ind}(G) < \frac{2 n\log d}{d}.
$$

%Next, we show that if we divide the connected components in $G^*-S$ into two groups, then one group must have fewer than $4n\log (d)/d$ vertices.
Next, we show that $G^*-S$ contains $s+1$ components whose total number of vertices among them is less than $M\triangleq 4n\log d/d$.
Specifically, let $T_0, T_1, \cdots, T_{r}$  denote the vertex sets of the connected components in $G^*-S$ with $|T_i|\triangleq t_i$ such that $t_0 \ge t_1 \ge \cdots \ge t_r \ge 1$, where $r+1=c(G^*-S)\ge s+2.$  Note that for any $0 \le j \le r-1$, $\cup_{i=0}^j T_i$ and $(\cup_{i=j+1}^r T_i) \cup I$ do not have any crossing edge and thus form a cut in $G$, recalling that $I$ denotes the set of isolated vertices in $G[V]$.  It follows from the no-big-cut property of $G$ in~\eqref{eq:no_big_cut_1} that 
\begin{align}
\text{ either } \sum_{i=0}^j t_i <M 
\text{ or } \sum_{i=j+1}^r t_i + |I| < M, \quad \forall 0 \le j \le r-1. \label{eq:either_or}
\end{align}
%where $N=4 n \log d/d$. \nb{It's a bit weird to define $N$ here given $4 n \log d/d$ appeared earlier. Also, $N$ can be confused with $N(\delta)$, maybe use $C$ or something else?}
Suppose $t_0 < M.$ Let $j\ge 1$ be the smallest index such that 
$\sum_{i=0}^{j-1} t_i < M \le  \sum_{i=0}^{j} t_i $. Then 
$\sum_{i=0}^j t_i \le \sum_{i=0}^{j-1} t_i +
t_0 \le 2M$. Since
$\sum_{i=0}^r t_i + |I| = |V|-s$, 
it follows that 
$$
\sum_{i=j+1}^r t_i + |I| \ge  |V| - s -2M
\ge (1-\eta)np - 10n \log d/d \ge M,
$$
where the last inequality holds for sufficiently large $n$, such that $d$ in \eqref{eq:cond_d} satisfies $\log d/d \le p/28$. %\nb{Here is where we used the condition $n$ is large enough. I prefer not impose condition on $d$ except \eqref{eq:cond_d}}
Thus, we arrive at a contradiction to~\prettyref{eq:either_or}. Hence, we must have $t_0 \ge M$, which, by~\prettyref{eq:either_or}, further implies that $\sum_{i=1}^r t_i + |I| < M$.
%and $t_0 \ge  \rho n -s -N.$ \nbr{If I am not mistaken, we do not need the following: Note that 
% $T_0$ and $(\cup_{i=j+1}^r T_j) \cup I$ form a cut in $\tilde{G}$. Thus, we apply
% the no-big-cut property of $\tilde{G}$ established in~\prettyref{lmm:no_big_cut_2} with
% $\alpha=C \log^2 d/(\rho^2 d^2) $
% and $\beta=1- (s+N)/(\rho n) \ge 1-6\log (d)/(\rho d)$. The conditions in~\prettyref{lmm:no_big_cut_2}  are all satisfied by choosing $C$ to be a sufficiently large constant. Hence, applying~\prettyref{lmm:no_big_cut_2}  yields that 
% $$
% \sum_{i=1}^r t_i + |I| \le C \frac{n \log^2(d)}{\rho d^2}.
% $$}
Let $T=\cup_{i=1}^r T_i$. Then $|T| \le M.$

We now apply the little-congestion property~\eqref{eq:congestion} to $S$ and $T$. 
Since $r \ge s+1$, $|T| = |\cup_{i=1}^r T_i| \ge s+1.$ 
Note that $e(V\setminus T, T) = e(S,T)$, since the connected components $T=\cup_{i=1}^r T_i$ have no edges to the remaining component $T_0$ or to the isolated nodes $I$. Moreover, since $|T|\le M = 4n \log d/d $, we have $|T| \le \gamma n$ for $\gamma$ in Lemma \ref{lmm:match-prop} and sufficiently large $d$. Therefore, by \eqref{eq:congestion}, we have $e(S,T) \le s$. 

Conversely, 
$e(S, T_i) \ge 1$ for every $1 \le i \le r$, 
because otherwise $T_i$ is a connected component in $G[V]$, contradicting the no-small-component property of $G[V]$ in~\eqref{eq:no_small_comp}. Therefore, 
$e(S, T) \ge \sum_{i=1}^r e(S, T_i) \ge r \ge s+1$, leading to a contradiction. 

In summary, when none of the low-probability events in \prettyref{lmm:match-prop} occur, \eqref{eq:to-prove} holds, and consequently $\mu(G^*) = \floor{|V(G^*)|/2}$.  
Since each event occurs with probability at most $\zeta$, a union bound yields that none of them occur with probability at least $1-4\zeta = 1-O(n^{-1-\delta/8}).$ 
\end{proof}

\blue{The proof of Proposition \ref{prop:perfect matching} also reveals a close parallel with the process-flexibility analysis.
In both settings, we first control large bottlenecks and then rule out small bottlenecks. We also note a close parallel between the small-bottleneck arguments via Lemmas \ref{lmm:epsilon_connectivity} and \ref{lmm:match-prop}. In both settings, the proof combines two connectivity conditions: one ensures sufficiently many relevant edges, while the other prevents these edges from being excessively concentrated. For process flexibility, these roles are played by \eqref{eq:no_waste} and \eqref{eq:flex-d-no-cong}; for transportation, by \eqref{eq:no_small_comp} and \eqref{eq:congestion}. Thus, although the notions of bottleneck differ, the corresponding connectivity properties play analogous roles in the two analyses.}

\subsubsection{Proof of Theorem \ref{thm:constant d}}
\blue{For the fractional-loss regime, the degree $d$ remains constant as $n\to \infty$, and a different approach becomes useful. In this regime, a random $d$-regular graph is sparse: each vertex has only a constant number of neighbors, and because the edges are formed randomly over an increasingly large graph,  short cycles are rare. Consequently, the graph is locally tree-like: any fixed-radius neighborhood around a typical vertex resembles a tree with high probability. This behavior is formalized by local weak convergence, which describes the limiting distribution of such local neighborhoods.} A celebrated result from the literature shows that, for such locally tree-like graphs, the asymptotic maximum matching size can be characterized through a local recursion on the limiting tree. We defer the formal definitions of local weak convergence and its limiting object, the unimodular Galton--Watson trees, to Section \ref{sect:ec-lwc}.

%characterizes the matching size for graphs with this locally tree-like structure, and defer the formal definitions of local weak convergence and unimodular Galton--Watson trees to Section \ref{sect:ec-lwc}.
%This property allows the asymptotic maximum matching size to be determined by solving local recursive equations on the limiting tree. We recall here a celebrated result from the literature that characterizes the matching size for graphs with this locally tree-like structure, and defer the formal definitions of local weak convergence and unimodular Galton--Watson trees to Section \ref{sect:ec-lwc}.

\begin{lemma}[\citep{bordenave2013matchings}]
\label{lem:matching_limit_ugw}
Suppose that $\{G_n= (V_n, E_n)\}_{n \in \mathbb{N}}$ is a sequence of finite graphs admitting a local weak convergence to a unimodular Galton--Watson (UGW) tree with degree distribution $\pi$. Recall that $\mu(G_n)$ denotes the maximum matching size of $G_n$. Then %the ratio of the maximum matching size $\mu(G_n)$ to the number of vertices $|V_n|$ is given by: \nbr{this is a bit wordy...}
\[
\frac{\mu(G_n)}{|V_n|} \xlongrightarrow{} \frac{1 - \max_{t \in [0,1]} F(t)}{2}, 
\]
where
\begin{align}
F(t) = t\phi'(1-t) + \phi(1-t) + \phi\left(1 - \frac{\phi'(1-t)}{\phi'(1)}\right) - 1, \label{eq:def_F}
\end{align}
and $\phi(t) = \sum_k \pi_k t^k$ is the probability generating function of the degree distribution $\pi$. 
% \blue{I thought the MGF is $\E[e^{tX}] = \sum_{k} M_k t^k/k!$ with $M_k = \E[t^k]$ being the $k$th moment. While the PGF (only for discrete random variables) is $\E[t^{X}] = \sum_{k} \pi_k t^k$ with $\pi_k = \pr(X=k)$ being the probability mass function. Although \citep{bordenave2013matchings} refers to it as MGF, maybe it's better to call it PGF to avoid confusion?}  
% \nbr{OK. I do not think the name really matters here, but I'm Ok to use probability generating function for better clarify.}
%\nbr{check whether this degree distribution requires zero mass at degree 0.}
\end{lemma}

Lemma \ref{lem:matching_limit_ugw} provides powerful machinery for determining the asymptotic size of the maximum matching in sparse random graphs. \blue{In our setting, the random node deletions induced by the demand realization translate naturally into independent node deletions on the local limiting tree. Applying Lemma \ref{lem:matching_limit_ugw} then involves two steps: first, identifying the limiting tree after node deletion; and second, analyzing the corresponding matching objective $\max_{t \in [0,1]} F(t)$.} While this is not the main focus of the paper, we illustrate this approach with a warm-up example: we characterize the matching loss of \ER graphs with constant average degree $d$, thereby resolving the open question posed in~\citet[Section 5.4]{feng2024designing}.  %In this setting, we can show that the local degree distribution of ${\rm ER}_n[\supp(D)]$ is $\Pois(\lambda)$ with $\lambda = (1-q)d$. %Since \cite{bordenave2013matchings} also discusses this $\Pois(\lambda)$ case as an example, we omit some computational details. \nbr{omit what???} 

\begin{corollary}\label{cor:er-const}
Suppose $D\sim\cD_p$. For the \ER random graph ${\rm ER}_n$ with constant average degree $d$, the matching loss satisfies
$$
\frac{L({\rm ER}_n, D)}{np} \xlongrightarrow{\mathbb{P}}
F_\lambda(t_*), 
$$
where $F_\lambda(t_*)=(\lambda t_*+1)\exp(-\lambda t_*) + t_*-1$, $\lambda = (1-q)d$, $t_* \in (0,1)$ is the smallest fixed point of the equation $t=\exp(-\lambda \exp(-\lambda t))$, and $F_\lambda(t_*)=\exp(-(1-q)d)(1+o_d(1))$ for large $d$.
%The convergence also holds in $L^r$ norm.
\end{corollary}
\begin{proof}
%Let ${\rm ER}_n$ be an \ER random graph with average degree $d$, and l
Let ${\rm ER}_n[\supp(D)]$ denote the subgraph induced by the random vertex set $\supp(D).$ By the Poisson splitting property, one can show that ${\rm ER}_n[\supp(D)]$ converges locally in probability to a UGW tree with Poisson degree distribution $\Pois(\lambda)$,
where $\lambda=pd=(1-q)d$.
The probability generating function for $\Pois(\lambda)$ is given by 
$\phi(t)=\exp(\lambda(t-1))$.
Applying~Lemma \ref{lem:matching_limit_ugw} and solving for the maximizer of $F(t)$, we obtain that 
   $$
\frac{\mu\left({\rm ER}_n[\supp(D)]\right)}{|\supp(D)|}
   \xlongrightarrow{\mathbb{P}} \frac{1 - F_\lambda(t_*) }{2}, 
   $$
   where  $t_* \in (0,1)$ is the smallest fixed point of the equation $t=\exp(-\lambda \exp(-\lambda t))$
   and $F_\lambda(t_*)=(\lambda t_*+1)\exp(-\lambda t_*) + t_*-1$. For large $\lambda$, we have $t_*=(1+o_\lambda(1))\exp(-\lambda)$
and $F_\lambda(t_*)=(1+o_\lambda(1))\exp(-\lambda)$. Since $|\supp(D)|/(np) \to 1$ in probability, it follows that the matching loss satisfies 
   \begin{equation*}
   \frac{L({\rm ER}_n,D)}{np}  \xlongrightarrow{\mathbb{P}}
   F_\lambda(t_*). \qedhere
   \end{equation*}
%   Because $L({\rm ER}_n,D)/(np)$ is uniformly bounded by $1/p$,  the convergence also holds in $L^r$ norm. 
\end{proof}
As shown in \prettyref{cor:er-const}, 
the limiting matching loss of \ER graphs with constant average degree $d$ is on the order of $(1+o_d(1))\exp(-(1-q)d)$. Thus, following the same steps used later to prove Theorem \ref{thm:constant d}, we conclude that an average degree $d = (1+o_\epsilon(1)) \log(1/\epsilon)/(1-q)$ is necessary and sufficient for \ER graphs to achieve an $\epsilon$-fractional loss.

Returning to our main focus, the same two steps become substantially more delicate for random $d$-regular graphs.  First, although it is natural to conjecture that $G[\supp(D)]$ (with $G\sim\cG(n,d)$) converges locally in probability to a UGW tree with degree distribution $\Binom(d,p)$, this convergence needs to be established rigorously. Second, sharply characterizing $\max_{t \in [0,1]} F(t)$ is technically more involved. The following lemmas establish these key results.
\begin{lemma}\label{lmm:lwc_rrg}
Let $\{\widetilde G_n\}_{n \in \mathbb{N}}$ be the sequence of random graphs, where $\tilde{G}_n=G_{n,d}[\supp(D)]$, $D\sim \cD_p$, and $G_{n,d}\sim\cG(n,d)$ with a constant $d.$ Then $\widetilde G_n$ converges locally in probability to the unimodular Galton--Watson tree with degree distribution $\Binom(d, p)$.
\end{lemma}
The probability generating function 
of $\Binom(d,p)$ is 
$ \phi(t) = \expects{t^X}{X\sim \Binom(d,p)}=(pt+1-p)^d$. 
Hence,  by~\prettyref{eq:def_F}, the corresponding function is $F_{d,p}(t)$, defined as
\$
F_{d,p}(t)\triangleq tdp(1-pt)^{d-1}+(1-pt)^d+(1-(1-pt)^{d-1}p)^d-1.
\$
Let $F^*_{d,p} \triangleq \max_{t \in [0,1]} F_{d,p}(t)$.  The next lemma characterizes $F^*_{d,p}$ %shows that $\max_{t \in [0,1]} F_{d,p}(t)=(1+o_d(1)) q^d$ 
for a large constant $d$. 
\begin{lemma}\label{lmm:fdddd}
%For $p>0$, let $d_0 \ge 3$ be the largest %\nb{smallest?} \blue{largest} value of $d$ such that $p^2(d-1)^2 q^{d-2} \ge 1/3 $. %\nb{$\le$} \nbr{I think this should be correct.} 
For $p\in(0,1)$ and all sufficiently large $d$, %such that $p^2(d-1)^2 q^{d-2} < 1/3 $, 
we have
$
q^d \le F^*_{d,p} \le q^d ( 1+ 9d^2 p^2 q^{d-2}). 
$
In particular, $F^*_{d,p}=q^d(1+o_d(1)),$ and  $F_{d,p}^* \le \epsilon$ for $d$ in \eqref{eq:match-const-d} with $c = [6p/(q\log(1/q))]^2$ and sufficiently small $\epsilon$.
%In particular, 
%=q^d+\frac{d(d-1)p^2q^{2d}}{2(p^2(d-2)q+q^2+p^2(d-1))}+R_d$, where $q=1-p,$ $R_d=O(d^3q^{3d-3})$.
%\yw{YW: the stated \(d_0\ge3\) need not exist as defined; for example, when \(p=0.9\), no \(d\ge3\) satisfies the displayed inequality.}
\end{lemma}
Combining Lemma~\ref{lem:matching_limit_ugw} with Lemmas~\ref{lmm:lwc_rrg} and~\ref{lmm:fdddd} yields the following proposition, from which Theorem~\ref{thm:constant d} follows immediately. We defer the short derivation of Theorem~\ref{thm:constant d} to Section \ref{sect:ec-lwc}.

\begin{proposition}\label{prop:match-const}
For the random $d$-regular graph $G\sim\cG(n,d)$ with constant degree $d$, we have %the ratio of the matching loss to $np$ converges in probability to $F_{d,p}^*$, that is,
$$
\frac{L(G, D)}{np} \xlongrightarrow{\mathbb{P}} F_{d,p}^*,
$$
where $F_{d,p}^* =  q^d(1+o_d(1))$. Moreover, the convergence also holds in $L_1$ norm, that is, 
 \$\lim_{n\rightarrow \infty} \expect{|L(G, D)/(np) - F_{d,p}^*|}=0.\$ 
%where $ F_{d,p}(t)=tdp(1-pt)^{d-1}+(1-pt)^d+(1-(1-pt)^{d-1}p)^d-1$.
%\nbr{I do not like the current rewriting...} \nb{There are a few issues with the Proposition is stated. First, my understanding is that for Theorem 5, we just need $\lim_{n\rightarrow \infty}E|\frac{L(G, D)}{np} - F_{d,p}^*| = 0$. If that's the case, I don't think the generality of the proposition (without appropriate definitions of $L^r$ norm convergence and convergence in probability) is counterproductive. Second, we should define $F_{d,p}^* \triangleq  \max_{t \in [0,1]} F_{d,p}(t) $ first.} 
\end{proposition}

\begin{proof}%[Proof of \prettyref{prop:match-const}]
Applying Lemma~\ref{lem:matching_limit_ugw} with Lemmas \ref{lmm:lwc_rrg} and \ref{lmm:fdddd}, 
the maximum matching size satisfies
$$
\frac{\mu\left({G}[\supp(D)]\right)}{|\supp(D)|}
\xlongrightarrow{\mathbb{P}} \frac{1 - F_{d,p}^*}{2}.
$$
Using the fact that 
$|\supp(D)|/(np) \to 1$ in probability, the matching loss satisfies 
%\nb{Can we provide another line to show the calculations.}
$$
\frac{L(G, D)}{np} \xlongrightarrow{\mathbb{P}}
  F_{d,p}^*. 
$$
Since $L(G,D)/(np)$ is uniformly bounded by $1/p$,  the convergence also holds in $L_1$ norm.    
\end{proof}

\section{Experiments}\label{sect:exp}
To support our theoretical results on the optimality of random $d$-regular graphs, we conduct numerical experiments for both problems. Additional results and discussions are deferred to Section \ref{sec:additional-flex-comparison}.
\subsection{Process Flexibility Design}\label{sect:exp-pro-flex}
%We evaluate the performance of various sparse bipartite graphs for the process flexibility problem.
%\nb{Shall we include an experiment study for general demand distribution?}
\paragraph{Experimental Setup.} In the simulations, we consider systems with $n$ supply and $n$ demand nodes. 
Each demand node is independently retained with probability $p$, and is assigned a demand of $1/p$ units if retained.
We numerically study three representative scenarios: $(n, p) = (200, 0.5)$ (high deletion probability, small $n$), $(n, p) = (500, 0.5)$ (high deletion probability, large $n$), 
and $(n, p) = (400, 0.8)$ (low deletion probability, medium $n$).
For each scenario, we generate $1000$ independent demand realizations. For each graph design, we evaluate its performance across these realizations by computing the average loss relative to the full flexibility design. 

Specifically, for each average degree level, we consider the following graph designs: the \ER graph, the probabilistic expander from \cite{chen2015optimal}, the random $d$-regular graph, and the $K$-chain. Here, a probabilistic expander is generated by uniformly sampling $d/2$ demand nodes with replacement for each supply node, and uniformly sampling $d/2$ supply nodes with replacement for each demand node. For the three random graph families, we generate $10$ independent instances of each graph design and report the average loss along with the corresponding $95\%$ confidence intervals computed across these instances. We also compute the theoretical lower bound on the expected loss directly using the formula $nq^d \prob{\Binom(n-d, p) \ge np}$ as derived in Theorem \ref{thm:lower_bound_bp}.

\paragraph{Results.}
Figure \ref{fig:flex-design} presents the results for the three scenarios. In each plot, the $x$-axis represents the average degree of the graphs, and the $y$-axis shows the average loss relative to the full flexibility design, along with the corresponding $95\%$ confidence intervals, plotted on a logarithmic scale. Note that some confidence intervals extend to negative infinity since the loss is shown on a logarithmic scale; this occurs when one or more of the $10$ graph instances achieve an average loss of exactly zero across the $1000$ demand realizations. Across the three plots with different values of $p$, recall that the optimal loss scales as $nq^d$; hence, for larger $p$ (and smaller $q$), a smaller average degree $d$ is sufficient to achieve the same level of loss.

\begin{figure}[htp]
\centering
\includegraphics[width=.33\textwidth]{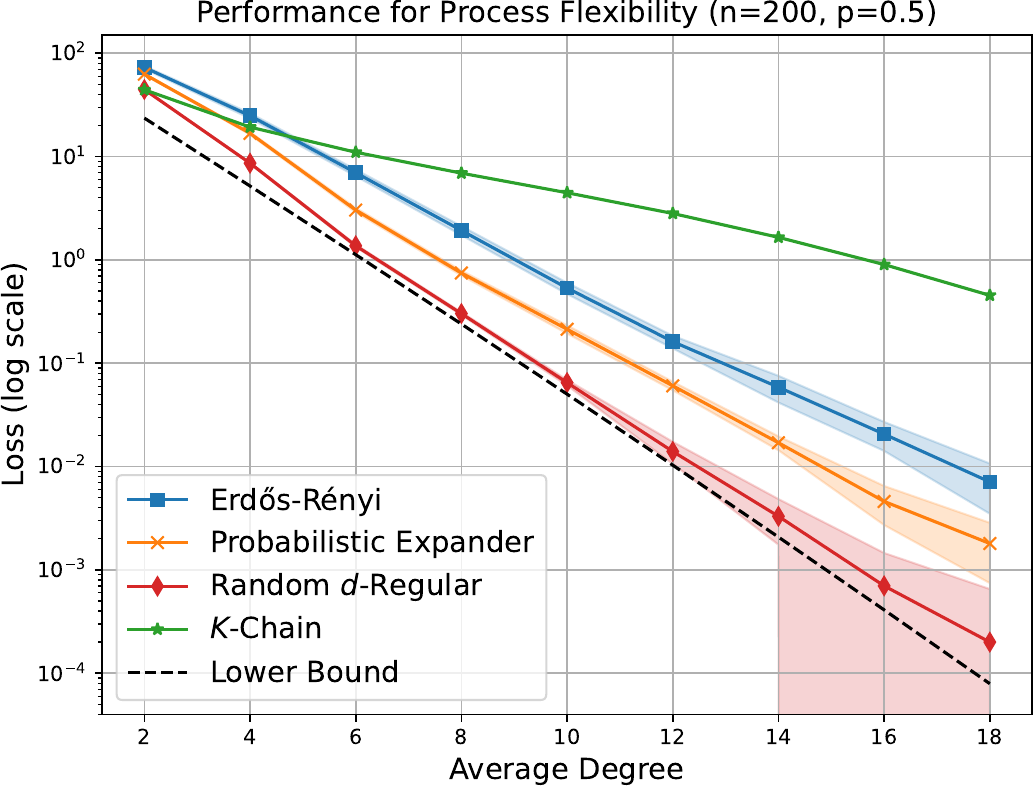}
\includegraphics[width=.33\textwidth]{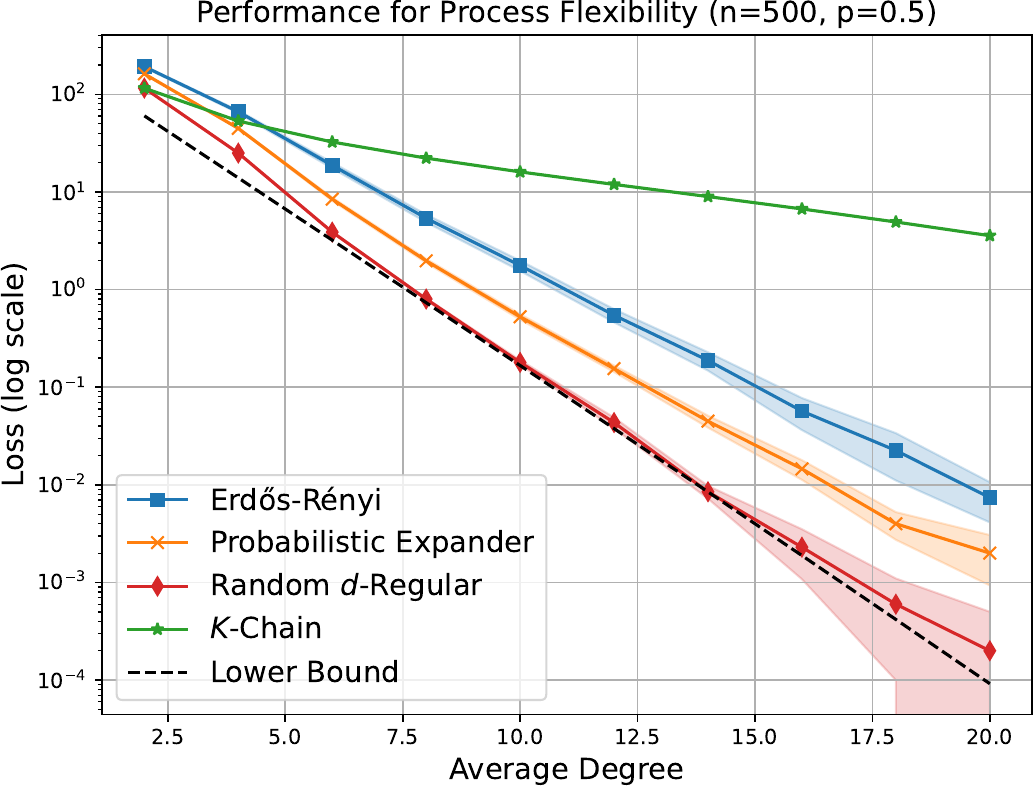}\hfill
\includegraphics[width=.33\textwidth]{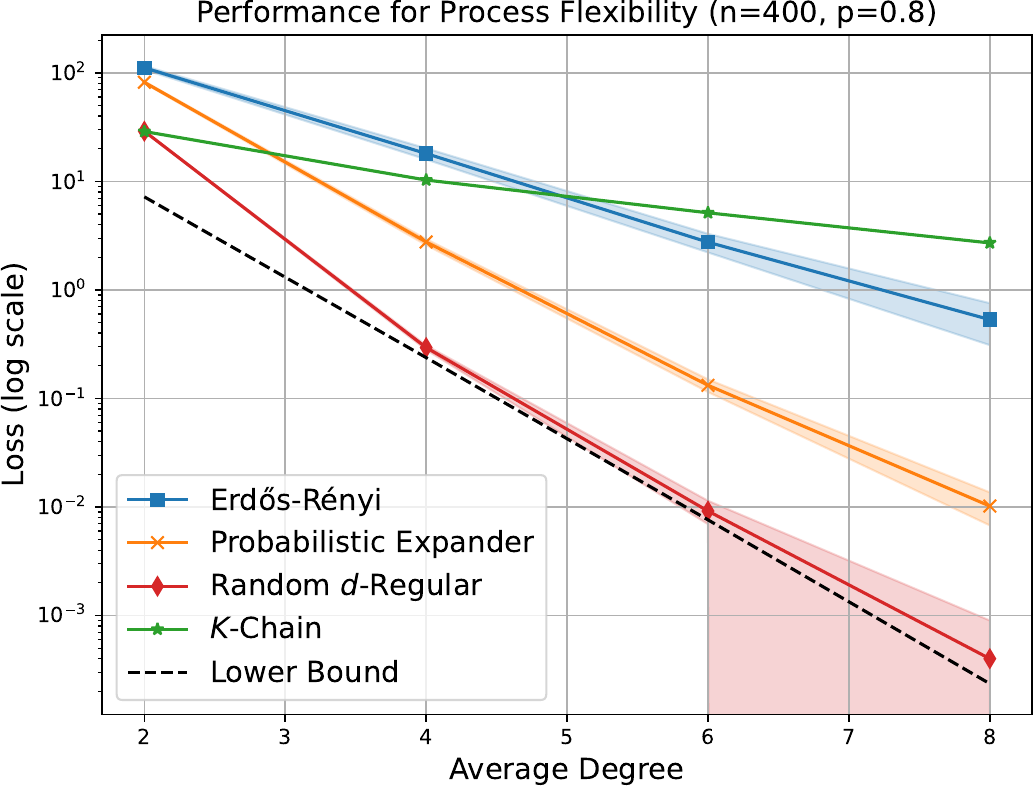}
% \hfill
% \includegraphics[width=.33\textwidth]{figures/flex-design/n_500_p_0.25.pdf}
% \includegraphics[width=.33\textwidth]{figures/flex-design/n_100_p_0.5.pdf}\hfill
% \includegraphics[width=.33\textwidth]{figures/flex-design/n_200_p_0.5.pdf}
\caption{Performance of graph designs for process flexibility.} %\nbr{Can we change the legend starting from $K$-Chain?}
\label{fig:flex-design}
\end{figure}

The results are consistent with our theoretical findings. In our experiments, the random $d$-regular design performs substantially better than the other designs, which fail to achieve the optimal asymptotic sparsity. While our analysis establishes optimality only in the asymptotic regime of large \( n \) and \( d \), its performance is already very close to the theoretical lower bound for moderate graph sizes (\( n = 200 \) to \( n = 500 \)) and relatively small average degrees (\( d = 2 \) to \( d = 20 \)).

\subsection{Middle-Mile Transportation Design}
Next, we compare the performance of sparse graphs for the middle-mile transportation problem.

\paragraph{Experimental Setup.} 
We use the same three $(n,p)$ scenarios and simulation procedure as for process flexibility. We compare the Erd\H{o}s--R\'enyi, random $d$-regular, and $K$-chain designs by their average loss relative to the complete graph. We also report the lower bound from Theorem~\ref{thm:lower_bound}, which is tighter than the $(npq^d-1)^+$ bound used by \citet{feng2024designing}.

%The setup is similar to that of the process flexibility. We consider systems with $n$ nodes, where each node is independently retained with probability $p$. We again examine three representative scenarios: $(n, p) \in \{(200, 0.5), (500, 0.5), (400, 0.8)\}$. For each scenario, we generate $1000$ independent node-deletion realizations. For each graph design, we evaluate its performance across these realizations by computing the average loss relative to the complete graph. For each average degree level, we consider the \ER graph, the random $d$-regular graph design, and the $K$-chain. For the random graph families, we generate $10$ independent instances of each graph design and report the average loss along with the $95\%$ confidence interval computed across these instances. We also compute the theoretical lower bound on the expected loss using Theorem \ref{thm:lower_bound}. Previous work \citep{feng2024designing} had a similar setup; however, our lower bound from Theorem~\ref{thm:lower_bound} is tighter than the $(npq^d - 1)^+$ bound used therein.

\begin{figure}[ht]
\centering
\includegraphics[width=.33\textwidth]{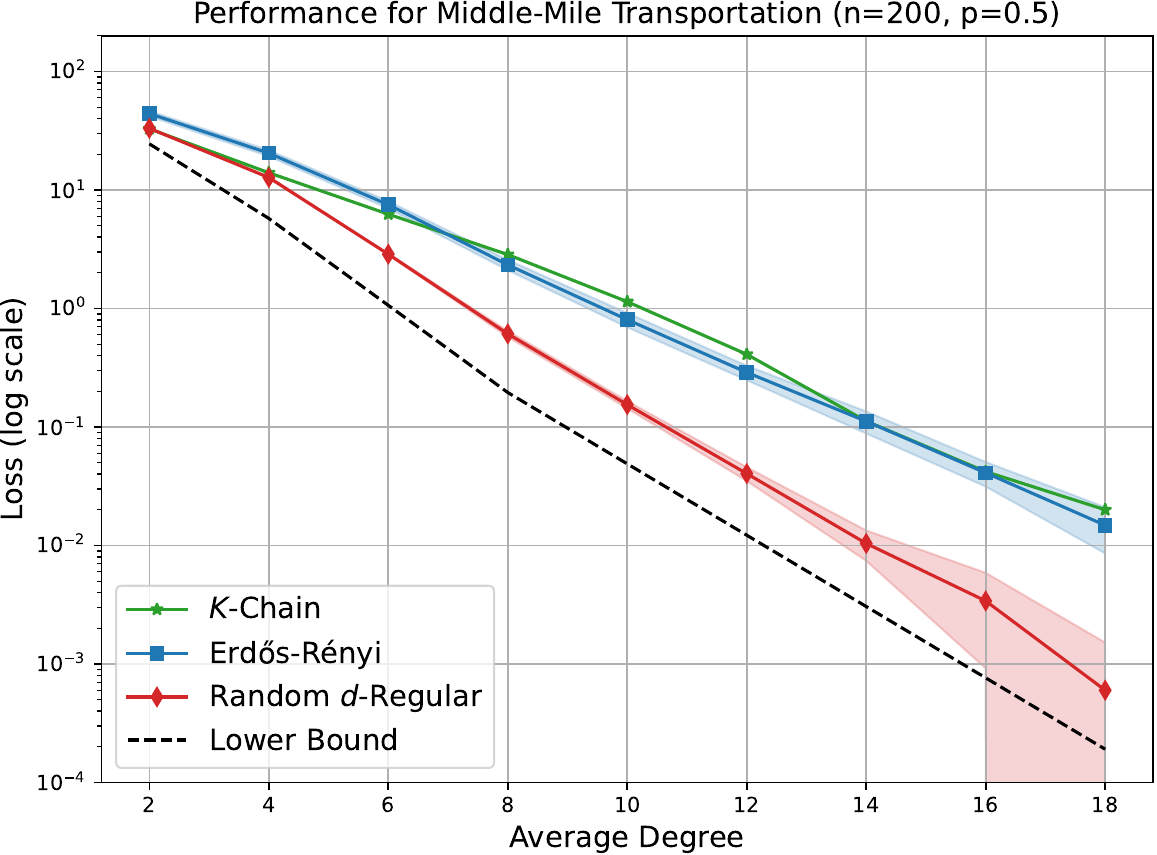}\hfill
\includegraphics[width=.33\textwidth]{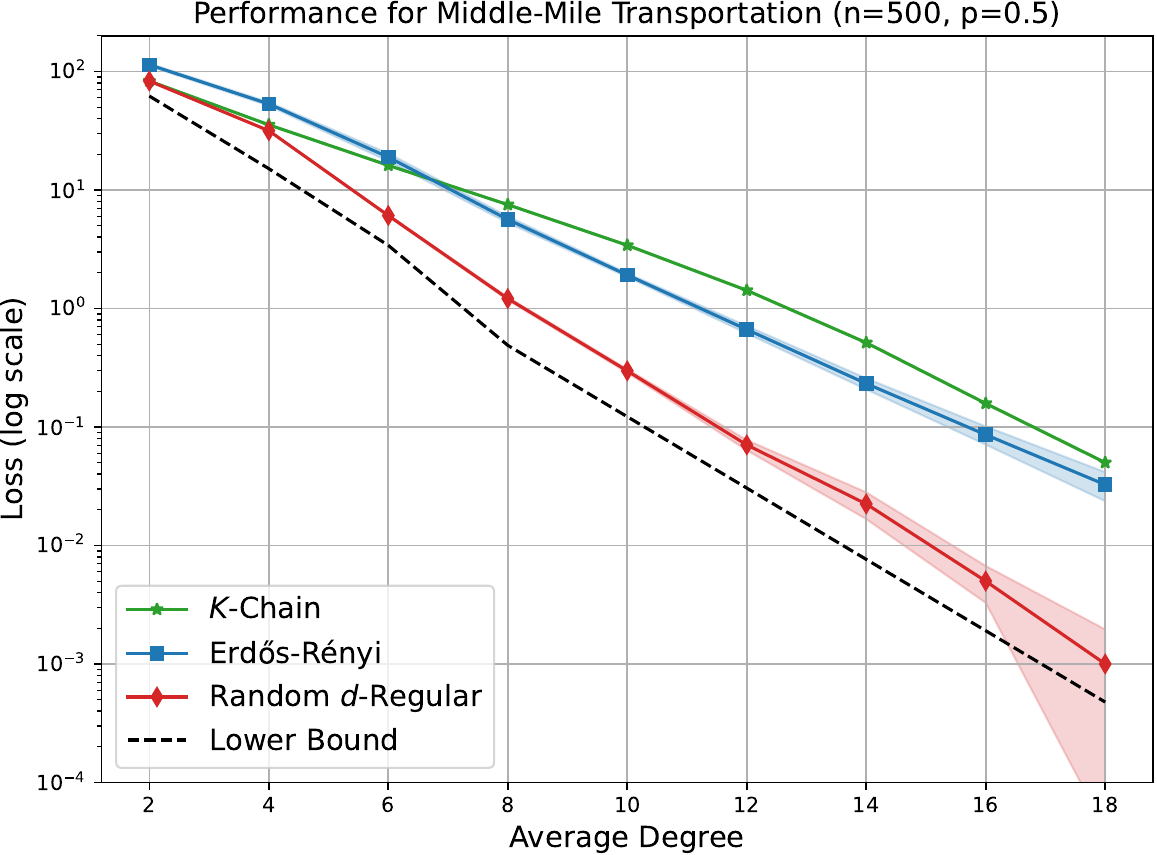}\hfill
\includegraphics[width=.33\textwidth]{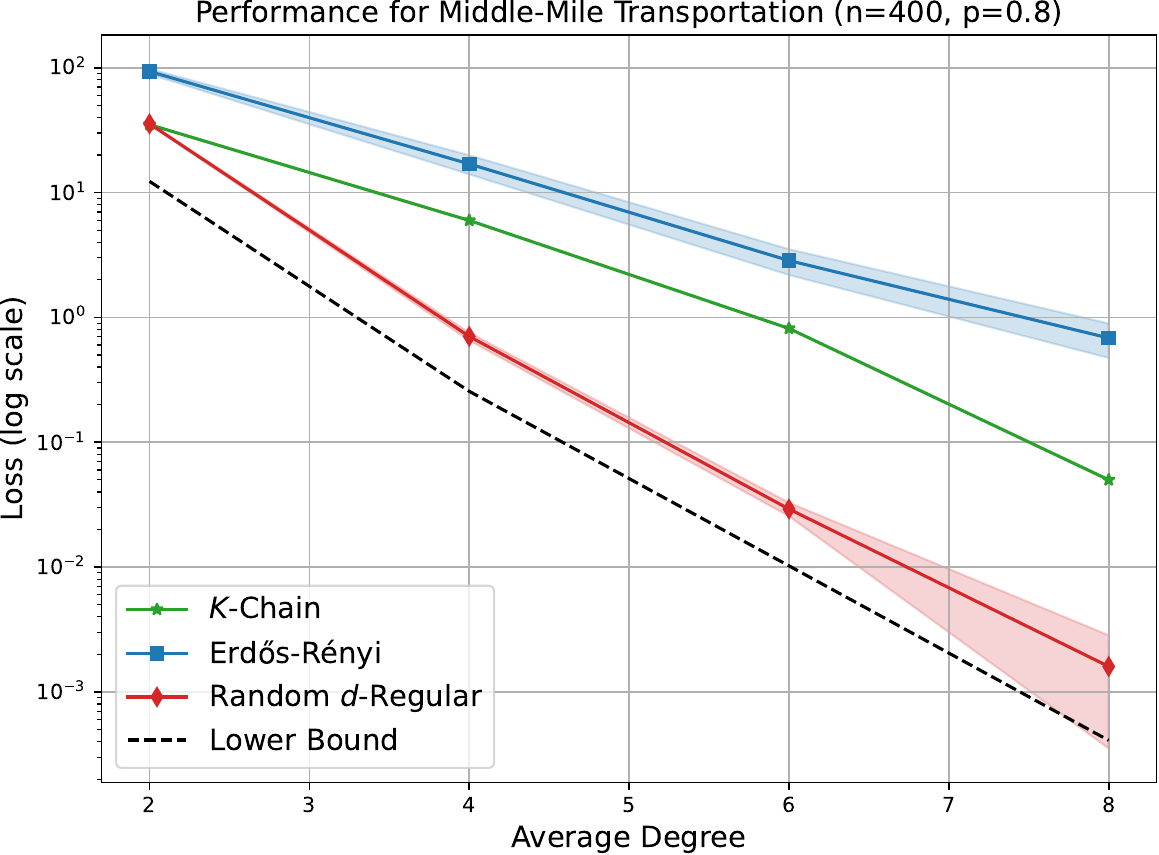}
\caption{Performance of graph designs for middle-mile transportation.} %Flexibility \nbr{The first plot has $n=400$ instead of $n=300$}}
\label{fig:trans-design}
\end{figure}
%\nbr{figures not in the same sizes compared to Figure 6?}
\paragraph{Results.} 
Figure~\ref{fig:trans-design} reports the average loss relative to the fully flexible system across three scenarios in middle-mile transportation. As in process flexibility, random $d$-regular graphs substantially outperform both the $K$-chain and the \ER graphs, since neither the $K$-chain nor the \ER design achieves optimal asymptotic sparsity. Unlike in process flexibility, however, a small but noticeable horizontal gap remains between the performance of the random $d$-regular design and the theoretical lower bound. Nevertheless, we observe numerically that the horizontal gap, i.e., the difference between the degree required by random $d$-regular graphs and the lower bound for a given loss, does not increase as the loss becomes small. As a result, the ratio of this gap to the lower bound approaches zero in the limit, where the loss becomes small and the lower bound becomes large.
This is consistent with our theoretical result that random $d$-regular graphs are asymptotically optimally sparse. Also, we note that when $p$ is fixed ($p=0.5$), the gap between the random $d$-regular design and the theoretical lower bound decreases as we increase $n$ from 200 to 500.

\section{Conclusion and Future Directions}\label{sect:conclude}

In this paper, we show that random $d$-regular designs achieve asymptotically optimal sparsity for both process flexibility and middle-mile transportation. For sufficiently large $d$, their average degree nearly matches the lower bound required to attain a prescribed fractional or constant loss. Our analysis relies on key connectivity properties of random $d$-regular graphs, and numerical experiments further validate this conclusion.

A natural direction for future work is to extend our analysis to broader operational settings. \blue{For process flexibility, one could consider unbalanced systems and heterogeneous supply capacities, perhaps using configuration models with appropriately adapted degree sequences. For middle-mile transportation, one could study split-delivery models, where trucks may combine portions of demand from multiple stations. Such models would not reduce to random node deletion and matching on the residual graph.} Another direction is the stochastic matching problem introduced by \cite{blum2015ignorance-ec} and motivated by kidney exchange: finding a sparse subgraph whose maximum matching remains close to that of the original graph after random edge deletions. Existing work provides sparse designs with strong guarantees \citep{assadi2016stochastic-ec,behnezhad2020stochastic}, but optimal sparsity remains open. Although this problem involves edge rather than node deletions, our design and techniques may offer useful insights. More broadly, it would be valuable to examine whether suitable forms of ``regularity'' in the decision variables, combined with low correlations, yield near-optimal solutions for broader classes of two-stage stochastic optimization problems.

%A natural direction for future research is to extend our analysis to more general settings. For example, one could consider demand distributions beyond the i.i.d.\ Bernoulli model in either the process flexibility or middle-mile transportation settings, or settings with asymmetric supply and demand. It would also be interesting to explore other sparse network design problems, such as the stochastic matching problem introduced by \cite{blum2015ignorance-ec} and motivated by kidney exchange. There, the goal is to design a sparse subgraph whose maximum matching remains nearly as large as that of the original graph after random edge deletions. Subsequent work has proposed sufficiently sparse designs with strong performance guarantees \citep{assadi2016stochastic-ec,behnezhad2020stochastic}, but identifying an optimal design—achieving a prescribed performance loss with the minimum number of edges—remains an open problem. This setting is closely related to the flexibility problems studied in this paper, which instead feature random node deletions. Despite this difference, we believe that our random regular graph design and analytical techniques may provide useful insights toward resolving the optimal design question in stochastic matching.  More broadly, an important direction is to study classes of two-stage stochastic optimization problems,  examining whether enforcing appropriate forms of ``regularity'' in the decision variables—while keeping correlations low—continues to yield near-optimal solutions.

\bibliographystyle{apalike} 

\ifsharedrootabove
  \bibliography{../reference}
\else
  \bibliography{reference}
\fi

\begin{appendix}
\section{Proofs for Process Flexibility Design}\label{sect:app-process}
This section completes the proofs of the main results for the process flexibility problem. Section~\ref{sect:flex-pr-lb} establishes the lower bound by observing that the loss is bounded below by the number of isolated supplies when the realized demand is relatively large. %This observation suggests that, for a graph design to achieve the optimal loss, the number of its isolated supplies must be of the optimal order.
To prove the optimality of random regular graphs, we introduce three key structural properties in Section~\ref{sect:flex-prop}. Among these, as highlighted in~\prettyref{rmk:properties}, the \emph{little-waste-of-supply} property plays a central role in ensuring optimality. Section~\ref{sect:flex-pf-con-fract} shows that these properties are sufficient to establish Lemma~\ref{lmm:epsilon_connectivity}. Section~\ref{sec:general_demand_proofs} proves Lemmas \ref{lmm:lt-convex} and \ref{lem:monotonicity}. Finally, Section \ref{sec:additional-flex-comparison} provides further numerical studies.

\subsection{Lower Bound in Theorem \ref{thm:lower_bound_bp}}\label{sect:flex-pr-lb}
Given a graph $G$ and a demand realization $D$, a supply node is \emph{isolated} if it is disconnected from all demand nodes with nonzero realized demand. The unfulfilled demand $n-z(G,D)$ is always lower bounded by the number of isolated supply nodes, since such supplies cannot serve any demand. When  $|\supp(D)|$ is large, $z(K_{n,n}, D) = n$, so the loss $z(K_{n,n}, D) -z(G, D)$ is also lower bounded by the number of isolated supply nodes. For small $|\supp(D)|$, we use zero as a trivial lower bound. %\prettyref{prop:flex-1-n-small-demand} has shown that random regular graphs can, in fact, achieve zero loss in this small-demand regime. 
\begin{proof}[Proof of Theorem \ref{thm:lower_bound_bp}]
We consider any graph $G$ with supply degrees $d_i$ for $i\in\ns$ and $\sum_{i\in\ns}d_i/n = d$.
For each supply node $i$, let $I_i=1$ if $i$ is disconnected from $\supp(D)$ and $I_i=0$ otherwise.
%Note that $\expect{I_i}=\prob{I_i=1}=q^{d_i}$, as each of its $d_i$ neighbors on the demand side is not in $\supp(D)$ independently with probability $q$.  \nb{This is not used in the later proof}
Since each $I_i=1$ leads to one unit of loss compared with $n$, we have $z(G,D) \le n-\sum_{i=1}^n I_i$ and thus 
\$ 
z(G,D) \le \min\{n-\sum_{i=1}^n I_i, |\supp(D)|/p\}.
\$
Recall $z(K_{n,n}, D) = \min\{n, |\supp(D)|/p\}$. When $|\supp(D)| \ge np$, we have $z(K_{n,n}, D) = n$ and thus $L(G,D) \triangleq z(K_{n,n}, D) - z(G,D) \ge \sum_{i=1}^n I_i$. Then fixing any small constant $\delta>0$, we have
\#\label{eq:flex-lb} 
L(G) = \E_D[L(G,D)] &= \sum_{k=0}^{n} \E[L(G,D)\given |\supp(D)| = k] \cdot \prob{|\supp(D)| = k} \notag\\
& \ge \sum_{k=np}^{n} \E[L(G,D)\given |\supp(D)| = k] \cdot \prob{|\supp(D)| = k} \notag \\
& \ge \sum_{i=1}^n \sum_{k=np}^{n} \E[I_i\given |\supp(D)| = k] \cdot \prob{|\supp(D)| = k}.
\#
%Fix $\Delta = n^{(1+\eta)/2}$.
For any $i\in [n]$, the event $I_i=1$ occurs if and only if no neighbor of supply $i$ belongs to $\supp(D)$. Thus, for any $np \le k \le n$, we have
\$
\E[I_i\given |\supp(D)| = k] &= \prob{I_i = 1\given |\supp(D)| = k} = \frac{\binom{n-d_i}{k}}{\binom{n}{k}},
%\\& = \frac{\binom{(n-k)d_i}{d_i}}{\binom{nd_i}{d_i}} \ge \left(\frac{(n-k-1)d_i}{nd_i}\right)^{d_i} \ge \left(q - \frac{\Delta +1 }{n}\right)^{d_i} \ge q^{d_i} e^{-4n^{-(1-\eta)/2}d_i/q}, \notag
\$
with the convention that $\binom{n-d_i}{k}=0$ for $k>n-d_i.$
%For $k> n-d_i$, we have $\prob{I_i = 1\given |\supp(D)| = k} = 0$. 
%since the demand set $\supp(D)$ is sufficiently large that at least one neighbor of supply $i$ must be included in it.
Thus, 
\begin{align}\label{eq:flex-lb-i} 
\sum_{k=np}^{n} \E[I_i\given |\supp(D)| = k] \cdot \prob{|\supp(D)| = k} & = \sum_{k=np}^{n} \frac{\binom{n-d_i}{k}}{\binom{n}{k}} \cdot \binom{n}{k}p^kq^{n-k} \notag\\
&=\sum_{k=np}^{n} \binom{n-d_i}{k}p^kq^{n-k}.
\end{align}
Substituting~\prettyref{eq:flex-lb-i} into~\prettyref{eq:flex-lb} and changing the summations, we get that 
\begin{align}
L(G)
&\ge \sum_{k=np}^{n}  \sum_{i=1}^n \binom{n-d_i}{k}p^kq^{n-k} \ge  n \sum_{k=np}^{n} \binom{n-d}{k}p^kq^{n-k} = n q^{d} \prob{\Binom(n-d, p) \ge np}, \label{eq:convex_lb}
\end{align}
%\nbr{check whether there is a simple non-asymptotic lower bound.. Do not find good one. But the following should already work.}
where the second inequality holds by applying Jensen's inequality to the convex function $\binom{x}{k} \indc{x\ge k}$, with equality if and only if $d_i \equiv d$ for all $i$.
%& = q^{d_i} \sum_{k=np}^{n} \binom{n-d_i}{k}p^kq^{n-d_i-k} \notag\\
%& = q^{d_i} \prob{\Binom(n-d_i, p) \ge np}.
%\#
The following Berry--Esseen bound establishes the concentration of $X\sim\Binom(n-d, p)$ \citep{dasgupta2008asymptotic}, 
\$ 
\sup_{x\in\R}\left|\prob{\frac{X - (n-d)p}{\sqrt{(n-d)pq}} \le x} - \Phi(x)\right| \le \frac{c(
1-2pq)}{\sqrt{(n-d)pq}},
\$
%\nbr{add references for Berry-Esseen. Also, how did you get $|1-2p|?$ \cite[Example 11.1]{dasgupta2008asymptotic}}
where $\Phi(\cdot)$ is the CDF for standard Gaussian distribution and $c>0$ is a constant. Therefore, %\nbr{I think you can make $c$ more explicit...} is a constant depending only on $p$. By applying the bound, we have
\#\label{eq:flex-lb-conc}
%\sum_{k=np}^{n-d_i} \binom{n-d_i}{k}p^kq^{n-d_i-k} = \prob{\Binom(n-d_i, p) \ge np} 
\prob{X \ge np} \ge 1- \Phi\left(\frac{d p}{\sqrt{(n-d)pq}}\right) - \frac{c}{\sqrt{n-d}} \ge \frac{1}{2}-\frac{d p}{\sqrt{2\pi(n-d)pq}}  - \frac{c(
1-2pq)}{\sqrt{(n-d)pq}},
\#
%\nbr{we can use the bound $\Phi(x) \le 1/2 + x/\sqrt{2\pi}.$}
where the last inequality is due to $\Phi(x) \le 1/2 + x/\sqrt{2\pi}$. %holds for $d=o(\sqrt{n})$, as ${d p}/\sqrt{(n-d)pq}\to 0$,
%= O(n^{-\eta})\to 0$, $\Phi(0) = 1/2$, and $c/\sqrt{n-d_i}\to 0.$
%d_i < n^{1/2-\eta}$ and sufficiently large $n$, noting that ${d_i p}/\sqrt{(n-d_i)pq} = O(n^{-\eta})\to 0$, $\Phi(0) = 1/2$, and $c/\sqrt{n-d_i}\to 0$. 
Substituting 
%\eqref{eq:flex-lb-i} and 
\eqref{eq:flex-lb-conc} into \eqref{eq:convex_lb}, we have
\begin{equation*}
L(G) \ge  \left(\frac{1}{2}-\frac{d p}{\sqrt{2\pi(n-d)pq}}  - \frac{c(
1-2pq)}{\sqrt{(n-d)pq}}\right) nq^{d}. 
%\ge \sum_{i=1}^n q^{d_i}/3 \ge n q^d/3,
\qedhere
\end{equation*}
%where the second inequality holds by applying Jensen's inequality to the convex function $q^{x}$.
\iffalse
\$ 
\E[L(G,D)] &= \E[L(G,D)\cdot \indc_{|\supp(D)|\ge np}] + \E[L(G,D)\indc_{|\supp(D)|< np}]\\
& \ge \sum_{i=1}^n \E\left[ I_i\cdot \indc_{|\supp(D)|\ge np}\right] \\
& = \sum_{i=1}^n \E\left[\frac{\binom{(n-|\supp(D)|)d_i}{d_i}}{\binom{nd_i}{d_i}}\cdot \indc_{|\supp(D)|\ge np}\right] \\
& \ge \sum_{i=1}^n \left(1-p - c/\sqrt{n}\right)^{d_i} \pr( n(p+c/\sqrt{n})\ge|\supp(D)|\ge np)\\
&\approx \sum_{i=1}^n q^{d_i}/2 \ge nq^{d}/2,
\$
\fi
%\nb{JX. I do not understand the second equality. Please clarify.} \nb{Also, I do not think we can use $\approx$. Need to make that step precise.} where the second inequality follows from the concentration of $|\supp(D)|$, and the last inequality holds by applying Jensen's inequality to the convex function $q^{x}.$
\end{proof}
\subsubsection{\blue{Suboptimality of \ER and Probabilistic Expanders}}

Lower bounds based on isolated supply nodes also yield necessary degree conditions for both \ER and probabilistic-expander designs, showing that neither can attain the sharp constant achieved by random regular graphs. Let $d_{\mathrm{RR}}$, $d_{\mathrm{ER}}$, and $d_{\mathrm{PE}}$ denote the average degrees required by random regular, Erd\H{o}s--R\'enyi, and probabilistic-expander designs, respectively, to achieve $\epsilon$-fractional loss.

For an \ER graph with edge probability $d/n$, a fixed supply node, with degree distribution $X \sim \text{Bin}(n, d/n) \approx \text{Pois($d$)}$, has no edge to any positive demand node with probability 
\[
\mathbb{E}_{X \sim \text{Pois($d$)}}\left[q^{X}\right] = e^{-(1-q)d}.
\]
Thus, the expected number of isolated supplies is asymptotically $ne^{-(1-q)d}$. Since isolated supplies provide a lower bound on the loss, any Erdős–Rényi graph achieving $\epsilon$-fractional loss must satisfy $d_{\mathrm{ER}} \ge (1-o_{\epsilon}(1))\log(1/\epsilon)/(1-q)$. Relative to the sharp random-regular threshold, this gives
\[
\frac{d_{\mathrm{ER}}}{d_{\mathrm{RR}}}
\ge
\frac{\log(1/q)}{1-q}>1, \quad \text{for any } q\in(0,1).
\]

We can similarly analyze the particular probabilistic expander design proposed by \citet{chen2015optimal}. In its construction, on the supply side, the first half of the edges, formed by uniformly sampling $d/2$ demand nodes for each supply node, resembles a random $(d/2)$-regular design, while the second half, formed by uniformly sampling $d/2$ supply nodes for each demand node, resembles an Erdős–Rényi design with an average degree of $d/2$. Intuitively, its performance should lie between those of Erdős–Rényi and random regular graphs, as confirmed by our numerical results in Figure~\ref{fig:flex-design}. More explicitly, the isolation probability of a fixed supply node is 
\[
\mathbb{E}_{X \sim \text{Pois($d/2$)}}\left[q^{X + d/2}\right] = \left[qe^{-(1-q)}\right]^{d/2}.
\] 
Thus, if $d_{\mathrm{PE}}$ denotes the average degree required by the probabilistic expander to achieve $\epsilon$-fractional loss, controlling the expected number of isolated nodes requires $d_{\mathrm{PE}} \ge (2-o_\epsilon(1))\log(1/\epsilon)/[\log(1/q) + 1-q]$. Consequently,
\[
\frac{d_{\mathrm{PE}}}{d_{\mathrm{RR}}}
\ge \frac{2\log(1/q)}{\log(1/q) + 1-q}>1, \quad \text{for any } q\in(0,1).
\]

\subsection{Key Properties of Random Regular Bipartite Graphs}  \label{sect:flex-prop}
This section presents several key properties of random $d$-regular bipartite graphs. Our analysis relies on the following fact about graphs generated by the configuration model. For any $U \subset [n]$ and $W \subset [n]$, recall that $e(U,W)$ denotes the number of edges between them (counting multiplicities). If $|U| = un$ and $|W| = wn$, then under $G\sim\cG(n,n,d)$, $e(U,W)
\sim \text{Hypergeometric}(nd, wnd, und)$, \ie,
\begin{align}
\probs{e(U,W) = k }{G\sim \cG(n,n,d)} = \frac{\binom{wnd}{k}\binom{(1-w)nd}{und-k}}{\binom{nd}{und}}, \quad \forall 0 \le k\le \min\{u,w\} \cdot nd. \label{eq:hypergeometric}
\end{align}

To see why the above expression holds, note that $e(U,W)=k$ if and only if there are $k$ half-edges incident to $W$ and $und-k$ half-edges incident to $[n]\setminus W$ that are paired with half-edges of $U.$  Thus, for the numerator, we first choose $k$ half-edges among $wnd$ half-edges incident to $W$, and there are $\binom{wnd}{k}$ different choices.  Then we choose $und-k$ half-edges among $(1-w)nd$ half-edges incident to $[n]\setminus W$, and there are $\binom{(1-w)nd}{und-k}$ different choices. Finally, these selected $und$ half-edges in total are paired with half-edges incident to $U.$ For the denominator, 
$\binom{nd}{und}$ counts all the possible selections of $und$ half-edges that are paired with these half-edges incident to $U$, without restriction.
%\nbr{check the effect of multi-edges.... The above expression does not consider the effect of multiple edges. It is possible that the selected $k$ half-edges form parallel edges, causing $e(U,W) < k$ if we do not count edge multiplicity. I think this will only affect the proof of~\prettyref{lmm:flex-reg2}.}

Based on the fact that $\text{Hypergeometric}(N,K,n)$ is dominated by $\Binom(n,K/n)$ in convex order, \cite{hoeffding1963probability} proved the following tail bounds for hypergeometric random variables.
\begin{lemma}\label{lem:hypergeo}
Let $X\sim\text{Hypergeometric}(N,K,n)$ and $p=K/N$. Then for any $0< t< p$, the following inequalities hold:  
\begin{align*}
\prob{X \le (p-t) n }
& \le \exp\left( - n \cdot D_{\mathrm KL}( p-t \| p )\right), \\
\prob{X \ge (p+t) n }
& \le \exp\left( - n \cdot D_{\mathrm KL}( p+t \| p )\right),  
\end{align*}
where the Kullback--Leibler divergence is defined as $D_{\mathrm KL}( a\| b ) = a\log(a/b) + (1-a)\log((1-a)/(1-b))$. 
\end{lemma}

Next, we present three key lemmas, all of which concern the subgraph of the random regular bipartite graph induced by a fixed demand subset $V$. Let $h(x) = x\log(1/x) + (1-x)\log(1/(1-x))$ be the binary entropy function. The following lemma establishes the \emph{no-big-cut} property: any two reasonably large sets $S$ and $T$ do not form a cut (\ie, $e(S,T)=0$) with high probability.

\begin{lemma}[No Big Cut]\label{lmm:flex-reg1} 
For any $V \subset \nd$ with $|V|=np$ and any $s,t  \in (0,1)$ with $s+t\ge 1$, if $d t\log (1/t) \ge 4h(t)$ and $d (1-s)\log (1/(1-s)) \ge 2h(1-s)/p$, 
then a random $d$-regular graph $G\sim \cG(n,n,d)$ satisfies
\$
& \probs{\exists S\subset \ns \text{ with } |S| = sn, \exists T\subset V \text{ with } |T| = tp n\colon e(T,S) = 0 }{G}
\le \exp\left\{ - \frac{d p nt}{4}\log \frac{1}{t} \right\}.
\$
\end{lemma}
\begin{proof}
Fix any $S \subset [n]$ with $|S| = s n$, and $T \subset V$ with $|T| = t p n$. In view of~\prettyref{eq:hypergeometric}, we have
\#\label{eq:reg-d-t}
\prob{e(T,S) = 0 } =  \frac{\binom{(1-s)nd}{tp n d}}{\binom{nd}{tp n d}} \le (1-s)^{tp n d},
\#
where the inequality holds due to 
$\binom{m}{k}/\binom{n}{k} \le (m/n)^k.$
Then a union bound over both $S$ and $T$
gives
\begin{align*}
&\probs{\exists S\subset \ns \text{ with } |S| = sn, \exists T\subset V \text{ with } |T| = tp n\colon e(T,S) = 0 }{G} \notag\\
& \qquad \le \binom{n}{sn} \binom{p n}{tp n} (1-s)^{tp n d}
\notag\\
&  \qquad\le \exp\left\{n\left[h(s) + p  h(t) - dp t\log \frac{1}{1-s} \right]\right\} \\
& \qquad \le \exp\left\{n\left[h(s) + p  h(t) - dp \left(\frac{t}{2}\log \frac{1}{t} + \frac{1-s}{2}\log \frac{1}{1-s} \right)\right]\right\} \le \exp\left\{ - \frac{dp nt}{4}\log \frac{1}{t} \right\},
\end{align*}
where the second inequality follows from 
$\binom{n}{k} \le \exp(nh(k/n))$, the third inequality holds since $t \ge 1-s$ by assumption, and the last inequality follows from $d t\log (1 /t )\ge 4h(t)$ and $d (1-s)\log (1 /(1-s) )\ge 2h(1-s)/p = 2h(s)/p$. 
\end{proof}

The next lemma shows the \emph{little-waste-of-supply} property: every supply set $S$ whose size  $sn$ is not too small---specifically, for $s$ lower bounded in terms of $d$ as stated in the lemma---%\nbr{In the lemma below, $s$ can be close to $1$?} \xn{[Yes.]} 
receives a sufficiently large number of demand edges, ensuring that little supply is wasted.

 %[Potentially relate to Step 3, 5]
% \xn{We should highlight this is the key lemma that requires the tight $d$.}
% \nbr{JX. Yes, that's right. In fact, if we take $S$ to be the set of isolated nodes and $u=\epsilon$, the lemma requires no existence of $\epsilon n$ isolated nodes, which will lead to the condition $d \ge \log (1/\epsilon)/\log(1/q).$ }

% \nb{Yehua: little supply is wasted?}
% \nbr{Check. One side regularity graph is enough???}

\begin{lemma}[Little Waste of Supply]
\label{lmm:flex-reg2}
For any subset $V$ with $|V| =np$, any $\delta \in (0,1)$ and $\beta \in (0,p)$ such that $
\beta \le \delta/8$
and $\beta \log (e/\beta) \le (\delta/8) \log (1/(1-p))$,  
and any $s \in (0,1)$ such that $\log(e/s) < d (1-\delta/4)\log(1/q)$, a random $d$-regular graph $G\sim \cG(n,n,d)$ satisfies 
\$
\probs{ \exists S\subset \ns \text{ with } |S| = sn\colon  e(S,V) \le \beta d|S|} {G} \le \exp\left\{-n s \left[ d (1-\delta/4) \log \frac{1}{q} - \log \frac{e}{s} \right]\right\}.
\$
\end{lemma}
\begin{proof}
For any $S \subset [n]$, 
recall from~\prettyref{eq:hypergeometric} that $e(S,V)$ follows $\text{Hypergeometric}(nd, p nd, snd)$, which is dominated by $\Binom(snd, p)$ in convex order. Thus, it follows from the Chernoff bound for binomial distribution (Lemma \ref{lem:hypergeo}) that for any $\beta < p$,
\#\label{eq:reg-d-s1}
\prob{e(S,V) \le \beta snd } \le \exp\left( - snd \cdot D_{\mathrm KL}( \beta \| p )\right), 
\#
where $D_{\mathrm KL}( \beta \| p ) = \beta\log(\beta/ p) + (1-\beta)\log((1-\beta)/q)$. 
Then a union bound over $S$ yields %that
\begin{align*}
\prob{\exists S \subset [n] \text{ with } |S|=sn: e(S,V) \le \beta d s n }
& \le \binom{n}{sn}\exp\left( - snd \cdot D_{\mathrm KL}( \beta \| p )\right) \\
& \le \exp\left( - sn\left(  d \cdot D_{\mathrm KL}\left( \beta \| p \right)   -  \log (e/s)\right)\right),
\end{align*}
where the last inequality holds due to $\binom{n}{k} \le (en/k)^k$. We conclude the proof by observing that 
$$
D_{\mathrm KL}( \beta \| p ) \ge   \beta \log \beta - \beta  + (1-\beta) \log (1/q)
\ge (1-\delta/4) \log (1/q),
$$
where the first inequality holds because $p \le 1$
and $(1-\beta)\log(1-\beta) \ge -\beta$, and the last inequality holds by the assumption on $\beta.$
\end{proof}

The following lemma shows the \emph{no-edge-congestion} property: with high probability, any small subsets $S$ and $W$ cannot have too many connecting edges, provided that their sizes $sn$ and $w p n$ are upper bounded by the conditions in the lemma.
\begin{lemma}[No Edge Congestion]\label{lmm:flex-reg3} 
Fix a subset $V$ with $|V| = np$. For any $d>0$ and $\beta,s,w \in (0,1)$ satisfying $d\beta \ge 8$, $\beta \log(1/s) \ge 4\log 2+8$, and  $w \le \min\{\beta, s\}/p,$ a random $d$-regular graph $G\sim \cG(n,n,d)$ satisfies 
\$
\probs{\exists S\subset \ns,  W\subset V \text{ with } |S| = sn, |W|=w|V|\colon e(S, W) \ge \beta d|S| }{G}\le \exp\left\{- \frac{\beta nd s}{2} \log \frac{1}{s}\right\}.
\$
\end{lemma}
\begin{proof}
For $S$ with $|S| = sn$ and $W$ with $|W| = w p n  \le \beta n$, recall from \eqref{eq:hypergeometric} that $e(S,W)\sim \text{Hypergeometric}(nd, w p nd, snd)$. 
Thus, the Chernoff bound in Lemma \ref{lem:hypergeo} gives that
\#\label{eq:reg-d-s2} 
\prob{e(S,W) \ge \beta snd } \le \exp\left( - snd \cdot D_{\mathrm KL}( \beta \| wp )\right),
\#
Then a union bound gives that
\$
&\prob{\exists S\subset [n], W \subset V \text{ with } |S| = s n, |W| = w p n:  e(S, W) \ge \beta s nd } \\
& \qquad\le \binom{n}{s n}\binom{p n}{wp n} \exp\left( - snd \cdot D_{\mathrm KL}( \beta \| wp )\right) \\
& \qquad \le \exp\big\{- n\big[s d \cdot D_{\mathrm KL}( \beta \| wp )- s \log (e/s) - p w \log (e/w) \big]\big\}.
\$  
Moreover, since $\beta \log\beta +(1-\beta) \log (1-\beta) \ge -\log 2$, it follows that 
$D_{\mathrm KL}( \beta \| wp )=\beta \log [\beta/(wp)]+ (1-\beta) \log [(1-\beta)/(1-wp)] \ge  \beta \log (1/wp) -\log 2$. Therefore, 
\$ 
& s d \cdot D_{\mathrm KL}( \beta \| wp )- s \log (e/s) - p w \log (e/w) \\
& \qquad \ge sd\left(\beta\log\frac{1}{wp} - \log 2\right) - s \log \frac{e}{s} - p w \log \frac{e}{w}\\
& \qquad \ge s \Big[d\Big(\beta\log\frac{1}{s} - \log 2 \Big) - 2\log\frac{e}{s} \Big]\ge \frac{1} {2} \beta ds \log\frac{1}{s},
\$
where the second inequality holds since the function $x\log(e/x)$ is increasing in $x\in(0,1)$ and $w p \le s$, and the last inequality follows from the assumptions that  $\beta\log(1/s)  \ge 4\log 2+8$ and $d\beta \ge  8$. Finally, we combine the above results and conclude the proof.
\end{proof}

\begin{remark}\label{rmk:properties}
Among the above three lemmas, 
\prettyref{lmm:flex-reg1} exploits the degree-$d$ regularity of the demand nodes, through the term $|T|d$ in the bound \eqref{eq:reg-d-t}, where $T$ is a subset on the demand side. The other two lemmas use the degree-$d$ regularity of the supply nodes through the term $|S|d$ in the bounds \eqref{eq:reg-d-s1} and \eqref{eq:reg-d-s2}, where $S$ is a subset on the supply side.

In particular, the \emph{little-waste-of-supply} property established in~\prettyref{lmm:flex-reg2} is the most critical, since it imposes the sharp condition on $d$ given by~\eqref{eq:degree-const-require}.
To illustrate, if we take $S$ to be the set of ``isolated'' supply nodes (with no connections to $V$) in~\prettyref{lmm:flex-reg2}, then ruling out the existence of  $\epsilon n$ isolated nodes requires the condition $d \ge \log (1/\epsilon)/\log(1/q)$, as shown in the proof of lower bound in~\prettyref{thm:lower_bound_bp}. 
Note that if we were using the \ER random graphs instead of random $d$-regular graphs, then to rule out the existence of  $\epsilon n$ isolated nodes and ensure little waste of supply, we would need a larger average degree, that is, $d \ge \log (1/\epsilon)/(1-q)$. This is because the expected number of isolated supply nodes in
 \ER random graphs is approximately $n \cdot \prob{\Pois(dp)=0}= n \exp(-dp)$ with $p=1-q.$
\end{remark}
\subsection{Proof of \prettyref{lmm:epsilon_connectivity}}\label{sect:flex-pf-con-fract} 
In this section, we prove Lemma~\ref{lmm:epsilon_connectivity} using the key properties established above.

\iffalse
\begin{proposition}\label{prop:flex-const}
For any constant $\delta>0$, there exists a constant $\epsilon(\delta)>0$ such that for any $0<\epsilon<\epsilon(\delta)$ %there exists $N(\epsilon)>0$ such that for any $n > N(\epsilon)$ 
and any demand realization $D$ with $|\supp(D)| \ge (1-\epsilon/2)np$, random $d$-regular graphs with $d$ given by \eqref{eq:d_cond_eps_flex} satisfy 
\$
\probs{z(G,D) \ge (1-\epsilon) n}{G\sim\mathbb{G}(d)} \ge 1- 3n\exp\left(-\delta n\epsilon\log(1/\epsilon)/8 \right).
\$
\end{proposition}  
%\nb{WJ:Can the specific expression of $\epsilon(\delta)$ here be written directly? Otherwise, it is not clear that it is a small value.} \nbr{I think if $\epsilon(\delta)$ is bigger, then the range for $\epsilon$ is bigger, which is in fact better.}

\begin{lemma}\label{lem:concentrate-V}
For any $n \ge 3\log(4/\gamma) / (\eta^2p)$, it holds that
\$
\pr\Big(\abs{|\supp(D)|-np} \le \eta np \Big) \ge 1-\frac{\gamma}{2}.
\$
\end{lemma}
\fi    
%\begin{lemma}\label{lem:pre-sub} Fix any graph $G$. For realizations of demand $V_1\subset V_2$, we have $z(G, V_1) \le z(G, V_2)$. \end{lemma}   
%Following from the proof sketch in Section \ref{sect:pf-sketch-proce}, it suffices to establish the three key statements in \eqref{eq:no_waste}-\eqref{eq:flex-d-large-s-frac}.
\begin{proof}[Proof of \prettyref{lmm:epsilon_connectivity}]
We apply Lemmas \ref{lmm:flex-reg1}-\ref{lmm:flex-reg3} and verify their conditions to show \eqref{eq:flex-d-large-s-frac}-\eqref{eq:flex-d-no-cong}.
\paragraph{Proof of \eqref{eq:flex-d-large-s-frac} via Lemma \ref{lmm:flex-reg1}.}
Choose small constants $\epsilon_0,c\in(0,1/4)$ such that $$
\sqrt{\epsilon_0}\log(1/\epsilon_0) \le 1/2,  \quad \log(1/\epsilon_0) \ge 8\log(1/q)\max\{1/p, 1+2\log 2\},
\quad c\le p/[32\log(1/q)].
$$
Fix any $\epsilon\in(0, \epsilon_0)$. These choices also imply $\epsilon^c \ge 2\epsilon$, since $\epsilon,c\le1/4$ give $\epsilon^{1-c} \le \sqrt{\epsilon} \le \sqrt{1/4}=1/2$.

A supply subset $S$ is a bottleneck, \ie, $|N(S)\cap V| < p(|S|- \epsilon n)$, if and only if there exists a demand set $T\subset V$ with $|T| = |V| - p(|S|- \epsilon n)$ such that $e(S,T)=0$. Here $T$ can be viewed as the complement of $N(S)$ in $V$. Then we have
\#\label{eq:ssss1}
& \prob{\exists S \subset [n] \text{ with } |S| \ge \epsilon^c n\colon |N(S)\cap V| < p(|S|- \epsilon n)} \\
& \qquad \le \sum_{sn=\epsilon^c n}^n \prob{\exists S \subset [n] \text{ with } |S| = sn, \exists T \subset V \text{ with } |T| = |V| - p(s- \epsilon)n\colon e(S,T) = 0}. \notag
\#
For $s\in[\epsilon^c, 1]$, define $t(s) \triangleq |T|/|V| = 1 - s + \epsilon$. Hence, we have $t(s)\ge 1-s$ and $t(s)\in[\epsilon, 1-\epsilon^c+\epsilon] \subset 
[\epsilon, 1-\epsilon^c/2]$ using $\epsilon^c \ge 2\epsilon$. Applying \prettyref{lmm:flex-reg1} to \eqref{eq:ssss1} for each $s\in[\epsilon^c, 1]$ and $t = t(s)$, 
\#\label{eq:sss1}
& \prob{\exists S \subset [n] \text{ with } |S| \ge \epsilon^c n\colon |N(S)\cap V| < p(|S|- \epsilon n)} \notag\\
& \qquad \le \sum_{sn=\epsilon^c n}^n \exp\left\{ - \frac{d p nt(s)}{4}\log \frac{1}{t(s)} \right\}  \le n\exp\left\{ - \frac{dp n\epsilon}{4}\log \frac{1}{\epsilon} \right\},
\#  
where the last inequality uses the fact that $g(t) \triangleq t\log (1/t) \ge g(\epsilon)$ for all $t \in [\epsilon, 1-\epsilon^c/2]$. To see this, note that $g$ is concave, so we only need to compare the endpoints. Moreover, 
\$g(1-\epsilon^c/2) \ge (1-\epsilon^c/2)\epsilon^c/2 \ge \sqrt{\epsilon}/4 
\ge g(\epsilon),\$ 
for sufficiently small $\epsilon$, since $c\le 1$. Finally, it remains to verify the lemma's conditions on $d$:
\$ 
d (1-s)\log (1/(1-s)) \ge 2h(1-s)/p, \ d t\log (1/t) \ge 4h(t), \qquad \forall s\in[\epsilon^c, 1], \ t\in [\epsilon, 1-\epsilon^c/2].
\$
Let $f(x) \triangleq h(x)/[x\log(1/x)]$. Since $f(x)$ is increasing on $(0,1)$, it suffices to verify the conditions at the corresponding endpoints, \ie, $ 
d \ge 2f(1-\epsilon^c)/p$ and $d \ge 4f(1-\epsilon^c/2)$.
Note that $f(1-y) = 1 + y\log(1/y)/[(1-y)\log(1/(1-y))] 
\le 1+ 2\log(1/y)$ for all $y\le 1/2$, since $y \le \log(1/(1-y))$ for all $y\ge 0$. Thus, we only need to ensure
\$ 
d \ge 2[1+ 2\log(1/\epsilon^c)]/p, \quad d \ge 4[1+ 2\log(2/(\epsilon^c))].
\$
Because $d$ is of order $\log(1/\epsilon)$, these conditions hold for sufficiently small $c$. In fact, our choices of $\epsilon,c$ ensure that $d > \log(1/\epsilon)/\log(1/q)  \ge \max\{4/p, 8\log(1/\epsilon^c)/p, 8(1+2\log 2), 16\log(1/\epsilon^c)\}$, which implies the above conditions.

\paragraph{Proof of \eqref{eq:flex-d-no-cong} via Lemma \ref{lmm:flex-reg3}.}
Given the constant $c$ fixed above, we choose 
$$
\gamma = \max\{8, (8\log 2+16)/[c\log(1/q)]\}, \quad \beta =  \gamma/d, \quad  w = p(s- \epsilon),
$$
and verify the conditions in \prettyref{lmm:flex-reg3} 
for all $s \in [\epsilon,\epsilon^c]$. First, since $s\le \epsilon^c$, we have 
\begin{align*}
\beta \log(1/s) \ge \beta c\log(1/\epsilon) = \gamma c\log(1/\epsilon)/d \ge  \gamma c\log(1/q) /2 \ge 4\log 2+8, 
\end{align*}
where we will later ensure $\delta\le1$, so that $d$ in \eqref{eq:ddddelta} satisfies $d\le2\log(1/\epsilon)/\log(1/q)$, and the last inequality follows from the definition of $\gamma$. Next,
\begin{align*}
w p = p^2(s-\epsilon)\le s \le  \epsilon^c \le \gamma\log(1/q)/[2\log(1/\epsilon)] \le \gamma/d = 
\beta,
\end{align*}
where the third inequality uses $\epsilon^c\log(1/\epsilon^c) \le 1\le \gamma c\log(1/q)/2$, and the last inequality again uses $d\le2\log(1/\epsilon)/\log(1/q)$.
Thus, \prettyref{lmm:flex-reg3} yields that
\#\label{eq:sss2} 
& \prob{\exists S \subset [n] \text{ with } \epsilon n\le |S| \le \epsilon^c n, \exists T\subset V \text{ with } |T| = p(|S|- \epsilon n) \colon e(S,T) > \gamma|S|} \notag\\
& \qquad \le \sum_{sn = \epsilon n}^{\epsilon^c n}\prob{\exists S \subset [n] \text{ with } |S| =s n, \exists T\subset V \text{ with } |T| = p(s- \epsilon )n \colon e(S,T) > \gamma|S|} \notag\\
& \qquad \le \sum_{sn = \epsilon n}^{\epsilon^c n}  \exp\left\{- \frac{\gamma n s}{2} \log \frac{1}{s}\right\} \le n \exp\left\{- \frac{\gamma n \epsilon}{2} \log \frac{1}{\epsilon}\right\}.
\#

\paragraph{Proof of \eqref{eq:no_waste} via Lemma \ref{lmm:flex-reg2}.} 
Given the constant $\gamma$ chosen above, we set
$$
c_1 = 16\gamma\log(1/q)/\log(1/(1-p/2)), \quad \delta = c_1 \log\log(1/\epsilon)/\log(1/\epsilon).$$
Choose $\epsilon_0$ sufficiently small to further satisfy that $\beta\le 1/e$ and $\delta\le 1$ for all $0<\epsilon\le \epsilon_0$. Recall that $\beta = \gamma/d$ is of order $1/\log(1/\epsilon)$, and hence $\beta\log(1/\beta)$ is of order $\log \log(1/\epsilon)/ \log(1/\epsilon)$, the same order as $\delta$. Thus, for $c_1$ large enough---indeed, for our specific choice above---the conditions of \prettyref{lmm:flex-reg2} hold: 
$$
\beta \le \delta/8, \text{ and } \beta \log (e/\beta) \le (\delta/8) \log (1/q).
$$
Moreover, when $d$ is given by \eqref{eq:ddddelta}, we have 
\#\label{eq:app-d-con} 
d(1-\delta/4)\log \frac{1}{q}  & \ge (1+\delta) \left( 1-\delta/4\right) \log \frac{1}{\epsilon} \ge (1+\delta/4)\log \frac{1}{\epsilon}.
\#
Thus, \prettyref{lmm:flex-reg2}, with a union bound over $|S|$, yields that
\#\label{eq:sss3} 
\prob{\exists S \subset [n] \text{ with } |S| \ge \epsilon n \colon e(S,V) \le \gamma|S|} &\le \sum_{sn = \epsilon n}^{n} \exp\left\{-n s \left[ d (1-\delta/4) \log \frac{1}{q} - \log \frac{e}{s} \right]\right\} \notag\\
&\le \sum_{sn = \epsilon n}^{n} \exp\left\{-n s \left[ \left(1+{\delta}/{4}\right)\log \frac{1}{\epsilon} - \log \frac{e}{s} \right]\right\} \notag\\
&
\le n\exp (- \delta n \epsilon\log(1/\epsilon)/8),
\#
where the second inequality holds due to \eqref{eq:app-d-con} and the last inequality holds since $\phi(s) = -s [ (1+{\delta}/{4})\log ({1}/{\epsilon}) - \log ({e}/{s})]$ decreases on $s\in(\epsilon, 1)$, and achieves the maximum at $s=\epsilon$. 

Finally, by taking the maximum of the probabilities in \eqref{eq:sss1}, \eqref{eq:sss2}, and \eqref{eq:sss3} that $\zeta = n\exp (- \delta n \epsilon\log(1/\epsilon)/8)$, we complete the proof of the lemma.
\end{proof}

\subsection{\blue{Proofs of Lemmas \ref{lmm:lt-convex} and \ref{lem:monotonicity}}}\label{sec:general_demand_proofs}
This subsection proves the two key ingredients used in Section~\ref{sec:flex-bound}. We first establish the convexity reduction on each fixed-total-demand slice and then prove the monotonicity properties used to compare the resulting extreme points.
\begin{proof}
[Proof of Lemma \ref{lmm:lt-convex}]
Let $\phi_a(u) = (u-a)^+$. Then $F_a(D) = \E_{G\sim\cG(n,n,d)}[\phi_a(L(G,D))]$. 
Fix a graph $G$, and let $N_G(S)$ denote the set of demand nodes adjacent to a supply-node set $S$. For any demand vector $D\in\mathbb R^n$, the max-flow min-cut theorem gives
\begin{equation}
z(G,D)
=
\min_{S\subseteq[n]}
\left\{
n-|S|+\sum_{j\in N_G(S)}D_j
\right\}.
\label{eq:mfmc-zg-general}
\end{equation}
In particular, $z(K_{n,n},D)=\min\{n,\sum_{j=1}^nD_j \}$.
Fix a graph $G$ and a vector $D\in\Sigma_t$. Using \eqref{eq:mfmc-zg-general} and $\sum_jD_j=t$, we obtain
\begin{align*}
L(G,D)
&=
\min\left\{n,\sum_jD_j\right\}
-
\min_{S\subseteq[n]}
\left\{
n-|S|+\sum_{j\in N_G(S)}D_j
\right\}\\
&=
\max_{S\subseteq[n]}
\left\{
|S|-n+\sum_{j\notin N_G(S)}D_j
\right\}
-\max\{t-n,0\}.
\end{align*}
This is the maximum of finitely many affine functions of $D$, minus a constant (for fixed $t$), and hence convex on $\Sigma_t$. Since $\phi_a$ is convex and nondecreasing, the composition
$\phi_a(L(G,\cdot))$ is also convex on $\Sigma_t$. Averaging over $G$ preserves convexity, so
$F_{a}$ is convex on $\Sigma_t$.

A convex function on the compact polytope $\Sigma_t$ attains its maximum at an extreme
point. If an extreme point $D$ had two distinct coordinates in $D_i, D_j\in (0,1/p)$, then for all sufficiently small $\epsilon>0$, both vectors $D\pm \epsilon(e_i-e_j)$
would lie in $\Sigma_t$, and have midpoint $D$, contradicting extremality. Hence, every
extreme point has at most one coordinate in $(0,1/p)$ and all others in $\{0,1/p\}$. Since the coordinates sum to $t$, every extreme point is a permutation of $D^\star(t)$.

Finally, the configuration-model distribution $\cG(n,n,d)$ is invariant under relabeling of the demand nodes, so $F_{a}$ is permutation invariant on $\Sigma_t$. Therefore, $F_{a}$ takes the same value at all extreme points, and $F_{a}(D)\le F_{a}(D^\star(t))$ for every $D\in\Sigma_t$.  
\end{proof}

\begin{proof}[Proof of \prettyref{lem:monotonicity}]
We first prove part (a).
Let $D' = D + re_j$. From $z(K_{n,n}, D) = \min\{n, \sum_j D_j\}$ and $z(K_{n,n}, D') = \min\{n, \sum_j D'_j\}$, we have  $z(K_{n,n}, D') \le z(K_{n,n}, D)+r$. Moreover, $D'\ge D$ coordinate-wise and $z(G, \cdot)$ is nondecreasing in the demand vector, so $z(G, D') \ge z(G,D)$. Combining them, we have
\$ 
L(G, D') \le z(K_{n,n}, D)+r - z(G, D) = L(G, D) + r \le L(G,D) + 1/p.
\$

We now prove part (b).
The proof follows from the simple observation that 
\begin{align}
z\left(K_{n,n}, D\right) - z\left(K_{n,n}, D'\right) = \sum_{j \in [n]}(D_j-D'_j), \label{eq:mono-down1} \\   
z\left(G, D\right) - z\left(G, D'\right) \leq \sum_{j \in [n]}(D_j-D'_j),  \label{eq:mono-down2} 
\end{align}
where \eqref{eq:mono-down1} immediately follows from the definition of the full flexibility network. For \eqref{eq:mono-down2}, by max-flow min-cut \eqref{eq:mfmc-zg-general}, $z(G, D') = n-|S'| + \sum_{j \in N_G(S')} D'_j$
for some $S' \subseteq [n]$, and 
$z(G, D) \leq n-|S'| + \sum_{j \in N_G(S')} D_j$, implying that
\begin{align*}
z(G, D) - z(G, D') \leq \sum_{j \in N_G(S')} (D_j - D'_j)\leq \sum_{j \in [n]}(D_j-D'_j).
\end{align*}
Thus, we have \eqref{eq:mono-down2}.
Subtract \eqref{eq:mono-down2} from 
\eqref{eq:mono-down1} and the lemma follows.
\end{proof}

\subsection{\blue{Additional Numerical Results}}
\label{sec:additional-flex-comparison}
\subsubsection{Tables for Experiments in Section \ref{sect:exp}}
To better illustrate the practical magnitude of the performance differences, Tables \ref{tab:original-fractional-loss} and \ref{tab:transport-fractional-loss} report fractional losses on their original scale, for all six experimental settings in Figures \ref{fig:flex-design} and \ref{fig:trans-design}. 
Across the settings, we observe that the random regular design often achieves substantially lower fractional loss than competing designs.

%its loss is less than half that of the best alternative and remains close to the theoretical lower bound. 

%These performance gains have important practical implications. 
%A network with average degree $d$ contains $dn$ edges in process flexibility and $dn/2$ edges in middle-mile transportation. Thus, even a modest reduction in the degree required to achieve a target loss can eliminate many edges in a large system. Because each edge may entail setup, qualification, and maintenance costs, this reduction can translate into meaningful implementation savings.

\begin{table*}[htbp]
\centering
\small
\setlength{\tabcolsep}{2pt}
\begin{tabular}{@{}l@{\hspace{9pt}}*{3}{>{\centering\arraybackslash}p{0.078\textwidth}}@{\hspace{14pt}}*{3}{>{\centering\arraybackslash}p{0.078\textwidth}}@{\hspace{14pt}}*{3}{>{\centering\arraybackslash}p{0.078\textwidth}}@{}}
\toprule
& \multicolumn{3}{c@{\hspace{14pt}}}{(a) $n=200$, $p=0.5$}
& \multicolumn{3}{c@{\hspace{14pt}}}{(b) $n=500$, $p=0.5$}
& \multicolumn{3}{c}{(c) $n=400$, $p=0.8$} \\
\cmidrule(l{4pt}r{18pt}){2-4}\cmidrule(l{4pt}r{18pt}){5-7}\cmidrule(l{4pt}r{4pt}){8-10}
Design
& $d=4$ & $d=6$ & $d=8$
& $d=4$ & $d=6$ & $d=8$
& $d=4$ & $d=6$ & $d=8$ \\
\midrule
$K$-chain
& $9.645\%$ & $5.510\%$ & $3.457\%$
& $10.667\%$ & $6.505\%$ & $4.446\%$
& $2.573\%$ & $1.284\%$ & $0.677\%$ \\
Erdős--Rényi
& $12.430\%$ & $3.490\%$ & $0.964\%$
& $13.278\%$ & $3.748\%$ & $1.078\%$
& $4.512\%$ & $0.688\%$ & $0.133\%$ \\
Prob.\ expander
& $8.400\%$ & $1.519\%$ & $0.373\%$
& $8.969\%$ & $1.692\%$ & $0.395\%$
& $0.691\%$ & $0.033\%$ & $0.003\%$ \\
Random regular
& $4.318\%$ & $0.685\%$ & $0.151\%$
& $5.011\%$ & $0.777\%$ & $0.160\%$
& $0.074\%$ & $0.002\%$ & \makebox[\linewidth][c]{$<\!0.001\%$} \\
Lower bound
& $2.595\%$ & $0.562\%$ & $0.120\%$
& $2.790\%$ & $0.642\%$ & $0.147\%$
& $0.059\%$ & $0.002\%$ & \makebox[\linewidth][c]{$<\!0.001\%$} \\
\bottomrule
\end{tabular}
\caption{\small{Fractional losses for process flexibility in the three settings of Figure \ref{fig:flex-design}.}}
\label{tab:original-fractional-loss}
\end{table*}

\begin{table*}[htbp]
\centering
\small
\setlength{\tabcolsep}{2pt}
\begin{tabular}{@{}l@{\hspace{9pt}}*{3}{>{\centering\arraybackslash}p{0.078\textwidth}}@{\hspace{14pt}}*{3}{>{\centering\arraybackslash}p{0.078\textwidth}}@{\hspace{14pt}}*{3}{>{\centering\arraybackslash}p{0.078\textwidth}}@{}}
\toprule
& \multicolumn{3}{c@{\hspace{14pt}}}{(a) $n=200$, $p=0.5$}
& \multicolumn{3}{c@{\hspace{14pt}}}{(b) $n=500$, $p=0.5$}
& \multicolumn{3}{c}{(c) $n=400$, $p=0.8$} \\
\cmidrule(l{4pt}r{18pt}){2-4}\cmidrule(l{4pt}r{18pt}){5-7}\cmidrule(l{4pt}r{4pt}){8-10}
Design
& $d=4$ & $d=6$ & $d=8$
& $d=4$ & $d=6$ & $d=8$
& $d=4$ & $d=6$ & $d=8$ \\
\midrule
$K$-chain
& $6.985\%$ & $3.123\%$ & $1.424\%$
& $7.093\%$ & $3.225\%$ & $1.503\%$
& $1.495\%$ & $0.204\%$ & $0.013\%$ \\
Erdős--Rényi
& $10.217\%$ & $3.756\%$ & $1.167\%$
& $10.645\%$ & $3.785\%$ & $1.127\%$
& $4.239\%$ & $0.713\%$ & $0.171\%$ \\
Random regular
& $6.369\%$ & $1.436\%$ & $0.305\%$
& $6.263\%$ & $1.155\%$ & $0.229\%$
& $0.155\%$ & $0.004\%$ & \makebox[\linewidth][c]{$<\!0.001\%$} \\
Lower bound
& $2.875\%$ & $0.531\%$ & $0.098\%$
& $3.025\%$ & $0.681\%$ & $0.098\%$
& $0.064\%$ & $0.003\%$ & \makebox[\linewidth][c]{$<\!0.001\%$} \\
\bottomrule
\end{tabular}
\caption{\small{Fractional losses for middle-mile transportation in the settings of Figure \ref{fig:trans-design}.}}
\label{tab:transport-fractional-loss}
\end{table*}

\subsubsection{Experiments for non-Bernoulli Process Flexibility Models}
We further examine the performance of random regular graphs in process flexibility beyond the scaled Bernoulli demand model. Specifically, we consider i.i.d.\ truncated Gaussian and Poisson demands, both with mean $1$, for systems with $n=400$, %\nbr{what's the variance for Gaussian?Maybe we can increase the variance to make the plot look better?} 
using the same graph designs and simulation procedure as in Section \ref{sect:exp-pro-flex}. The results, reported in Figure \ref{fig:flex-design-gen} and Table \ref{tab:general-demand-fractional-loss}, show that random regular graphs continue to outperform the competing designs under both demand distributions. For comparison, we also plot a general lower bound on the expected loss, 
given by the following claim.

\begin{claim}
Suppose that $D=(D_i)_{i\in [n]}$ consists of i.i.d.\ nonnegative random 
variables.
The expected loss $L(G) = \E_D[L(G, D)]$ of any bipartite graph $G$, with $n$ nodes in each partition and average degree $d$, is lower bounded by
\begin{align*}
L(G) & \ge  n \E\left[\left(1-\sum_{i=1}^{d} D_i\right)^+ \mathbf{1}_{\left\{\sum_{i=1}^n D_i \ge n\right\}} \right].
\end{align*}     
\end{claim}
\begin{proof}
For a demand realization $D$, consider a supply node $s$. The total 
amount served by \(s\) cannot exceed the aggregate demand of its neighbors,
and hence its unused capacity is at least
\[ 
\left(1-\sum_{i\in N(s)} D_i\right)^+.
\]
Summing over all supply nodes gives
\begin{align*}
z(G,D) & \le n - \sum_{s\in [n]}\left(1-\sum_{i\in N(s)} D_i\right)^+.
\end{align*} 
Recall that  $z(K_{n,n}, D) = \min\{ n, \sum_{j\in[n]} D_j\}.$
Therefore, when $\sum_{i=1}^n D_i \ge n$, we have $L(G,D) \ge \sum_{s\in[n]} \left(1-\sum_{i\in N(s)} D_i\right)^+$. When $\sum_{i=1}^n D_i < n$, we use the trivial lower bound \(L(G,D)\ge0\). Let $d_s = |N(s)|$ denote the degree of supply node $s$, and define, for 
\(k\in[n]\),
\[
f(k)
\triangleq
\E\left[
\left(1-\sum_{i=1}^{k}D_i\right)^+
\mathbf{1}_{\left\{\sum_{i=1}^{n}D_i\ge n\right\}}
\right].
\] By the i.i.d.\ assumption, we have
\begin{align*}
L(G) & = \E_D[L(G,D)] \ge \sum_{s\in [n]}\E\left[\left(1-\sum_{i\in N(s)} D_i\right)^+ \mathbf{1}_{\left\{\sum_{i=1}^n D_i \ge n\right\}}\right]  \\
& = \sum_{s\in [n]}\E\left[\left(1-\sum_{i=1}^{d_s} D_i\right)^+ \mathbf{1}_{\left\{\sum_{i=1}^n D_i \ge n\right\}} \right]  = \sum_{s\in [n]} f(d_s).
\end{align*} 
It remains to note that \(f(k)\) is convex in \(k\). Let
$\cA=\{\sum_{i=1}^{n}D_i\ge n\}$ and $h(x)=(1-x)^+$.
Conditional on \(\cA\), the random variables \(D_1,\ldots,D_n\) remain
exchangeable. Therefore, for \(0\le k\le n-2\),
\begin{align*}
&f(k+2)-2f(k+1)+f(k)\\
&\qquad =
\E\Big[
h(S_k+D_{k+1}+D_{k+2})
-h(S_k+D_{k+1})
-h(S_k+D_{k+2})
+h(S_k)
\,\Big|\cA
\Big] \prob{\cA},
\end{align*}
where \(S_k=\sum_{i=1}^kD_i\). Since \(h\) is convex and
\(D_{k+1},D_{k+2}\ge0\),
\[
h(x+a+b)-h(x+b)
\ge
h(x+a)-h(x),
\qquad a,b\ge0,
\]
so the quantity above is nonnegative. Hence \(f(k)\) is convex. Finally, since \(\sum_s d_s/n=d\), Jensen's inequality gives
\[
L(G)
\ge
\sum_{s\in[n]}f(d_s)
\ge
n f(d),
\]
which proves the claim.
\end{proof}

\begin{figure}[htp]
\centering
\includegraphics[width=.45\textwidth]{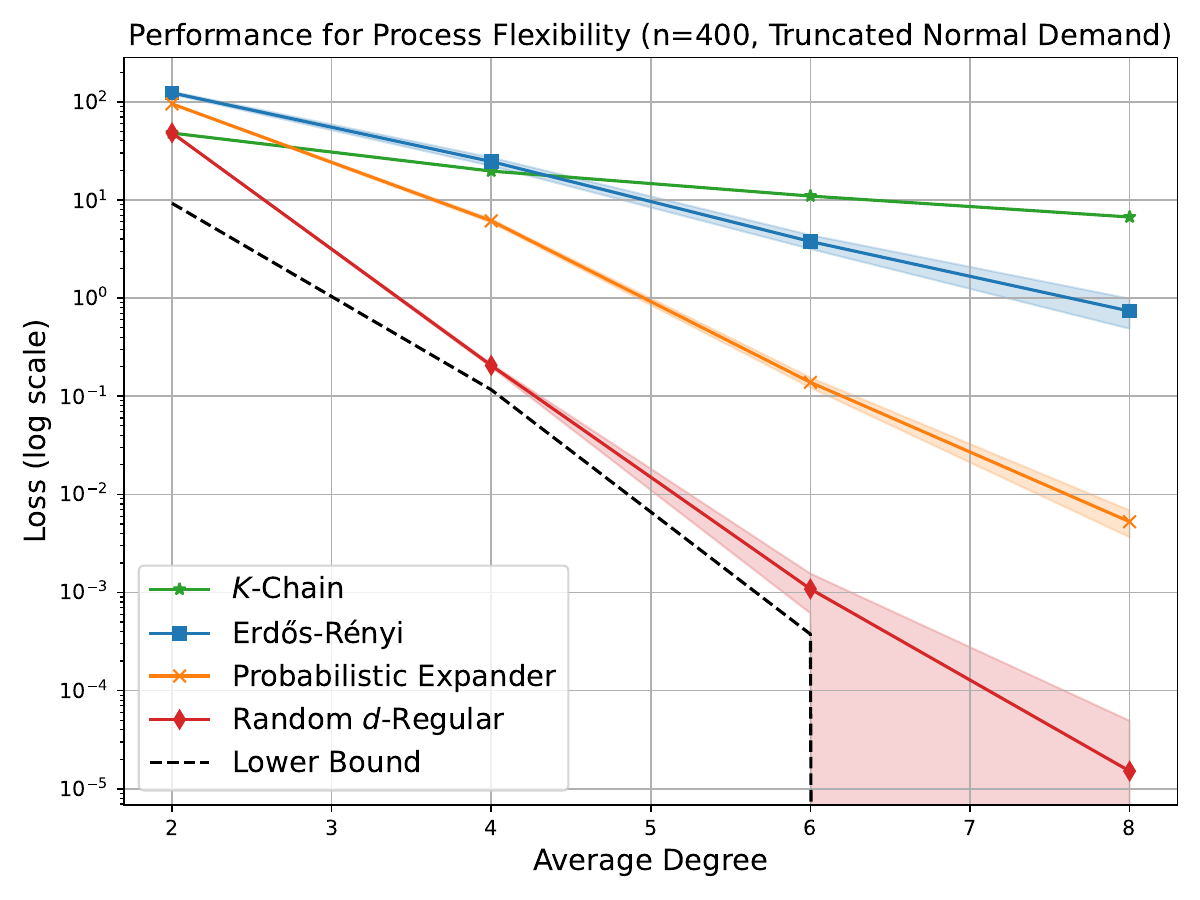}
\includegraphics[width=.45\textwidth]{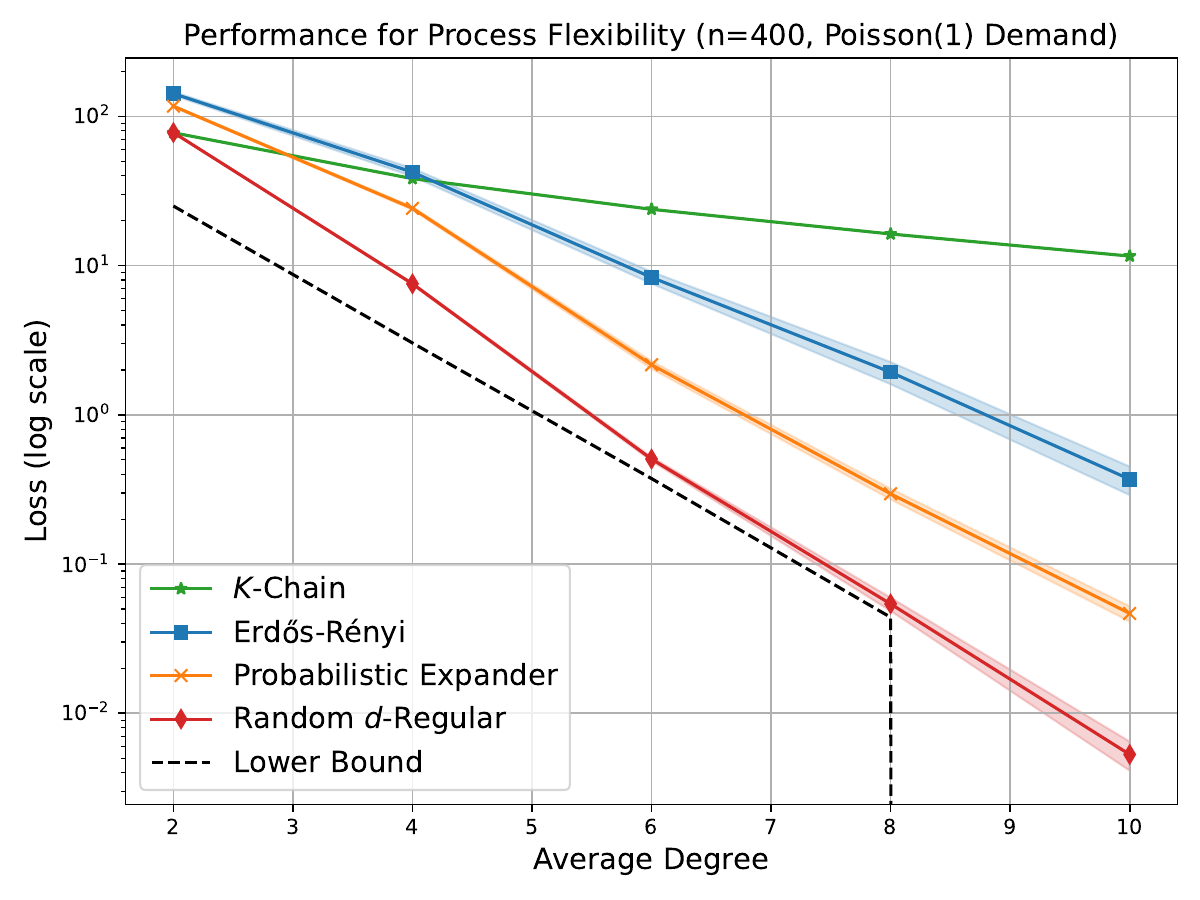}
\caption{Performance for process flexibility under general demand distributions.} %\nbr{Can we change the legend starting from $K$-Chain?}
\label{fig:flex-design-gen}
\end{figure}

\begin{table*}[htbp]
\centering
\small
\setlength{\tabcolsep}{2pt}
\begin{tabular}{@{}l@{\hspace{9pt}}*{3}{>{\centering\arraybackslash}p{0.10\textwidth}}@{\hspace{48pt}}*{3}{>{\centering\arraybackslash}p{0.10\textwidth}}@{}}
\toprule
& \multicolumn{3}{c@{\hspace{48pt}}}{(a) Truncated Gaussian demand}
& \multicolumn{3}{c}{(b) Poisson demand} \\
\cmidrule(l{4pt}r{52pt}){2-4}\cmidrule(l{4pt}r{4pt}){5-7}
Design
& $d=4$
& $d=6$
& $d=8$
& $d=4$
& $d=6$
& $d=8$ \\
\midrule
$K$-chain
& $4.919\%$
& $2.739\%$
& $1.675\%$
& $9.568\%$
& $5.961\%$
& $4.074\%$ \\
Erdős--Rényi
& $6.155\%$
& $0.943\%$
& $0.184\%$
& $10.576\%$
& $2.085\%$
& $0.483\%$ \\
Probabilistic expander
& $1.532\%$
& $0.035\%$
& $0.001\%$
& $6.054\%$
& $0.541\%$
& $0.074\%$ \\
Random regular
& $0.051\%$
& \mbox{$<0.001\%$}
& \mbox{$<0.001\%$}
& $1.889\%$
& $0.126\%$
& $0.014\%$ \\
Lower bound
& $0.029\%$
& \mbox{$<0.001\%$}
& \mbox{$<0.001\%$}
& $0.758\%$
& $0.094\%$
& $0.011\%$ \\
\bottomrule
\end{tabular}

\caption{Fractional losses under truncated Gaussian and Poisson demand ($n=400$).}
\label{tab:general-demand-fractional-loss}
\end{table*}

\section{Proofs for Middle-Mile Transportation Design}\label{sect:app-transportation}
This section proves Theorem \ref{thm:lower_bound} and Lemmas \ref{lmm:match-prop}, \ref{lmm:lwc_rrg}, and \ref{lmm:fdddd}, and completes the proofs of the main results for the transportation design in Section \ref{sect:trans}.
\subsection{Lower Bound in Theorem \ref{thm:lower_bound}}
\begin{proof}[Proof of Theorem \ref{thm:lower_bound}]
Recall that $\mu\left(K_n[\supp(D)]\right)=\lfloor |\supp(D)|/2 \rfloor$. Let $I \subset \supp(D)$ denote the isolated vertices in the subgraph $G[\supp(D)].$ 
Since isolated vertices cannot be matched, it follows that $\mu\left(G[\supp(D)]\right)\le \lfloor |\supp(D) \setminus I|/2 \rfloor$.
Therefore, 
\begin{align}
L(G, D) \ge 2 \left( \lfloor |\supp(D)|/2 \rfloor - \lfloor |\supp(D) \setminus I|/2 \rfloor \right) \label{eq:L_G_D_lower}
\end{align}
It follows that 
$L(G,D) \ge |I|- \indc{|\supp(D)| \text{ is odd}}$
and hence
$
L(G) \ge \expect{|I|} - \prob{|\supp(D)| \text{ is odd}}.
$
Let $I_i$ denote the indicator that node $i$ is isolated. Then $\prob{I_i=1}=pq^{d_i}$ and thus
$
\expect{|I|} = \sum_{i=1}^n p q^{d_i} \ge npq^{d},
$
where the inequality holds by applying Jensen's inequality to the convex function $q^{x}.$
Moreover,
$$
\prob{|\supp(D)| \text{ is odd}}
= \sum_{k \text{ is odd}}  \binom{n}{k}p^k q^{n-k}
= \frac{1}{2} \left( (p+q)^n - (q-p)^n \right)
= \frac{1}{2} \left( 1- (1-2p)^n \right).
$$
In conclusion, we get that 
$
L(G) \ge npq^d - ( 1- (1-2p)^n)/2.
$

Alternatively, we can get from~\prettyref{eq:L_G_D_lower} that $L(G, D) 
\ge 
|I| \indc{|\supp(D)| \text{ is even} }
$
and hence
\begin{align*}
L(G) & \ge \expect{ |I| \indc{|\supp(D)| \text{ is even}}} \\
& =\sum_{k \text{ is even}}
\prob{|\supp(D)|=k} \expect{|I| \mid |\supp(D)|=k} \\
&= \sum_{k \text{ is even}}
\prob{|\supp(D)|=k} \sum_{i=1}^n 
\expect{I_i \mid |\supp(D)|=k}.
\end{align*}
Note that $\prob{|\supp(D)|=k}=\binom{n}{k}p^k q^{n-k}$ and
$$
\expect{I_i \mid |\supp(D)|=k }
=\prob{I_i=1 \mid |\supp(D)|=k}
= \frac{\binom{n-d_i-1}{k-1}}{ \binom{n}{k}}
$$
It follows that 
\begin{align*}
L(G)
&\ge \sum_{k \text{ is even}} 
\sum_{i=1}^n
\binom{n-d_i-1}{k-1} p^k q^{n-k} \\
& \ge n \sum_{k \text{ is even}} \binom{n-d-1}{k-1} p^k q^{n-k} \\
& = npq^{d}
\sum_{k \text{ is even}} \binom{n-d-1}{k-1} p^{k-1} q^{n-d-k}  \\
&=
\frac{1}{2} np q^d  \left( 1 - (1-2p)^{n-d-1}\right),
\end{align*}
where the inequality holds by applying Jensen's inequality to the convex function $\binom{x}{k} \indc{x\ge k}$.
%}
% \xn{Another tighter lower bound compared to $npq^d-1$: In fact,
% \$ 
% L(G,D) \ge |I| \indc{|I| \text{ is even}} + (|I|+1) \indc{|I| \text{ is odd}, |\supp(D)| \text{ is even}} + (|I|-1) \indc{|I| \text{ is odd}, |\supp(D)| \text{ is odd}}.
% \$
% Then we have
% \$
% \E[L(G,D)] &\ge \E[|I|] + \prob{|I| \text{ is odd}, |\supp(D)| \text{ is even}} - \prob{|I| \text{ is odd}, |\supp(D)| \text{ is odd}} \\
% & \ge npq^d - \prob{|\supp(D)| \text{ is odd}}.
% \$
% Here $|\supp(D)| \sim \text{Bin}(n, p)$.
%}
\end{proof}

\subsection{Proofs for Constant Loss}\label{sec:app-trans-cons}
In this section, we prove Lemma \ref{lmm:match-prop} and establish \prettyref{thm:match-1-n}. 

\subsubsection{Proof of \prettyref{lmm:match-prop}}
We begin by stating the key properties in the following four lemmas, and then verify that the minimum-degree requirement $d$ in \eqref{eq:cond_d} applies to all of these lemmas, which in turn yields \prettyref{lmm:match-prop}.

The following standard result bounds the maximum size of an independent set in random $d$-regular graphs. The precise asymptotic value has been determined in~\cite{ding2016maximum}.
\begin{lemma}[No Large Independent Set]
\label{lmm:ind_set}
There exist constants $C, c>0$ such that for all $d \ge C,$ the random $d$-regular graph $G \sim \calG(n,d)$ satisfies
$$
\probs{\mathrm{Ind}(G) \ge \frac{2n\log d}{d} }{G} \le
\exp\left( - c \frac{n\log d}{d} \right).
$$
\end{lemma}
\begin{proof}
Let $e(T)$ denote the number of edges connecting two distinct vertices in $T$ (counting edge multiplicities but not self-loops). Recall that $\mathrm{Ind}(G)$ is the maximum size of an independent set in $G$. Thus, it suffices to prove that
\$
\prob{\exists T \subset [n] \text{ with } |T| \ge  \frac{2n \log d }{d} : e(T)=0 } \le
\exp\left( - c \frac{n\log d}{d} \right).
\$
Fix a set $T \subset [n]$ with $|T|=t$. 
We first bound the probability that $T$ is an independent set, \ie, $e(T)=0$. 
There are $td$ half-edges incident to $T$. Suppose $2\ell$ of these half-edges are matched among themselves to form $\ell$ self-loops. Then $e(T)=0$ if and only if all remaining $td-2\ell$ half-edges are paired with half-edges incident to $[n]\setminus T.$
From this observation, we obtain the following bound: 
$$
\prob{e(T)=0} \le
\sum_{\ell} \binom{td}{\ell} d^\ell 
(nd-td)_{td-2\ell} \cdot \frac{1}{[nd]_{td-\ell}} ,
$$
where $(x)_n=x(x-1)\ldots (x-n+1)$ denotes the falling factorial and $[x]_n=(x-1)(x-3)\ldots (x-2n+1)$ denotes the falling double factorial. 
%Note that the denominator $(nd-1)!!$ counts all possible pairings of $nd$ half-edges without restriction. 
In the above expression, we sum over $\ell$, the possible number of self-loops incident to vertices in $T.$ For each value of $\ell$, we first select 
$\ell$ half-edges out of $td$ half-edges incident to $T$, to serve as one endpoint of each self-loop; there are $\binom{td}{\ell}$ choices.
Each chosen half-edge can be paired with at most  $d$ half-edges from the same vertex to form a loop, giving at most $d^\ell$ possibilities in total.
Then we pair the remaining $td-2\ell$
half-edges with half-edges incident to $[n]\setminus T$
and there are 
$(nd-td)_{td-2\ell}$ such ways. So far we have formed $\ell+(td-2\ell)=td-\ell$ edges.
The denominator $[nd]_{td-\ell}$ counts all the possible ways to form $td-\ell$ edges without restriction.

We further bound the above probability as follows. Note that 
\begin{align*}
\frac{(nd-td)_{td-2\ell}} {[nd]_{td-\ell}}  = \frac{(nd-td)_{td}} {[nd]_{td} \cdot [nd-2td+2\ell+1]_\ell }.
\end{align*} 
Now,
 % Therefore,
% \begin{align*}
% \prob{e(T)=0} & \le
% \frac{(nd-td)_{td}} {[nd]_{td} }
% \sum_{\ell} \binom{td}{\ell}  \frac{d^\ell}{[nd-2td+2\ell+1]_\ell} \\
% & \le 
% \frac{(nd-td)_{td}} {[nd]_{td} }
% \sum_{\ell} \binom{td}{\ell}  \left( \frac{1}{n-2t}\right)^{\ell} \\
% & \le \frac{(nd-td)_{td}} {[nd]_{td} } \exp\left( 
% \frac{td}{n-2t}\right)
% \end{align*}
$$
\frac{(nd-td)_{td}} {[nd]_{td}}
=\prod_{i=0}^{td-1}
\frac{nd-td-i}{nd-2i-1}
\le \prod_{i=0}^{td-1}
\frac{nd-td+i+1}{nd}
\le \exp \left( - \sum_{i=0}^{td-1} \frac{td-i-1}{nd} \right)
=\exp \left( - \frac{t(td-1)}{2n} \right).
$$
Moreover,
$$
 \sum_{\ell} \binom{td}{\ell}  \frac{d^\ell}{[nd-2td+2\ell+1]_\ell}  \le 
\sum_{\ell} \binom{td}{\ell}  \left( \frac{1}{n-2t}\right)^{\ell} =  \left(1+ \frac{1}{n-2t}\right)^{td} \le \exp\left( \frac{td}{n-2t}\right).
$$
Therefore, we obtain that 
$$
\prob{e(T)=0}
\le 
\exp \left( - \frac{t(td-1)}{2n} 
+ \frac{td}{n-2t} \right).
$$
Finally, since there are $\binom{n}{t}$ different choices of $T \subset [n]$ with $|T|=t$, applying union bound yields that %\nbr{The upper bound here is for $|T|=t$, why we can get $|T|\ge t$ below? Perhaps we should sum over $t$ and then illustrate this using monotonic or exponential decay? } \nb{not needed. Because the existence of $|T| \ge t$ implies the existence of $|T|=t$. Add an extra step to clarify.} \xn{The summation over $t$ is needed. Fix later.}
\begin{align*} 
\prob{\exists T \subset [n] \text{ with } |T| \ge  t : e(T)=0 } 
& = \prob{\exists T \subset [n] \text{ with } |T| = t : e(T)=0 }  \\
& \le 
\binom{n}{t} \exp \left( - \frac{t(td-1)}{2n} 
+ \frac{td}{n-2t} \right) \\
& \le 
\exp\left( t \left[ 
\log \frac{en}{t} - \frac{td-1}{2n} + \frac{d}{n-2t}\right]\right). 
\end{align*}
Choosing $t=2n\log d/d$, we have
\begin{align*}
\log \frac{en}{t} - \frac{td-1}{2n} + \frac{d}{n-2t}
= \log \frac{ed }{2\log d}
- \log d + \frac{1}{2n}
+ \frac{d}{n(1-4 \log d/d)} \le -c,
\end{align*}
for some constant $c>0$, where the last inequality holds under the assumption that $d\ge C.$
\end{proof}

The following result shows that, with high probability,  random $d$-regular graphs do not contain two sufficiently large subsets of vertices with no crossing edges. 

\begin{lemma}[No Big Cut]\label{lmm:no_big_cut_1}
Consider the random $d$-regular graph $G \sim \calG(n,d).$ There exist constants $C, c>0$ such that for all $d \ge C,$
$$\probs{ \exists \text{ disjoint } S, T \subset [n] \text{ with } |S|, |T| \ge  \frac{4 n \log d}{d}: e(S, T)=0}{G} \le \exp\left( - c \frac{n\log d}{d} \right).
$$
\end{lemma}
\begin{proof}
Fix two disjoint sets $S, T \subset [n]$ with $|S|=s$ and $|T|=t$. 
We bound the probability that there is no crossing edge between $S$ and $T$ as follows:
\$
\prob{e(S,T)=0}
\le (1-t/n)^{sd/2} \le \exp(-std/(2n)), \$
where the second inequality follows from $1-x \le \exp(-x).$ 
To see why the first inequality is true, we sequentially reveal the pairings for the half-edges incident to the vertices in $S$. 
Conditioned on the pairings revealed so far and no crossing edge between $S$ and $T$
has formed, letting $m$ denote the total number of half-edges remained, the  probability that the next half-edge is not paired with a half-edge incident to vertices in $T$ is $(m-td-1)/(m-1)$. Now, crucially, no matter the value of $m$ is, 
$(m-td-1)/(m-1) \le (nd-td)/nd = 1-t/n$. Also, since the half-edges incident to the vertices in $S$ could be paired among themselves, to reveal all their pairings, we need at least $(sd/2)$ steps. 

Since there are at most $\binom{n}{s}$ choices of $S \subset [n]$ with $|S|=s$ and $\binom{n}{t}$ choices of $T \subset [n]$ with $|T|=t$, fixing $s,t\ge 4n \log d/d$ and applying the union bound yields that %\xn{Similarly, we need an additional sum over $s,t$.} \nb{I do not understand why we need this sum. Let's discuss.}
\begin{align*}
&\prob{ \exists \text{ disjoint } S, T \subset [n] \text{ with } |S|, |T| \ge t: e(S, T)=0}  \\
& = \prob{ \exists \text{ disjoint } S, T \subset [n] \text{ with } |S|, |T| = t: e(S, T)=0}  \\
& \le 
\binom{n}{t} \binom{n}{t}
\exp \left( -\frac{t^2 d}{2n} \right) \\
&\le \exp \left( 2 t \log \frac{en}{t} - \frac{t^2d}{2n} \right).
\end{align*}
The proof is complete by choosing $t=4n \log d/d$ and noting that $\log (en/t) - (td)/(4n)
=\log (e/4) - \log \log d.$
\end{proof}

The following lemma shows that there is little congestion, in the sense that with high probability, there is no $S, T \subset V$ with 
$|S|+1 \le |T| \le \delta n$, for which there are at least $|S|+1$ outgoing edges from $T$ to $V\setminus T$ and all are connected to $S.$
In particular, for the special case of $|S|=1$ and $|T|=2$, there are no two degree-$1$ vertices that are connected to the same vertex in $G[V].$ Note that this is clearly necessary for the existence of a perfect matching for $G[V]$, as otherwise, the two degree-$1$ vertices cannot be matched simultaneously. 
\begin{lemma}[Little Congestion]\label{lmm:congestion}
Fix a set $V \subset [n]$ with $|V|=n-m$ and $m \le \alpha n$ for some constant $\alpha.$ Suppose that for some constant $\delta \in (0,1)$,
$$
d \ge \frac{(1+\delta)\log n}{\log(1/\alpha)}. 
$$
Then there exists a constant $\epsilon \equiv\epsilon(\delta,\alpha)>0$ such that for all sufficiently large $n,$ 
$$\prob{ \exists  \text{ disjoint } S, T\subset  V \text{ with }  |S|+1 \le |T| \le \epsilon n: e(S,T)= e(V\setminus T, T)  \geq |S|+1} \le n^{-1-\delta/4}.$$
%$$\prob{ \exists  \text{ disjoint } S, T\subset  V \text{ such that }  |S|+1 \le |T| \le \epsilon n, N(T)\setminus T \subset S, e(S,T) \geq |S|+1} \le n^{-1-\delta/4}.$$
\end{lemma}
%\nbr{Please carefully check: Here we can take $\epsilon$ to be an arbitrarily small constant, so that later we do not need to ensure $|T| \lesssim n \log^2(d)/d^2$ in the later proof. In other words, we do not need the no-big-cut property-case II.}
\begin{proof}
Fix two disjoint sets $S, T\subset  V$ with $|S| = s$ and $|T|=t$. We bound the probability that there are at least $k$ edges going out of $T$ to $V\setminus T$ and all are connected to vertices in $S,$ \ie, 
$e(S,T)=e(V\setminus T, T) \ge k $.
There are $td$ half-edges incident to $T$. Suppose $2\ell$ of these half-edges are matched among themselves to form $\ell$ edges within $T.$ Then $e(S,T)=e(V\setminus T, T) \ge k $ if and only if among the remaining $td-2\ell$ half-edges, $k$ of them are paired with half-edges incident to $S$ and $td-2\ell-k$ of them  are paired with half-edges incident to $S\cup ([n]\setminus V)$. From this observation, we obtain the following bound: for $k\le \min\{s,t\}\cdot d,$
\begin{align*}
&\prob{e(S,T)=e(V\setminus T, T) \ge k} \\
& \le \binom{td}{k}  (sd)_{k}
\cdot \sum_{\ell} \binom{td-k}{2\ell} (2\ell-1)!!
\cdot (sd-k+md)_{td-k-2\ell}  \cdot \frac{1}{[nd]_{td-\ell}} 
\end{align*}
 In the above expression, we first select $k$ half-edges from $td$ half-edges incident to $T$, and there are $\binom{td}{k}$ choices. These chosen $k$ half-edges are paired with half-edges incident to $S$, and there are $(sd)_{k}$ different ways.
 %with probability $\prod_{i=0}^{k-1} \frac{sd-i}{nd-2i-1}$. 
 Then we sum over $\ell,$ the possible number of edges within $T.$ For each value of $\ell,$ we select the $2\ell$ half-edges from the remaining $td-k$ ones,
 and there are $\binom{td-k}{2\ell}$ choices. These $2\ell$ chosen half-edges are paired among themselves to form $\ell$ edges within $T$, and there are $(2\ell-1)!!$ different pairings.
 %with probability $\prod_{i=0}^{\ell-1} 
%\frac{2\ell-2i-1}{nd-2k-2i-1}$. 
Finally, the remaining $td-k-2\ell$ half-edges are paired with half-edges incident to $S \cup ([n]\setminus V)$, and there are $(sd-k+md)_{td-k-2\ell}$ different pairings.
So far we have formed 
$k+\ell+(td-k-2\ell)=td-\ell$ edges. Thus the denominator $[nd]_{td-\ell}$ counts all the possible pairings to form $td-\ell$ edges without restriction.
%with probability $\prod_{j=0}^{td-2\ell-k-1} 
%\frac{sd-k+ md-j}{nd-2k-2\ell-2j-1}$.

% \begin{align*}
% &\prob{e(S,T)=e(V\setminus T, T) \ge k} \\
% & \le \sum_{\ell} \binom{td}{2\ell} 
% \prod_{i=0}^{\ell-1} 
% \frac{2\ell-2i-1}{nd-2i-1}
% \binom{td-2\ell}{k} 
% \prod_{i=0}^{k-1} \frac{sd-i}{nd-2\ell-2i-1}
% \prod_{j=0}^{td-2\ell-k-1} 
% \frac{sd-k+ md-j}{nd-2\ell-2k-2j-1}
% \end{align*}
%  In the above expression, we sum 
%  over $\ell,$ the possible numbers of edges within $T.$ For each value of $\ell,$ we select the $2\ell$ half-edges from $kd$ ones,
%  and there are $\binom{kd}{2\ell}$ choices. These $2\ell$ chosen half-edges are paired among themselves to form $\ell$-edges within $T$, with probability $\prod_{i=0}^{\ell-1} 
% \frac{2\ell-2i-1}{nd-2i-1}$. Then from the remaining $td-2\ell$ half-edges, we select $k$ half-edges, and there are $\binom{td-2\ell}{k}$ choices. These chosen $k$ half-edges are paired with half-edges incident to $S$, with probability $\prod_{i=0}^{k-1} \frac{sd-i}{nd-2\ell-2i-1}$. Finally, the remaining $td-2\ell-k$ half-edges are paired with half-edges incident to $S \cup ([n]\setminus V)$, with probability $\prod_{j=0}^{td-2\ell-k-1} 
% \frac{sd-k+ md-j}{nd-2\ell-2k-2j-1}$.

Note that since $2\ell+k \le td$ and $k \le sd$, it follows that
\begin{align*}
\frac{(sd)_k}{[nd]_k}
= \prod_{i=0}^{k-1} \frac{sd-i}{nd-2i-1} 
\le \prod_{i=0}^{k-1} \frac{sd+i+1}{nd} \le 
\left(\frac{2s}{n}\right)^k. 
\end{align*}
Similarly
$$
\frac{(2\ell-1)!!}{[nd-2k]_{\ell}}=
\prod_{i=0}^{\ell-1}
\frac{2\ell-2i-1}{nd-2k-2i-1} 
\le \left(\frac{2\ell+2k}{nd} \right)^\ell \le \left( \frac{2t}{n} \right)^\ell
$$
and 
\begin{align*}
\frac{(sd-k+md)_{td-k-2\ell}}{[nd-2k-2\ell]_{td-k-2\ell}}
 =\prod_{i=0}^{td-k-2\ell-1} \frac{sd-k+ md-i}{nd-2\ell-2k-2i-1} 
& \le \prod_{i=0}^{td-k-2\ell-1}\frac{sd+md+2\ell+k+i+1}{nd} \\
& \le \left(\frac{m+s+t}{n}\right)^{td-k-2\ell}.
\end{align*}
Since $[nd]_{td-\ell}=[nd]_{k} \cdot [nd-2k]_{\ell} \cdot [nd-2k-2\ell]_{td-k-2\ell}$, it follows that 
\begin{align*}
\prob{e(S,T)=e(V\setminus T, T) \ge k}
& \le \binom{td}{k} \left(\frac{2s}{n} \right)^k \cdot \sum_\ell \binom{td-k}{2\ell} 
\left( \frac{2t}{n}\right)^\ell 
\left(\frac{m+s+t}{n} \right)^{td-2\ell-k} \\
%&=  \left(\frac{m+s+t}{n} \right)^{td} \binom{td}{k} \left(\frac{2s}{m+s+t} \right)^k
%\sum_\ell \binom{td-k}{2\ell} 
%\left( \frac{\sqrt{2nt}}{m+s+t}\right)^{2\ell} \\
%&\le  \left(\frac{m+s+t}{n} \right)^{td} \binom{td}{k}\left(\frac{2s}{m+s+t} \right)^k
%\left( 1+ \frac{\sqrt{2nt}}{m+s+t}\right)^{td}\\
%&=\left(\frac{m+s+t+\sqrt{2nt}}{n} \right)^{td} \binom{td}{k}\left(\frac{2s}{m+s+t} \right)^k \\
& \le \binom{td}{k} \left(\frac{2s}{n} \right)^k
\left( \sqrt{\frac{2t}{n}} + \frac{m+s+t}{n} \right)^{td-k} \\
& \le \left(\frac{\alpha n+s+t+\sqrt{2nt}}{n} \right)^{td} \left(\frac{2estd}{\alpha n k} \right)^k,
\end{align*}
where the last inequality holds by the assumption that $m\le \alpha n.$
Finally, applying a union bound over the choices of $S,T$ yields that 
\begin{align*}
& \prob{ \exists  \text{ disjoint } S, T\subset  V \text{ with } |S|+1 \le |T| \le \epsilon n : e(S,T)= e(V\setminus T, T)  \geq |S|+1} \\
&\le \sum_{s} \sum_{s+1\le t \le \epsilon n} \binom{n}{s}\binom{n-s}{t}
\left(\frac{\alpha n+s+t+\sqrt{2nt}}{n} \right)^{td} \left(\frac{2e std}{\alpha n(s+1)}\right)^{s+1} 
\end{align*}
%%YW: the factor \(\sqrt{2}\) disappears: previous equation has \(\sqrt{2nt}\), but this equation has \(\sqrt{nt}\).
In the above expression, we first sum over $t$ and get that 
\begin{align*}
\sum_{s+1\le t \le \epsilon n }
\binom{n-s}{t}\left(\frac{\alpha n+s+t+\sqrt{2nt}}{n} \right)^{td}  t^{s+1}
& \le \sum_{s+1\le t \le \epsilon n}
\left(\frac{en}{t} \right)^t
\left(\frac{\alpha n+2t+\sqrt{2nt}}{n} \right)^{td}
t^{s+1} \\
& \le \sum_{s+1\le t \le \epsilon n } 
\left[en
\left(\alpha+2\epsilon +\sqrt{2\epsilon} \right)^{d}\right]^t \\
& \le \left[e \beta^{d\delta'} \right]^{s+1}/(1-e\beta^{d\delta'})
\end{align*}
where $\beta=\alpha+2\epsilon +\sqrt{2\epsilon}$
and $\delta'=\delta/(2+\delta)$,
and the last inequality holds because  $d \ge (1+\delta)\log n/\log(1/\alpha)$ so that  there exists $\epsilon(\delta,\alpha)$ such that 
$d\ge (1+\delta/2)\log n/\log (1/\beta) $.
Therefore, we obtain that 
\begin{align*}
& \prob{ \exists  \text{ disjoint } S, T\subset  V \text{ with } |S|+1 \le |T| \le \epsilon n: e(S,T)= e(V\setminus T, T)  \geq |S|+1} \\
&\le \frac{1}{1-e\beta^{d\delta'}} 
\sum_{s } \binom{n}{s}
\left(\frac{2e^2  d\beta^{d\delta'}}{\alpha n }\right)^{s+1}
= \frac{1}{1-e\beta^{d\delta'}} 
\left(\frac{2e^2  d\beta^{d\delta'}}{\alpha n }\right)
\left( 1+ \frac{2e^2  d\beta^{d\delta'}}{\alpha n }\right)^n
\le n^{-1-\delta/4},
\end{align*}
where the last inequality holds for all sufficiently large $n$ because 
$\beta^{d\delta'} \le n^{-\delta/2}$.
\end{proof}

Finally, the next lemma shows that the subgraph $G[V]$ induced by a fixed vertex set $V$ does not contain a small component (in particular, isolated vertices), with high probability. Since each odd component contains at least one unmatched vertex, no odd component is necessary for the existence of a perfect matching for $G[V].$ Here $X_k$ is the number of connected components of size $k$ in $G[V]$.
\begin{lemma}[No Small Component]\label{lmm:no_small_component}
Fix a set $V \subset [n]$ with $|V|=n-m$ and $m\le\alpha n.$  
%$I$ denote the isolated vertices in $G[V]$, and $G[V]-I$ denote the graph $G[V]$ after removing the isolated vertices. 
%Suppose that $d$ is given by~\prettyref{eq:cond_d}.
%$m \le  nq (1-\epsilon)$ for some constant $\alpha>0$ and  
Suppose that for some constant $\delta \in (0,1)$,
$$
d \ge \frac{(1+\delta)\log n}{\log(1/\alpha)}. 
$$
 Then there exists a constant $\epsilon \equiv \epsilon(\delta,\alpha)>0$ such that for
all sufficiently large $n$
\$
\prob{X_1 \ge 1} \le \expect{X_1} \le n^{-\delta/2} ,\$
and 
\$
\prob{ \sum_{k=2}^{\epsilon n} X_k \ge 1 }  
\le \sum_{k=2}^{\epsilon n} \expect{X_k} \le n^{-1-\delta/2}.  \$

\end{lemma}
\begin{proof}

Fix a subset $T \subset V$ with $k$ vertices. There are $kd$ half-edges incident to $T$. Then $T$ is a connected component of $G[V]$ if and only if $2(k-1)$ of these half-edges are paired among themselves to form a spanning tree with $k-1$ edges, and all the remaining
$kd-2(k-1)$ half-edges are either paired among themselves or paired with edges incident to $[n]\setminus V.$  From this observation, we obtain the following probability bound:
\begin{align*}
&\prob{ T \text{ is a component}} \\
& \qquad \le k^{k-2} d^{2k-2} \sum_{\ell} \frac{1}{[nd]_{kd-(k-1)-\ell}} \binom{kd-2(k-1)}{2\ell} (2\ell-1)!! \cdot (md)_{kd-2(k-1)-2\ell}.
\end{align*}
In the above expression, we first choose a spanning tree on the vertex set $T$ and there are $k^{k-2}$ different choices. 
For each chosen spanning tree, we then select and pair the half-edges to form the $k-1$ edges in the spanning tree, and there are at most $d^{2(k-1)}$ ways, as each edge can be formed by at most $d^2$ different choices of half-edges. Next, we sum over $\ell$, the possible number of edges within $T$ in addition to the $k-1$ edges in the chosen spanning tree.  For each value of $\ell,$ we select the $2 \ell$ half-edges from the remaining $kd-2(k-1)$ half-edges incident to $T$, and there are $\binom{kd-2(k-1)}{2\ell}$ choices. These $2\ell$ chosen half-edges are paired among themselves to form $\ell$ edges within $T$, and there are $(2\ell-1)!!$ different ways. Finally, the remaining $j\triangleq kd-2(k-1)-2\ell$ half-edges are paired with half-edges incident to $[n]\setminus V$, and there are $(md)_{j}$ different ways. Note that in total we have formed $(k-1)+\ell+j=kd-(k-1)-\ell$ edges. Thus, the denominator $[nd]_{kd-(k-1)-\ell}$ counts all possible pairings to form  $kd-(k-1)-\ell$ edges without restriction. 

We further simplify the upper bound as follows. 
$$
\frac{(md)_{j}}{[nd]_{j}}
= \prod_{i=1}^{j} \frac{md-i+1}{nd-2i+1}
\le \prod_{i=1}^{j} \frac{md+i}{nd}
\le \alpha^{j}
\exp \left( \frac{1}{\alpha  nd} \sum_{i=1}^{j} i \right) 
\le \left(\alpha \exp \left( \frac{\epsilon}{\alpha} \right) \right)^{j},
$$
where we used the facts that $m\le \alpha n$, $j \le kd$, and $k\le \epsilon n.$
Furthermore, 
$$
\frac{(2\ell-1)!!}{[nd-2j]_{\ell}}
=\prod_{i=1}^\ell \frac{2\ell-2i+1}{nd-2j-2i+1}
\le \prod_{i=1}^\ell \frac{2\ell+2j}{nd} \le 
\left( \frac{2k}{n}\right)^\ell
\le (2\epsilon)^\ell
$$
and 
$$
[nd-2j-2\ell]_{k-1} \ge \left(nd-2j-2\ell-2(k-1) \right)^{k-1} \ge [(n-2k)d]^{k-1}.
$$
Since $[nd]_{kd-(k-1)-\ell}=[nd]_{j} \cdot [nd-2j]_{\ell} \cdot [nd-2j-2\ell]_{k-1}$, 
% YW: [nd-j]_{\ell} should be [nd-2j]_{\ell}?
it follows that 
\begin{align*}
\prob{ T \text{ is a component}} 
& \le \left(\frac{kd}{n-2k} \right)^{k-1} 
\sum_{\ell} \binom{kd-2(k-1)}{2\ell}
\left( 2\epsilon \right)^\ell
\left(\alpha e^{\epsilon/\alpha} \right)^{kd-2(k-1)-2\ell} \\
& \le 
\left(\frac{kd}{n-2k} \right)^{k-1} \left( \alpha e^{\epsilon/\alpha} + \sqrt{2\epsilon}\right)^{kd-2(k-1)} \\
&= \left(\frac{kd}{(n-2k)\beta^2} \right)^{k-1} \beta ^{kd}. 
%\le \left(\frac{kd}{(n-2k)\beta^2} \right)^{k-1} n^{-k(1+\epsilon/2)},
\end{align*}
where  $\beta =\alpha e^{\epsilon/\alpha} + \sqrt{2\epsilon}$.

%and the last inequality holds because by assumption $m=\alpha n$ and $d \ge (1+\epsilon) \log (n)/\log(1/\alpha)$ so that there exists $\epsilon(\epsilon,\alpha)$ such that $d \ge (1+\epsilon/2)\log n/\log(1/\beta).$

Since there are at most $\binom{n}{k} $ different choices of $T \subset V$ with $k$ vertices, it follows that 
$$
\expect{X_k} \le  \binom{n}{k} \left(\frac{kd}{(n-2k)\beta^2} \right)^{k-1} \beta^{kd}
%\le e  \left(\frac{ed}{(n-2k)\beta^2} \right)^{k-1} (n\beta^d)^k
\le  \binom{n}{k} \left(\frac{kd}{(n-2k)\beta^2} \right)^{k-1} n^{-k(1+\delta/2)},
$$
where the last inequality holds because by assumption  $d \ge (1+\delta) \log (n)/\log(1/\alpha)$ so that there exists $\epsilon(\delta,\alpha)$ such that $d \ge (1+\delta/2)\log n/\log(1/\beta).$

Therefore, for $k=1$, $\expect{X_1} \le n^{-\delta/2}$. Moreover, 
\begin{align*}
\sum_{k=2}^{\epsilon n} \expect{X_k}
& \le  \sum_{k=2}^{\epsilon n} 
\left( \frac{en}{k} \right)^k \left(\frac{kd}{(n-2k)\beta^2} \right)^{k-1}  n^{-k(1+\delta/2)} \\
&  \le en^{-\delta/2} \sum_{k=2}^{\epsilon n}  \left(\frac{e d}{(1-2\epsilon)\beta^2 n^{1+\delta/2}} \right)^{k-1} \le  n^{-1-\delta/2}.
\end{align*}
The proof is complete by applying Markov's inequality. 
%It follows that 
%$\expect{X_1} \le n^{-\epsilon/2}$
\end{proof}

\iffalse
\begin{remark}\label{rmk:properties_regular_graphs}
    Note that \prettyref{lmm:ind_set} and \prettyref{lmm:no_big_cut_1} establish properties of the entire random regular graph $G^d$. In contrast, \prettyref{lmm:congestion} and \prettyref{lmm:no_small_component} concern properties of the subgraph of $G^d$ induced by a fixed vertex set $V$. It is also worth emphasizing that the first two properties hold whenever $d$ exceeds a sufficiently large constant. Hence, they are not essential for random regular graphs to achieve bounded loss under the minimum degree requirement; analogous properties also hold for Erdős–Rényi random graphs with the same average degree. By contrast, the latter two properties hold only under the critical condition $d \ge (1+\delta) \log(n)/\log(1/\alpha)$.
For Erdős–Rényi random graphs, however, a higher connectivity condition is required, namely degree $d \ge (1+\delta) \log (n)/(1-\alpha)$.
\end{remark}
\fi

\begin{proof}[Proof of \prettyref{lmm:match-prop}]
Since Lemmas \ref{lmm:ind_set} and \ref{lmm:no_big_cut_1} require only that $d$ be a sufficiently large constant, the choice $d = \Theta(\log n)$ \eqref{eq:cond_d} satisfies their assumptions for large $n$; hence, \eqref{eq:ind_set} and \eqref{eq:no_big_cut_1} follow directly. It remains to show \eqref{eq:congestion} and \eqref{eq:no_small_comp} by verifying the assumptions in Lemmas \ref{lmm:congestion} and \ref{lmm:no_small_component}. 

Since $|V| \ge (1-\eta)np$, we have $m=n-|V| \le n-(1-\eta)np =(q+\eta p) n$. Let $\alpha=(q+\eta p).$
Recalling the condition~\prettyref{eq:cond_d} that $d \ge (1+\delta) \log n/ \log(1/q)$, we may choose $\eta$ to be a sufficiently small constant so that  
$d \ge (1+\delta') \log n/\log(1/\alpha)$ with $\delta'=\delta/2.$ Under this condition, \eqref{eq:congestion} holds with $\epsilon_1\equiv \epsilon_1(\delta', \alpha)$ as given in~\prettyref{lmm:congestion}, and \eqref{eq:no_small_comp} holds with $\epsilon_2\equiv \epsilon_2(\delta', \alpha)$ as given in~\prettyref{lmm:no_small_component}. Setting $\gamma = \min\{\epsilon_1, \epsilon_2\}$, we obtain both \eqref{eq:congestion} and \eqref{eq:no_small_comp}.
\end{proof}

\subsubsection{Proof of \prettyref{thm:match-1-n}}
Finally, we prove \prettyref{thm:match-1-n} via \prettyref{prop:perfect matching}.
\begin{proof}[Proof of Theorem \ref{thm:match-1-n}]
Let $\calA$ denote the event that $\mu(G[\supp(D)])=\lfloor |\supp(D) \setminus I|/2 \rfloor$ and $\calB$ denote the event that $|\supp(D)| \ge (1-\epsilon)np.$
Combining~\prettyref{prop:perfect matching} with the Chernoff bound on $|\supp(D)|$ yields that 
$$
\probs{\calA \cap \calB}{G, D}
= \prob{\calA \mid \calB} \prob{\calB}
\ge 1- n^{-1-\delta/8} - \exp\left(-\epsilon^2 np/2 \right).
$$
Moreover, $\E_{G,D}[|I| \mathbf{1}_{\calB}] \le n^{-\delta/8},$
and on event $\calA,$
\begin{align*}
L(G, D) & = 2 \left( \lfloor |\supp(D)|/2 \rfloor - \mu(G[\supp(D)]) \right)  = 2 \left( \lfloor |\supp(D)|/2 \rfloor -  \lfloor |\supp(D) \setminus I |/2 \rfloor \right)  \\
& \le |I| + \indc{|I| \text{ is odd}} \le 2|I|.
\end{align*}
Therefore, since it always holds that $L(G, D)\le n$, we have
\begin{align*}
\expects{L(G, D)}{G, D}
& = \expects{L(G, D) \indc{\calA\cap \calB}}{G, D} + \expects{L(G, D) \indc{(\calA\cap \calB)^c}}{G, D} \\
& \le 2 \expects{|I| \mathbf{1}_{\calB}}{G, D}
+ n \prob{(\calA\cap\calB)^c} \\
& \le 2n^{-\delta/8} + n^{-\delta/8} + n \exp\left(-\epsilon^2 np/2 \right)
=O(n^{-\delta/8}).
\end{align*}
Finally, the proof is complete by applying Markov's inequality: 
\begin{equation*}
\probs{\expects{L(G,D)}{D} \ge n^{-\delta/16}}{G} \le n^{\delta/16}\expects{\expects{L(G,D)}{D}}{G}=O(n^{-\delta/16}). \qedhere
\end{equation*}
\end{proof}

\subsection{Proofs for $\epsilon$-Fractional Loss} \label{sect:ec-lwc}
This section proves Lemmas \ref{lmm:lwc_rrg} and \ref{lmm:fdddd} and thereby completes the analysis of the $\epsilon$-fractional loss. We begin by introducing the notion of local weak convergence for deterministic graphs. Interested readers are referred to \cite{van2024random} for further details.
A \emph{rooted graph} is a pair $(G, o)$, where $G$ is a graph with root vertex $o\in V(G)$. 
Two rooted graphs $(G_1, o_1)$ and $(G_2, o_2)$ are \emph{isomorphic}, denoted by $(G_1, o_1) \cong (G_2, o_2)$, if there exists a bijection $\phi\colon V(G_1)\to V(G_2)$ such that $o_2 = \phi(o_1)$ and $\{u, v\}\in E(G_1)$ if and only if $\{\phi(u), \phi(v)\}\in E(G_2)$. Let $\cG_\star$ be the space of rooted graphs modulo isomorphisms; that is, all $(G', o')$ such that $(G', o')\cong (G,o)$ are considered to be the same. For any rooted graph $(G,o)$, let $B_r^G(o)$ denote the rooted subgraph induced by all vertices at graph distance at most $r$ from $o$, called the \emph{neighborhood} of radius $r$. With these definitions in place, the following lemma characterizes local weak convergence for sequences of rooted graphs.

\begin{lemma}[{\citep[Theorem 2.7]{van2024random}}]
Let %$\{G_n\}_{n \in \mathbb{N}}$ be a sequence of finite graphs, and let 
$\{(G_n, o_n)\}_{n \in \mathbb{N}}$ be the sequence of rooted graphs obtained by choosing $o_n\in V(G_n)$ uniformly at random. Then $(G_n, o_n)$ converges locally weakly to $(G,o) \sim \nu$ if and only if,
for every rooted graph $H_\star \in \cG_\star$ and all integers $r\ge 0$, 
\$
p^{G_n}(H_\star) = \frac{1}{|V(G_n)|} \sum_{u\in V(G_n)}\indc{\{B_r^{G_n}(u)\cong H_\star\}} \xlongrightarrow{} \nu(B_r^{G}(o)\cong H_\star).
\$
\end{lemma}

\iffalse
\begin{definition}[Local weak convergence]
Let $G_n = (V(G_n), E(G_n))$ denote a finite (possibly disconnected \nbr{delete this???}) graph.
Let $(G_n, o_n)$ be the rooted graph obtained by letting $o_n \in V(G_n)$
be chosen uniformly at random, and restricting $G_n$ to the connected
component of $o_n$ in $G_n$ \nbr{what is this?}. We say that $(G_n, o_n)$ \emph{converges locally weakly} to the connected \nbr{why emphasize this???}
rooted graph $(G, o)$, which is a (possibly random) element of $\mathcal{G}_\star$
having law $\mu$, when, for every bounded and continuous function
$h : \mathcal{G}_\star \to \mathbb{R}$,
\$
\mathbb{E}\bigl[h(G_n, o_n)\bigr] \;\longrightarrow\;
\mathbb{E}_{\mu}\bigl[h(G, o)\bigr],
\$
where the expectation on the right-hand side is with respect
to $(G,o)$ having law $\mu$, while the expectation on the left-hand side
is with respect to the random vertex $o_n$.
\end{definition}
\fi

Since we study random graph designs, we also introduce the notion of local convergence in probability for random graphs.
\begin{lemma}[{\citep[Theorem 2.15]{van2024random}}]\label{lmm:lwc}
Let $\{(G_n, o_n)\}_{n \in \mathbb{N}}$ be the sequence of rooted graphs obtained by choosing $o_n\in V(G_n)$ uniformly at random. Then $(G_n, o_n)$ converges locally in probability to $(G,o) \sim \nu$ if and only if, %precisely when \nbr{what does this mean?}, 
for every rooted graph $H_\star \in \cG_\star$ and all integers $r\ge 0$, 
\$
p^{G_n}(H_\star) = \frac{1}{|V(G_n)|} \sum_{u\in V(G_n)}\indc{\{B_r^{G_n}(u)\cong H_\star\}} \xlongrightarrow{\mathbb{P}} \nu(B_r^{G}(o)\cong H_\star).
\$
Here the convergence in probability is with respect to the random graph $G_n$. 
\end{lemma}

%To apply Lemma \ref{lem:matching_limit_ugw}, we first verify the local weak convergence of the random $d$-regular graph. 
%\nbr{need to define the notation $\cong$}

\iffalse
\begin{definition}[Local convergence in probability of random graphs]
Let $\{G_n\}_{n \ge 1}$ with $G_n = (V(G_n), E(G_n))$ denote a sequence of
(possibly disconnected) random graphs. Then,
$G_n$ converges \emph{locally in probability} to $(G, o)$ having law $\mu$
when, for every bounded and continuous function
$h : \mathcal{G}_\star \to \mathbb{R}$,
\$
\mathbb{E}\bigl[h(G_n, o_n)\mid G_n\bigr]
\;\xrightarrow{\mathbb{P}}\;
\mathbb{E}_{\mu}\bigl[h(G, o)\bigr],
\$
where the expectation on the left-hand side
is with respect to the random vertex $o_n$, and the random variable $\mathbb{E}\bigl[h(G_n, o_n)\mid G_n\bigr]$, which depends on the random graph $G_n$, converges in probability to $\mathbb{E}_{\mu}\bigl[h(G, o)\bigr]$. %and the random graph $G_n$. \nbr{This does not look correct...}
\end{definition}
\fi
Finally, we introduce unimodular Galton--Watson trees, which arise as local limits of many classical sparse random graphs.
\begin{definition}[Unimodular Galton--Watson Tree]
Fix a probability distribution $(\pi_k)_{k \ge 0}$, where
$\pi_k = \mathbb{P}(D = k)$ for some integer-valued random variable $D$.
The \emph{unimodular Galton--Watson tree} with root-offspring
distribution $(\pi_k)_{k \ge 0}$ is the branching process where the root
has offspring distribution $(\pi_k)_{k \ge 0}$, while all vertices in other
generations have offspring distribution $(\widehat\pi_k)_{k \ge 0}$, given by
\#
\widehat\pi_k
= \mathbb{P}(\widehat D - 1 = k)
%= \frac{k+1}{\mathbb{E}[D]} \, \mathbb{P}(D = k+1) 
= \frac{(k+1)\pi_{k+1}}{\sum_{k\ge 1}k\pi_k},
\label{eq:def-GW}
\#
where $\widehat D$ denotes the size-biased version of $D$
as defined in \eqref{eq:def-GW}.
\end{definition}

\subsubsection{Local Weak Convergence}

We establish the
local weak convergence of $\widetilde{G}_n$ in Lemma \ref{lmm:lwc_rrg} by verifying the criterion stated in \prettyref{lmm:lwc}. Recall that $\widetilde{G}_n$ is the subgraph of a random $d$-regular graph induced by the random vertex set $\supp(D).$ 
\begin{proof}[Proof of \prettyref{lmm:lwc_rrg}]
It is well established that random $d$-regular graph $G_n$ converges locally in probability to the rooted $d$-regular tree $(T_d, o)$; see, for example, \citet[Theorem 2.17]{van2024random}. In particular, let $o_n$ be a vertex chosen from $V(G_n)$ uniformly at random. We have
\# \label{eq:local-wk-regular}
\pr\left( B_r^{ G_n}(o_n) \cong T_d \right) \to 1.
\#
\iffalse
In particular, for any $\epsilon>0$, there exists $N(\epsilon) >0$ such that for any $n>N$ and $v\in V(G_n)$, we have
\# \label{eq:local-wk-regular}
\left|\pr(B_r^{G_n}(v) \cong T_d) - 1\right| < \epsilon.
\#
\xn{I rewrite the proof a little bit by emphasizing the conditional event $\{B_r^{G_n}(v) \cong T_d\}$ holds for any $v$ uniformly (with the same $N$ for all $v$ in the $\epsilon, \delta$ statement). This helps to establish the convergence at a uar $\widetilde{o}_n$ from $\widetilde{G}_n$ uniformly for all $\widetilde{G}_n$.}
\fi

We now consider $\widetilde G_n$, obtained by independently deleting each vertex in $G_n$ with probability $1-p$.
\iffalse
For each surviving vertex $v$ in $\widetilde G_n$, its local neighborhood in $G_n$ satisfies $\pr(B_r^{G_n}(v)\cong B_r^{T_d}(o)) \to 1$ for any integer $r\ge 1$. Since the vertex deletions are independent of the graph, the local neighborhood $B_r^{\widetilde G_n}(v)$ converges to $(T_d, o)$ after independent vertex deletion. 

It is well-known that the latter is given by the unimodular Galton--Watson tree with offspring distribution $\text{Bin}(d, p)$ \citep{lyons2017probability}. Thus, we conclude the proof of the lemma.
\fi
In view of~\prettyref{lmm:lwc}, it suffices to show that for any rooted tree $H_\star$ and any radius $r\ge 1$,
\begin{align}
p^{\widetilde{G}_n}(H_\star) = \frac{1}{|V(\widetilde G_n)|} \sum_{v\in V( \widetilde G_n)}\indc{\{B_r^{\widetilde G_n}(v)\cong H_\star\}} \xlongrightarrow{\mathbb{P}} \nu(B_r^{ G}(o)\cong H_\star),
\label{eq:lwc_d_regular_tree}
\end{align}
where $(G, o) \sim \nu$ is a unimodular Galton--Watson tree with offspring distribution $\Binom(d,p)$. We will apply a standard second-moment method. First, we show that
$$
\E\big[p^{\widetilde{G}_n}(H_\star)\big] \to \nu(B_r^{ G}(o)\cong H_\star).
$$
Let $\widetilde o_n$ be a vertex chosen from $V(\widetilde G_n)$ uniformly at random. 
By the definition of $p^{\widetilde{G}_n}(H_\star)$, we have
\$
\E\big[p^{\widetilde{G}_n}(H_\star)\big] & = \pr\left(B_r^{\widetilde G_n}(\widetilde o_n)\cong H_\star\right) \\
& = \pr\left(B_r^{\widetilde G_n}(\widetilde o_n)\cong H_\star \Biggiven B_r^{ G_n}(\widetilde o_n) \cong T_d \right) \pr\left( B_r^{ G_n}(\widetilde o_n) \cong T_d \right)  \\
&\qquad + \pr\left(B_r^{\widetilde G_n}(\widetilde o_n)\cong H_\star, B_r^{ G_n}(\widetilde o_n) \not\cong T_d \right) \\
& = \nu\left(B_r^{G}(o)\cong H_\star \right) \pr\left( B_r^{ G_n}(\widetilde o_n) \cong T_d \right) + \pr\left(B_r^{\widetilde G_n}(\widetilde o_n)\cong H_\star, B_r^{ G_n}(\widetilde o_n) \not\cong T_d \right) \\
& \to \nu\left(B_r^{G}(o)\cong H_\star \right),
\$
\iffalse
\$ 
\E\big[p^{\widetilde{G}_n}(H_\star)\big] & = \E_{V(\widetilde G_n)}\Bigg[\frac{1}{|V(\widetilde G_n)|} \sum_{v\in V( \widetilde G_n)}\pr\left(B_r^{\widetilde G_n}(v)\cong H_\star\right) \Bigg] \\
& = \E_{V(\widetilde G_n)}\Bigg[\frac{1}{|V(\widetilde G_n)|} \sum_{v\in V( \widetilde G_n)}\pr\left(B_r^{\widetilde G_n}(v)\cong H_\star \Biggiven B_r^{ G_n}(v) \cong T_d \right) \pr\left( B_r^{ G_n}(v) \cong T_d \right) \Bigg] \\
&\qquad + \E_{V(\widetilde G_n)}\Bigg[\frac{1}{|V(\widetilde G_n)|} \sum_{v\in V( \widetilde G_n)}\pr\left(B_r^{\widetilde G_n}(v)\cong H_\star, B_r^{ G_n}(v) \not\cong T_d \right) \Bigg] \\
& = \nu\left(B_r^{G}(o)\cong H_\star \right) \cdot \E_{V(\widetilde G_n)}\Bigg[\frac{1}{|V(\widetilde G_n)|} \sum_{v\in V( \widetilde G_n)} \pr\left( B_r^{ G_n}(v) \cong T_d \right) \Bigg]  \\
&\qquad + \E_{V(\widetilde G_n)}\Bigg[\frac{1}{|V(\widetilde G_n)|} \sum_{v\in V( \widetilde G_n)}\pr\left(B_r^{\widetilde G_n}(v)\cong H_\star, B_r^{ G_n}(v) \not\cong T_d \right) \Bigg] \\
& \to \nu\left(B_r^{G}(o)\cong H_\star \right),
\$
\fi
%\xn{Subtlety here on the second line where $o_n$ is uar from $V(\widetilde{G}_n)$. But it should be very similar.} 
where the third equality follows because $B_r^{ G_n}(\widetilde o_n) \cong T_d$ is a tree and the vertex deletions are independent of the graph, the local neighborhood $B_r^{\widetilde G_n}(\widetilde o_n)$ converges to $(T_d, o)$ after independent vertex deletions. This limit is the unimodular Galton--Watson tree with degree distribution $\text{Bin}(d, p)$ \citep{lyons2017probability}. % \nbr{I am confused. Isn't this your definition?} 
The final convergence follows from \eqref{eq:local-wk-regular}.

\iffalse
\begin{align*}
\expect{p^{\widetilde{G}_n}(H_\star)} 
%&=\frac{1}{|V(\widetilde G_n)|} \sum_{v\in V( \widetilde G_n)} \prob{B_r^{\widetilde G_n}(v)\cong H_\star} \\
& = \prob{B_r^{\widetilde G_n}(o_n)\cong H_\star} \\
& = \prob{B_r^{\widetilde G_n}(o_n)\cong H_\star \mid   
B_r^{ G_n}(o_n) \cong T_d} 
\prob{B_r^{ G_n}(o_n) \cong T_d} \\
 & \quad + \prob{B_r^{\widetilde G_n}(o_n)\cong H_\star,  
B_r^{ G_n}(o_n)     \not\cong T_d} \\
& = \nu(B_r^{ G}(o)\cong H_\star) \prob{B_r^{ G_n}(o_n) \text{ is a tree} }
+\prob{B_r^{\widetilde G_n}(o_n)\cong H_\star,  
B_r^{ G_n}(o_n) \text{ is not a tree}},
\end{align*}
$B_r^{ \widetilde{G}_n}(o_n)$ is obtained
from $B_r^{ G_n}(o_n) $  by deleting each non-root vertex in the tree independently with probability $p$.
It follows that 
\begin{align*}
\expect{p^{\widetilde{G}_n}(H_\star)}
- \nu(B_r^{ G}(o)\cong H_\star)
& = 
-\nu(B_r^{ G}(o)\cong H_\star)
\prob{B_r^{ G_n}(o_n) \text{ is not a tree} } \\
& \quad + \prob{B_r^{\widetilde G_n}(o_n)\cong H_\star,  
B_r^{ G_n}(o_n) \text{ is not a tree}} \\
& \to 0,
\end{align*}
where the last line follows from 
$\prob{B_r^{ G_n}(o_n) \text{ is not a tree}}\to 0.$
\fi

Next, we show $\var(p^{\widetilde{G}_n}(H_\star)) \to 0$. By definition, we have 
\begin{align*}
\E\big[\big(p^{\widetilde{G}_n}(H_\star)\big)^2\big] &= \prob{B_r^{\widetilde G_n}(o_n^{(1)})\cong H_\star, B_r^{\widetilde G_n}(o_n^{(2)})\cong H_\star},
\end{align*}
where $o_n^{(1)}$
and $o_n^{(2)}$  are two vertices chosen from $V(\widetilde G_n)$   independently and uniformly at random. We define the event $\calE$ such that  
$B_r^{G_n}(o_n^{(1)}) \cong T_d$
and $B_r^{G_n}(o_n^{(2)}) \cong T_d$ are vertex-disjoint trees. Then 
\begin{align*}
\prob{B_r^{\widetilde G_n}(o_n^{(1)})\cong H_\star, B_r^{\widetilde G_n}(o_n^{(2)})\cong H_\star} 
& \le  \prob{B_r^{\widetilde G_n}(o_n^{(1)})\cong H_\star, B_r^{\widetilde G_n}(o_n^{(2)})\cong H_\star \Biggiven
\calE} \prob{\calE} 
+ \prob{ \calE^c} \\
& = \nu^2(B_r^{ G}(o)\cong H_\star) \prob{\calE} 
+ \prob{ \calE^c},  
\end{align*}
where the last equality holds because when $B_r^{G_n}(o_n^{(1)})$
and $B_r^{G_n}(o_n^{(2)})$ are vertex-disjoint  trees,
$B_r^{ \widetilde{G}_n}(o_n^{(1)})$ 
and $B_r^{\widetilde G_n}(o_n^{(2)})$ are obtained
from $B_r^{G_n}(o_n^{(1)})$
and $
B_r^{G_n}(o_n^{(2)})$ by deleting each non-root vertex in the two trees independently with probability $1-p$.

Next, we prove that $\prob{\calE^c} \to 0$. We define event $\calE_1$
(resp.\ $\calE_2$)
such that $B_{2r}^{ {G}_n}(o_n^{(1)})$ (resp.\ $B_{2r}^{ {G}_n}(o_n^{(2)})$) is a tree. Then 
\begin{align*}
\prob{\calE^c}
\le  \prob{ \calE_1^c }
+ \prob{ \calE_2^c} 
+ \prob{\calE^c \mid \calE_1 \cap \calE_2 }.
\end{align*}
Note that $\prob{\calE_1^c} \to 0$ and $\prob{\calE_2^c} \to 0$ due to \eqref{eq:local-wk-regular}. Moreover, 
$$
\prob{\calE^c \mid \calE_1 \cap \calE_2 }
= 
\prob{o_n^{(2)} \in V\left(B_{2r}^{ {G}_n}(o_n^{(1)})\right)}
= \frac{\Big|V\left(B_{2r}^{ {G}_n}(o_n^{(1)})\right)\Big|}{n}
\le \frac{d^{2r}}{n}  \to 0.
$$
Therefore, we have proved that $\prob{\calE^c} \to 0$. It follows that 
$$
\E\big[\big(p^{\widetilde{G}_n}(H_\star)\big)^2\big]=\prob{B_r^{\widetilde G_n}(o_n^{(1)})\cong H_\star, B_r^{\widetilde G_n}(o_n^{(2)})\cong H_\star}
\to \nu^2(B_r^{ G}(o)\cong H_\star),
$$
and hence
$$
\var\big(p^{\widetilde{G}_n}(H_\star)\big) =\E\left[\big(p^{\widetilde{G}_n}(H_\star)\big)^2\right]
-
\left(\E\big[p^{\widetilde{G}_n}(H_\star)\big]\right)^2 \to 0.
$$
Finally, applying Chebyshev's inequality leads to the desired~\prettyref{eq:lwc_d_regular_tree}.
\end{proof}

\subsubsection{Maximum Value of $F_{d,p}$}
We now prove $\max_{t \in [0,1]} F_{d,p}(t)=(1+o_d(1)) (1-p)^d$ for large constant $d$, as shown in \prettyref{lmm:fdddd}. The following lemma characterizes local extrema of $F$ for a generic degree distribution $\pi.$ We say $t$ is a local extremum of $F$, if $F'(t)=0.$

\begin{lemma}\label{lmm:F_extrema}
Given a degree distribution $\pi$, let $\phi(t)=\sum_k \pi_k t^k$ and  $F$ be defined as in~\prettyref{eq:def_F}. If $\pi_k>0$ for $k=1,2,3$
and $\log \phi''(t)$ is strictly concave, then
$F$ has either one or three local extrema, with the first extremum always being a global maximum.
\end{lemma}
\begin{proof}
To study local extrema of $F(t)$, we take the derivative of $F(t)$ and get
\begin{align}
F'(t)=\phi''(1-t)\left[\frac{\phi'\Big(1-\frac{\phi'(1-t)}{\phi'(1)}\Big)}{\phi'(1)}-t\right].
\label{eq:F_derivative}
\end{align}
By definition, $\phi(t)=\expect{t^X}$ where $X$ follows the degree distribution $\pi.$ Thus, $\phi'(t)=\expect{X t^{X-1}}>0$, $\phi''(t)=\expect{X(X-1)t^{X-2}}>0$, and $\phi'''(t)=\expect{X(X-1)(X-2)t^{X-3}}>0$ by the assumptions that
$\pi_k>0$ for $k=1,2,3.$ In particular, $\phi'(0)=\pi_1>0$, and $\phi'(1)=\expect{X}>0$.
Therefore, $t$ is a local extremum of $F$ if and only if
$$
\frac{\phi'\Big(1-\frac{\phi'(1-t)}{\phi'(1)}\Big)}{\phi'(1)}=t,
$$
or equivalently 
$f(f(t))=t$, where $f(t)=\frac{\phi'(1-t)}{\phi'(1)}$. Let $g(t)=f(f(t))$. Next, we prove that $g$ has either one or three fixed points. Note that 
$$
f'(t)=- \frac{\phi''(1-t)}{\phi'(1)}<0, \quad  \text{ and } f''(t)=\frac{\phi'''(1-t)}{\phi'(1)}>0,
$$
Thus, $f$ is a strictly decreasing, strictly convex, differentiable function in $[0,1].$ Moreover,
$f(0)=1$, $f(1)=\phi'(0)/\phi'(1) \in (0,1).$ Thus, $f$ has a unique fixed point, denoted by $t_1$, that is $f(t_1)=t_1.$ Clearly, 
$t_1$ is also a fixed point for $g.$

Now, suppose $g$ has an additional fixed point $t_0 \neq t_1$.  Let $t_2=f(t_0) \neq t_0.$ Then $f(t_2)=f(f(t_0))=t_0$ and
$f(f(t_2)) = f(t_0)=t_2$, so that $t_2$ is also a fixed point for $g.$ Without loss of generality, let us assume $t_0<t_2$; otherwise, we just switch the roles of $t_0$ and $t_2.$ %Note that $f(t)<t$ for $t<t_1$ and $f(t)>t$ for $t>t_2$. 
Note that $f(t_0)= t_2 > t_0$, $f(t_2)= t_0 < t_2$, and $f(t_1) = t_1$, where $f$ is strictly decreasing. Thus, $t_0<t_1<t_2.$
%Moreover, the secant line on $[t_0,t_1]$ connecting $f(t_0)$ and $f(t_1)$ has slope 
%$$
%\frac{f(t_2)-f(t_0)}{t_2-t_0} =-1.
%$$
%Since $f''(t)>0$, it follows that $[t_0,t_2]$ is the unique subinterval in $[0,1]$ where the secant slope equals $-1.$ Thus, such pair of fixed points $\{t_0,t_2\}$ must be unique. \nb{WJ: I don't think this can be deduced that this interval is unique.}
Next, we show such a pair of fixed points $\{t_0,t_2\}$ must be unique. Direct computation yields that 
$g'(t)=f'(f(t))f'(t)>0$ and 
$$
g''(t)=f''(f(t))\left(f'(t)\right)^2
+f'(f(t))f''(t)
={f'(f(t))\left(f'(t)\right)^2}
\left[ \frac{f''(f(t))}{f'(f(t))} + \frac{f''(t)}{\left(f'(t)\right)^2}\right].
$$
%\nb{$g''(t)=f'(f(t))(f'(t))^2[...]$ }\nbr{I think this should be correct.} 
Define
$$
A(t)\triangleq \frac{f''(t)}{f'(t)} = - \frac{\phi'''(1-t)}{\phi''(1-t)}
\quad \text{ and }  \quad 
B(t) \triangleq \frac{f''(t)}{(f'(t))^2} = \phi'(1) \frac{\phi'''(1-t) }{
\left(\phi''(1-t)\right)^2}.
$$
Therefore, $g''(t)=0$ if and only if 
$A(f(t))+B(t)=0$. 
By the assumption that $\log (\phi''(t))$ is strictly concave, we have
$\phi'''(t)/\phi''(t)$ is decreasing. 
Therefore, $A(t)$ is decreasing and hence 
$A(f(t))$ is increasing. Moreover, since $1/\phi''(1-t)$ is increasing, $B(t)$ is also increasing. Thus, $A(f(t))+B(t)$ is increasing. 
Therefore, there exists at most one point at which $A(f(t))+B(t)=0$, implying that $g''(t)$ also has at most one zero.
By Rolle's theorem, it follows that there are at most two distinct points at which $g'(t)-1=0$, which further implies that there are at most three distinct points at which $g(t)-t=0$.

% $$
% \left(\log (\phi''(t))\right)'' = 
% \left(\frac{\phi'''(t)}{\phi''(t)}\right)'
% =\frac{1}{\left(\phi''(t)\right)^2}
% \left( \phi''''(t)\phi''(t) - \left(\phi'''(t)\right)^2 \right) \le 0.
% $$
% Hence $\phi''''(t)\phi''(t) \le \left(\phi'''(t)\right)^2.$ It follows that 
% $$
% \left(A(f(t))\right)'= A'(f(t)) f'(t) = 

% $$

% Therefore, $f''(t)/f'(t)=-\phi'''(1-t)/\phi''(1-t)$ is also non-increasing.

% Note that 
% $$
% [A(f(t))]'= A'(f(t))f'(t) = 
% $$

% We can check that $g''(t)$ is non-increasing. To see this, 
%  Since $f(t)$ is decreasing,
% it follows that $f''(f(t))/f'(f(t))$ is non-decreasing. Moreover, $f'(t)$ is increasing, so that 
% $f''(t)/(f'(t))^2$ is non-increasing. 

Finally, we prove that the smallest fixed point of $g$, \ie, the smallest local extremum of $F,$ denoted by $t_*$, must be a global maximum of $F$. 
Note that 
$F'(0)=\phi''(1) \phi'(0)/\phi'(1)>0$
and $F'(1)<0$. Since $t_*$ is the first point at which $F'(t)$ crosses zero, it follows that $F'(t)>0$ for $t<t_*$. Therefore, 
$t_*$ must be a local maximum of $F$. 
Now, if $t_*=t_1$ is the unique fixed point of $g$, then $t_*$ is also the unique extremum point of $F$, so clearly $t_*$ is the global maximum of $F$. Otherwise, $g$ has three fixed points, $t_*=t_0<t_1<t_2$, which are the three local extrema of $F$. Thus $F'(t)$ crosses zero three times at $t_0, t_1, t_2,$ so that $F'(t)>0$ for $[0,t_0)\cup (t_1,t_2)$ and $F'(t)<0$
for $t \in (t_0,t_1)\cup (t_2,1]$. 
Thus, $t_0,t_2$ are local maxima and $t_1$ is the local minimum.  Finally, 
\begin{align*}
   F(t_0)= &t_0\phi'(1-t_0) + \phi(1-t_0) + \phi\left(1 - \frac{\phi'(1-t_0)}{\phi'(1)}\right) - 1\\
  % = & f(t_2) \phi'\left(1-f(t_2)\right) + \phi 
  % \left(1-f(t_2) \right) + \phi(1-t_2) -1 \\
   =&\frac{\phi'(1-t_2)}{\phi'(1)}\phi'(1-\frac{\phi'(1-t_2)}{\phi'(1)})+ \phi\left(1 - \frac{\phi'(1-t_2)}{\phi'(1)}\right)+\phi(1-t_2) - 1\\
   =&\phi'(1-t_2) t_2 + \phi\left(1 - \frac{\phi'(1-t_2)}{\phi'(1)}\right)  + \phi(1-t_2)- 1 
   =F(t_2),
\end{align*}
where the second equality holds by plugging in $t_0=f(t_2)$ and $f(t_0)=t_2$, and the third equality holds because $f(f(t_2))=t_2.$ 
Therefore, $t_0$ and $t_2$ are global maxima. 
\end{proof}

With~\prettyref{lmm:F_extrema}, we show \prettyref{lmm:fdddd} by analyzing the smallest local extremum of $F.$
\begin{proof}[Proof of \prettyref{lmm:fdddd}]
For ease of notation, we abbreviate $F_{d,p}$ as $F$. 
Recall that the probability generating function of $\text{Bin}(d,p)$ is $\phi(t)=(pt+q)^d$. Thus,
$\phi'(t)=dp(pt+q)^{d-1}$ and $\phi''(t)=d(d-1)p^2(pt+q)^{d-2}.$ 
Then $\log \phi''(t)$ is strictly concave. It follows from~\prettyref{lmm:F_extrema} that the smallest local extremum, denoted by $t_*$, is the global maximum. Thus, it suffices to bound $F(t_*).$
We first bound $t_*.$
Recall that $t_*$ is the smallest fixed point of $g(t)=f(f(t))$, where 
$$
f(t)=(1-pt)^{d-1}, \quad \text{ and } \quad g(t)=
\left(1-p(1-pt)^{d-1} \right)^{d-1}.
$$
Note that 
$$
g(t) \le \left( 1 - p \left( 1-pt(d-1)\right)  \right)^{d-1}
= q^{d-1} \left(1 + \frac{p^2(d-1)t}{q} \right)^{d-1}
\le q^{d-1} \left( 1+  \frac{2p^2(d-1)^2t}{q} \right),  
$$
where the first inequality follows from $(1-x)^a \ge 1-ax$ and the last inequality holds for all $t$
such that $p^2(d-1)^2 t/q \le 1$ because $(1+x)^a
\le 1+2ax$ for all $x \le 1/a$. Therefore,
$$
g(t)-t
\le q^{d-1} - \left( 1 - 2p^2(d-1)^2 q^{d-2} \right) t.
$$
Choose 
$$
t_0 \triangleq  \frac{q^{d-1}}{1 -2p^2(d-1)^2 q^{d-2} }. 
$$
For all sufficiently large $d$, we have $p^2(d-1)^2q^{d-2}\leq 1/3$, which implies that $p^2(d-1)^2t_0/q\leq 1$ and hence $g(t_0)-t_0\leq 0$. On the other hand, $g(0)=q^{d-1}>0$. Since $t_*$ is the first point at which $g(t)-t$ crosses zero, it follows that $t_*\leq t_0$.

Next, we bound $F(t_*).$ By the mean-value theorem,
$
F(t_*) = F(0) + F'(\xi) t_*,
$
for some $\xi \in [0, t_*].$ Note that $F(0)=q^d$,
so it remains to bound $F'(\xi)$. Since $F'(0)=\phi''(1) \phi'(0)/\phi'(1)>0$, $t_*$ is the first point at which $F'$ crosses zero, it follows that
$F'(\xi)\ge 0.$ Recall the expression of $F'(t)$ as given in~\prettyref{eq:F_derivative}, we have
$$
F'(\xi) \le \phi''(1-\xi) g(\xi) \le 
 \phi''(1) g(t_*) \le \phi''(1) t_*
 \le d(d-1)p^2 t_0.
$$
where the second inequality holds because $\phi''$ and
$g$ are increasing functions. 
Therefore, 
$$
q^d \le F(t_*) \le q^d + d(d-1)p^2 t_0^2
\le q^d \left[ 1 + \frac{d(d-1)p^2 q^{d-2}}{
\left(1 -2p^2(d-1)^2 q^{d-2}\right)^2} \right]
\le q^d \left( 1+ 9 d^2 p^2 q^{d-2} \right).
$$
It follows that $F(t_*)/q^d \to 1$ as $d \to \infty.$

Finally, we show that $F_{d,p}^* \le \epsilon$ for $d$ in \eqref{eq:match-const-d} with $c = [6p/(q\log(1/q))]^2$ and sufficiently small $\epsilon$. Applying
\prettyref{lmm:fdddd} with $d \ge (1+c\epsilon\log(1/\epsilon))\log(1/\epsilon)/\log(1/q)$, we have
\iffalse
\$F_{d,p}^* & \le q^d ( 1+ 9d^2 p^2 q^{d-2}) \le \epsilon^{1+c\epsilon\log(1/\epsilon)}\left[1+c\epsilon^{1+c\epsilon\log(1/\epsilon)}\log^2(1/\epsilon) \right] \\
& = \epsilon \left[1 - c\epsilon\log^2(1/\epsilon) + c^2\epsilon^2\log^4(1/\epsilon)/2 +  o\left(\epsilon^2\log^4(1/\epsilon)\right) \right] \left[1+c\epsilon^{1+c\epsilon\log(1/\epsilon)}\log^2(1/\epsilon) \right] \\
& \le \epsilon \left[1 - c\epsilon\left(1 -  \epsilon^{c\epsilon\log(1/\epsilon)}\right)\log^2(1/\epsilon) + c^2\epsilon^2\log^4(1/\epsilon)/2 +  o\left(\epsilon^2\log^4(1/\epsilon)\right) \right] \\
& = \epsilon \left[1 - c^2\epsilon^2\log^4(1/\epsilon)/2 +  o\left(\epsilon^2\log^4(1/\epsilon)\right) \right] \le \epsilon\$
\fi
%\nbr{The above can be simplified as follows: 
\$
F_{d,p}^*  \le q^d ( 1+ 9d^2 p^2 q^{d-2}) \le \epsilon^{1+c\epsilon\log(1/\epsilon)}\left[1+c\epsilon^{1+c\epsilon\log(1/\epsilon)}\log^2(1/\epsilon) \right] 
 = \epsilon e^{-x} \left(1+x e^{-x} \right)  \le \epsilon,
\$
where the second inequality holds when $c\epsilon \log(1/\epsilon)<1$; the equality holds by taking 
$x=c\epsilon \log^2(1/\epsilon)$
so that 
$\epsilon^{c\epsilon\log(1/\epsilon)} = e^{-x}$; the last inequality holds because $1+xe^{-x} \le 1+ x \le e^{x}$ for all $x\ge 0.$ 
%}
%where  the equalities follow from the Taylor expansion of $e^x$ at $x=0$ applied to $\epsilon^{c\epsilon\log(1/\epsilon)} = e^{-c\epsilon\log^2(1/\epsilon)}$; and the last inequality holds for sufficiently small $\epsilon$.
\end{proof}

\subsubsection{Proof of Theorem \ref{thm:constant d} via \prettyref{prop:match-const}} Finally, we prove Theorem \ref{thm:constant d}.
\begin{proof}[Proof of Theorem \ref{thm:constant d}]
\prettyref{prop:match-const} yields that
\$ 
\lim_{n\to\infty} \expects{|L(G,D)/(np) - F^*_{d,p}|}{G,D} = 0.
\$
Recall that $F_{d,p}^* \le \epsilon$ from \prettyref{lmm:fdddd}. Thus, by Markov's inequality, we have
\$ 
\probs{L(G) \ge \epsilon n}{G} \le \probs{|L(G)/(np) - F^*_{d,p}| \ge q\epsilon/p }{G} \le \frac{p\E_{G}\big[|L(G)/(np) - F^*_{d,p}|\big]}{q\epsilon}.
\$
Recall that $L(G) = \E_D[L(G,D)]$. Then by Jensen's inequality, we have $|L(G)/(np) - F^*_{d,p}| \le \E_D[|L(G,D)/(np) - F^*_{d,p}|]$ and
\$ 
\expects{|L(G)/(np) - F^*_{d,p}| }{G} \le \expects{|L(G,D)/(np) - F^*_{d,p}| }{G,D}. 
\$
Combining these three equations, we conclude that $ 
\lim_{n\to\infty} \probs{L(G) \ge \epsilon n}{G}= 0$.
\iffalse
In other words, for any $\epsilon,\kappa>0$, there exists $N>0$ such that for any $n\ge N$,
\$
\probs{\left|{L(G_{n,d}, D)}/(np) - F_{d,p}^*\right| \ge \epsilon(1-p)/(2p)}{G,D} \le \kappa\epsilon(1-p)/2.
\$
Note that $F_{d,p}^* \le q^d ( 1+ 9d^2 p^2 q^{d-2}) \le \epsilon^{1+\epsilon}(1+o_\epsilon(1))\le \epsilon$ for $d$ in \eqref{eq:match-const-d} and sufficiently small $\epsilon$. Thus, $L(G_{n,d}, D) \ge \epsilon n(1+p)/2 \ge [F_{d,p}^* + \epsilon(1-p)/(2p)]
np$ implies that $|{L(G_{n,d}, D)}/(np) - F_{d,p}^* | \ge \epsilon(1-p)/(2p)$. Consequently,
\#\label{eq:prob-gd-loss}
\probs{L(G_{n,d}, D)\ge \epsilon n(1+p)/2}{G,D} \le \kappa\epsilon(1-p)/2.
\#
Moreover, since $L(G,D)\le n$ always holds, conditioned on $G$, we have
\$
\E_{D}[L(G, D)] & = \E_{D}\big[L(G, D)\indc{L(G,D)\ge \epsilon n(1+p)/2}\big] 
+ \E_{D}\big[L(G, D)\indc{L(G,D)< \epsilon n(1+p)/2}\big] \\
& \le  n \probs{L(G,D)\ge \epsilon n(1+p)/2}{D} + \epsilon n(1+p)/2.
\$  
Thus, $\E_{D}[L(G, D)] \ge \epsilon n$ implies that $\probs{L(G,D)\ge \epsilon n(1+p)/2}{D} \ge \epsilon (1-p)/2$.  Therefore, we have
\$ 
& \probs{\E_{D}[ L(G, D)] \ge \epsilon n}{G} 
\le \probs{\probs{L(G,D)\ge \epsilon n(1+p)/2}{D} \ge \epsilon (1-p)/2}{G} \\
& \qquad \qquad \le \frac{2\E_G\left[\probs{L(G,D)\ge \epsilon n(1+p)/2}{D}\right] }{\epsilon(1-p)}  = \frac{2\probs{L(G,D)\ge \epsilon n(1+p)/2}{G,D}}{\epsilon(1-p)} \le \kappa,
\$
where the second inequality holds by Markov's inequality, the equality by the law of total probability, and the third inequality by \eqref{eq:prob-gd-loss}.
\fi
\end{proof}

\section{Tight Bounds on Matching Loss for $K$-Chains}\label{sect:app-chain}

In this section, we provide a tight characterization for the matching loss of $K$-chain graphs. 
%\begin{definition}
A $K$-chain graph on $n$ vertices, denoted by $C_{n,K}$, is defined as follows: place the vertices $1,\cdots, n$ on a circle in sequential order. Two distinct vertices $i, j\in [n]$ are connected if and only if their circular distance satisfies $d(i,j) \triangleq \min\{|i-j|, n-|i-j|\} \le K$.
%\end{definition}
These graphs are also known as $2K$-nearest neighbor graphs on a circle.

The motivation for studying the $K$-chain designs comes from the process flexibility design literature \citep{jordan1995principles}. 
%The $K$-chain graphs are regular graphs. 
%K-chain: benefits and definition %\nbr{move it to the appendix} \nbr{check with Yehua on the terminology of $k$-chains???} To address the location information in the graph, we consider the KNN(K Nearest Neighbour) designs. We denote the graph with $n$ nodes and each node connected with its surrounding $k$ nodes, left and right separately as $C_{n,k}$. 
The existing work ~\citep[Corollary~1]{feng2024designing} shows that the matching loss of $K$-chains satisfies $L(C_{n,K}) \leq 4nq^{d/4}/d$. We improve this upper bound and further derive a nearly matching lower bound, obtaining a tight characterization of $L(C_{n,K})$.
\begin{theorem}\label{thm:KNN}
%Fix any constant $c>0$ and sufficiently large $n$ \nbr{where do you need this???}. For any $K$ satisfying $n pq^K \ge c$, 
For any $K$ and $n$, we have
\$
%npq^{K}\bigg(\frac{(1-q^K)^2}{2-q^K} - e^{-c/2}\bigg) - 1
npq^{K}\bigg(\frac{(1-q^K)^2}{2-q^K} - \exp(-(n-3K-2)pq^K/2)\bigg) - 1 \le L(C_{n,K}) \le n p q^K \bigg(  \frac{(1-q^K)^2}{2-q^K} + q^{K}+ 1 \bigg) + 1.
\$
%For any constant $K$, there exists $N>0$ such that for any $n \ge N$, \$L(C_{n,K})=(1/2 \pm cq^K) npq^K =(1/2 \pm cq^{d/2}) npq^{d/2},\$ where $c$ is a constant independent of $n$ and $K$. \xn{It also works for $K = \log n/\log(1/q)$.}
\end{theorem}
By Theorem \ref{thm:KNN}, when $K$ satisfying $npq^K \ge 2$, the matching loss of $K$-chains is of order $npq^K = npq^{d/2}$. To ensure the loss bounded by an $\epsilon$-fraction $\epsilon n$, it is necessary that $d = 2K \ge 2\log(1/\epsilon)/\log(1/q)$. Similarly, to ensure a loss of at most a constant (independent of $n$), one must have $d =2K\ge 2\log(n)/\log(1/q)$.

%\nbr{need to add some intuition on the suboptimality of $k$-chains compared to random regular graphs.}
Recall that optimal matching loss is achieved by random-$d$ regular graphs of order $npq^d$.
Although $K$-chains are also regular graphs, the connections are restricted to the predetermined sets of neighboring vertices, which limits design flexibility and leads to a suboptimal matching loss.

To build intuition, observe that the matching loss is at least the expected number of odd components, since each odd component gives at least one unmatched vertex in any matching. In fact, as shown in Lemma \ref{lem:x-odd}, these two quantities coincide for $K$-chains. Moreover, the number of odd components is of the same order as the total number of connected components, which in turn is approximately the expected number of fully deleted $K$-segments. The latter is of order $npq^K$, as each vertex marks the end of a fully deleted $K$-segment if it lies in $\supp(D)$ while all of its $K$ predecessors lie outside $\supp(D)$, an event that occurs with probability $pq^K$. This order nearly matches the loss bounds established in \prettyref{thm:KNN}.
%easy to compute, as it corresponds to the expected number of fully deleted $K$-segments. Specifically, a new component starts at $i$ when $i\in\supp(D)$ while its $K$ predecessors are outside $\supp(D)$. This occurs with probability $pq^K$. Thus, we have \$ \E[c(C_{n,k}[\supp(D)])] = \sum_{i=1}^n \prob{\text{A component starts at vertex }i} = npq^K.\$
%Accordingly, the expected number of odd components, which determines the loss, is approximately half the expected number of components, namely about $npq^K/2$. 

In comparison, the lower bound in \prettyref{thm:lower_bound} comes from isolated vertices. In $K$-chain, a vertex is isolated when its $K$ predecessors and $K$ successors are all outside $\supp(D)$, which occurs with probability $pq^{2K}$. Thus, the expected number of isolated vertices in $C_{n,K}[\supp(D)]$ is $npq^{2K}$. %given by $\sum_{i=1}^n \prob{ \text{Vertex } i \text{ is isolated}} = npq^{2K}$. 
Hence, the expected number of odd components in $K$-chains is much larger than the number of isolated vertices, which explains the larger matching loss. In contrast, for random $d$-regular graphs, which achieve the optimal loss, these two quantities are asymptotically equal. %This is also consistent with Lemma \ref{eq:no_small_component}, which shows that random  $d$ regular graphs have no small components.} \nb{Probably delete the last sentence...}

\subsection{Proof of Theorem \ref{thm:KNN}}
Theorem \ref{thm:KNN} follows by observing that, for any induced subgraph $C_{n,K}[V]$, the number of unmatched vertices in a maximum matching equals the number of odd components, due to the one-dimensional geometric structure. %Recall that $o(G)$ denotes the number of odd components in the graph $G$. 
A similar result is provided for the one-dimensional random geometric graph in \citet[Proposition 1]{sentenac2025online}.

%For the upper bound, we study the residual graph resulted from randomly deleting nodes from KNN graph. The residual graph consists of a number of connected components. We apply a greedy matching rule to each connected component, in which we show that even component can be perfectly matched and odd component will have one unmatched. 

%More precisely, given a KNN graph $C_{n,k}$, let $\widetilde{C}_{n,k}$ denote the residual random subgraph after random node deletion. Let $X$ be the number of connected components in $\widetilde{C}_{n,k}$ and $X_{\odd}$ be the number of odd-sized components.

%Combining Lemmas \ref{lem:x-odd} and \ref{lem:exp-x-odd} completes the proof of Theorem \ref{thm:KNN}. the number of unmatched vertices in a maximum matching equals the number of its odd components. \cite{sentenac2025online}

\begin{lemma}\label{lem:x-odd}
For any $V \subset [n]$, %the number of unmatched vertices in a maximum matching equals the number of odd components  
we have $|V| - 2\mu(C_{n,K}[V]) = \odd(C_{n,K}[V])$.
\end{lemma}
%\nbr{let's not use $o$ notation. use $\odd$}
\begin{proof}
%(1) Each connected component contains a path (2) Greedy matching two consecutive vertices.
List the vertices of $V$ in clockwise order on the circle as $v_1, \cdots, v_m$. 
%Whenever two consecutive vertices have circular distance $>K$ in $C_{n,K}$, a new component is formed, since no edges are crossing between the two sides. 
Suppose that~$v_a, v_{a+1}, \cdots, \allowbreak v_{a+\ell}$ form a connected component. Then each pair of consecutive vertices has a circular distance at most $K$ and is thus connected. 
%Within each component, consecutive vertices always have circular distance $\le K$ in $C_{n,K}$ and are thus connected. 
%Therefore, each connected component in $C_{n,K}[V]$ contains a Hamiltonian path.
%If there are two consecutive vertices whose circular distance is greater than $K$ in $C_{n,K}$, we refer to them as endpoints. We start from one endpoint, denote it by $ v_1$, and then traverse the component sequentially to the other endpoint, labeling the vertices as $v_1, v_2, \cdots, v_{\ell}$. By definition, we have $d(v_1,v_{\ell})>K$. Then two consecutive vertices, $v_{i}$ and $v_{i+1}$, must be connected. Otherwise, within $C_{n,K}$, it must be the case that $d(v_i,v_{i+1})>K$. Then for all $v_{j}$ with $j\ge i+2$, we have $d(v_i,v_{j})\ge \min\{d(v_i,v_{i+1}), d(v_1,v_{\ell})\}>K$. If there is no such endpoints, there is only one connected component of the entire graph $C_{n,K}[V]$. 
%Since each pair of consecutive vertices is connected, there exists a Hamiltonian path from $v_a$ to $v_{a+\ell}$. Along this path 
Hence, we can construct a maximum matching by pairing consecutive vertices: $\{v_a, v_{a+1}\}, \{v_{a+2}, v_{a+3}\}, \cdots$. In this construction, every vertex in an even component is matched, while in an odd component exactly one vertex remains unmatched.
%we traverse the component clockwise: starting from the first vertex, we match it with the second, then match the third with the fourth, and continue in this manner until the last vertex of the component. This matching is valid because, in $C_{n,K}$, each vertex is connected to its $2K$ nearest neighbors. Hence in the induced subgraph $C_{n,K}[V]$, vertices remain connected to their nearest available neighbors within the component.\nbr{This sentence reads unclear. I feel that we need to show that each component contains a Hamiltonian path???}  In this way, in an even component, all vertices are matched, whereas in an odd component, exactly one vertex remains unmatched. % Due to our linkage mechanism, %\nbr{need some further clarification here. By the way, is this result known in the literature? If yes, we should add the reference} there will be one remaining unmatched point for odd components and even components get a perfect match.
\end{proof}

The matching loss differs from the expected number of odd components by at most one. Indeed, Lemma \ref{lem:x-odd} yields that $\expect{\odd(C_{n,K}[\supp(D)])} = \expect{|\supp(D)| - 2 \mu \left(C_{n,K}[\supp(D)] \right)}$. Recall that $L(C_{n,K})=  2 \expect{\floor{|\supp(D)|/2} - \mu\left(C_{n,K}[\supp(D)] \right)}$ and note that $\expect{|\supp(D)|}-1\le 2 \expect{\floor{|\supp(D)|/2}} \le \expect{|\supp(D)|}$. From these, we have
\$
\expect{\odd(C_{n,K}[\supp(D)])} -1 \le L(C_{n,K}) 
& \le \expect{\odd(C_{n,K}[\supp(D)])}. 
\$
\iffalse
Recall that the matching loss $L(C_{n,K})=  2\E[\floor{|\supp(D)|/2} - \mu(C_{n,K}[\supp(D)])]$ differs from the expected number of unmatched vertices in a maximum matching, $\E[|\supp(D)| - \mu(C_{n,K}[\supp(D)])]$, by at most one. \nbr{I do not understand the above reasoning. Isn't just 
\begin{align*}
L(C_{n,K}) & =  2 \expect{\floor{|\supp(D)|/2} - \mu\left(C_{n,K}[\supp(D)] \right)} \\
& \le \expect{|\supp(D)| - 2 \mu \left(C_{n,K}[\supp(D)] \right)} \\
& =\expect{\odd(C_{n,K}[\supp(D)])}
\end{align*}} 
\fi
It therefore remains to compute the expected number of odd components %$\E[\odd(C_{n,K}[\supp(D)])]$ 
in order to complete the proof of Theorem \ref{thm:KNN}. %This result will be shown in Lemma \ref{lem:exp-x-odd}. 
%The expected number of connected components is comparatively easy to compute, as it reduces to counting the number of gaps of length $\geq K$ on the circle. \nbr{If this is easy, maybe you could briefly explain how this is done...} \xn{intuition on $pq^k$} 
As discussed, the expected number of connected components is roughly of the order $npq^K$, which provides an immediate upper bound on the odd components. However, a more careful analysis is required to derive a nearly matching lower bound.

To this end, we adopt a generating-function approach and establish Lemma~\ref{lem:exp-x-odd}. Each odd component can be represented as a binary sequence, where a $1$ denotes a vertex outside $\supp(D)$. Such a sequence contains no $k$ consecutive $1$’s and has an odd number of $0$’s. The generating function will provide a convenient tool for summing over the odd terms only.
%By contrast, analyzing the expected number of odd components is more delicate, and 
%We adopt the generating-function approach to compute the expected number of odd components. 
To prepare, we present the following auxiliary lemma on the closed form of the generating function needed later. %\xn{To do: How does the generating function appear.} %The lemma also provides the solution to an exercise in \cite{wilf2005generatingfunctionology}. %\nbr{I thought this is known...}
Let $\alpha(i,j,k)$ denote the number of binary sequences of length $i$ that contain exactly $j$ $1$'s and no $k$ consecutive $1$'s. We define the bivariate generating function and the corresponding finite sum as %\xn{To do: notation $F_{k,n}$.}
\#\label{eq:def-gen-f-finite-remaining}
F_k(x,y) = \sum_{i= 0}^\infty \sum_{j = 0}^i \alpha(i,j,k)\, x^i y^j, \quad F_{k,n}(x,y) = \sum_{i= 0}^n \sum_{j = 0}^i \alpha(i,j,k)\, x^i y^j,
\#
and the truncation error $R_{k,n}(x,y) = F_k(x,y) - F_{k,n}(x,y)$. The following lemma gives a closed-form expression for $F_k(x,y)$ and characterizes the rate of $R_{k,n}(x,y)$. Its proof follows from the recurrence relation for $\alpha(i,j,k)$.
\begin{lemma}\label{lem:no-k-cons}
For $x,y$ such that $|x|(1+ |y|) \le 1$, the generating function admits the closed form 
\#\label{eq:gen-f-close}
F_k(x,y) = \frac{1 - x^k y^k}{1 - x - xy + x^{k+1} y^k}. 
\#
In addition, the convergence rate at $(x,y)$ satisfying $|x| = p$ and $|y|=q/p$ satisfies,
\$
|R_{k,n}(x,y)| \le \frac{1-q^k}{pq^k(1+pq^k)^{n-k}}.
\$
\end{lemma}
\begin{proof}
We prove the statement step by step. 
\paragraph{Recurrence formula for $\alpha(i,j,k)$.}
We first present formulas for $\alpha(i,j,k)$ for different ranges of $i$ and $j$.
For any $i \ge 0$ and $0\le j\le \min\{i, k-1\}$, the condition that the sequence contains no $k$ consecutive $1$'s is automatically satisfied, since there are at most $j\le k-1$ $1$'s in total. Thus, we have
\#\label{eq:gen-fn-rec1} 
\alpha(i,j,k) = \binom{i}{j}.
\#
For $i=j \ge k$, we have $\alpha(i,j,k) = 0$, since any such sequence always contains $k$ consecutive $1$'s.

For $i\ge k+1$ and $k \le j\le i-1$, we provide a recurrence formula for $\alpha(i,j,k)$ by considering the last position of a valid sequence: 
\begin{enumerate}
    \item If the last position is $0$, then the first $i-1$ positions can be any valid sequence of length $i-1$ having $j$ $1$'s. There are $\alpha(i-1,j,k)$ such sequences.
    \item If the last position is $1$, then the first $i-1$ positions must be a valid sequence of length $i-1$ having $j-1$ $1$'s. This contributes $\alpha(i-1,j-1,k)$ sequences.
    
    However, among these, we must exclude those sequences that end with $k-1$ $1$'s, since adding another $1$ at the last position will violate the condition. Indeed, any such sequence must end with $01\cdots 1$ (a $0$ followed by $k-1$ $1$'s). Thus, the first $i-k-1$ positions can be any valid sequence having $j-k$ $1$'s. Hence, there are $\alpha(i-k-1,j-k,k)$ such sequences.
\end{enumerate}
Putting them together, we have for $i\ge k+1$ and $k \le j\le i-1$ that
\#\label{eq:gen-fn-rec} 
\alpha(i,j,k) = \alpha(i-1,j,k) + \alpha(i-1,j-1,k) - \alpha(i-k-1,j-k,k).
\#

\paragraph{Recurrence formula for the finite sum.}
We split $F_n(x,y)$ by noting that $\alpha(i,j,k) = 0$ for $i=j\ge k$,
\#\label{eq:gen-f}
F_n(x,y) = \sum_{i = 0}^n \sum_{j = 0}^i \alpha(i,j,k)\, x^i y^j & = \sum_{i = 0}^n  \sum_{j = 0}^{\min\{i, k-1\}} \alpha(i,j,k)\, x^i y^j + \sum_{i = k+1}^n  \sum_{j = k}^{i-1} \alpha(i,j,k)\, x^i y^j. 
%\\& = \sum_{j = 0}^{k-1} \sum_{i \ge j}  \binom{i}{j}\, x^i y^j + \sum_{i \ge k+1}  \sum_{j = k}^{i-1} \alpha(i,j,k)\, x^i y^j.
\#
Next, we bound the two terms separately using the above formulas for $\alpha(i,j,k)$. Applying \eqref{eq:gen-fn-rec1} to the first term,  we have
\$
A_n(x,y) &\triangleq \sum_{i = 0}^n  \sum_{j = 0}^{\min\{i, k-1\}} \alpha(i,j,k)\, x^i y^j = \sum_{j = 0}^{k-1} \sum_{i = j}^n  \binom{i}{j}\, x^i y^j.
\$
When $|x|<1$, taking $n\to \infty$ yields
\$
A(x,y)& \triangleq \lim_{n\to \infty} A_n(x,y) = \sum_{j = 0}^{k-1} x^jy^j \sum_{m = 0}^\infty  \binom{m+j}{j}\, x^{m} = \sum_{j = 0}^{k-1} \frac{x^jy^j}{(1-x)^{j+1}} = %\frac{1-\big(\frac{xy}{1-x}\big)^k}{(1-x)\big(1-\frac{xy}{1-x}\big)} 
\frac{1-\frac{x^ky^k}{(1-x)^k}}{1-x - xy},
\$
%\nbr{We need to be careful about the last equality when $1=x+xy$. The expression is still correct with the understanding that the ratio is equal to $k$ when $1=x+xy.$ }
where the second-to-last inequality follows from the generalized binomial identity that $(1+x)^{-j} %= \sum_{m \ge 0}  \binom{-m}{j}\, x^{j} 
= \sum_{m \ge 0}  \binom{j + m-1}{m}\, (-x)^{m}$ for $|x|<1$, and note that the last equality still holds when $1=x+xy$ with the understanding that the ratio equals $k$. Moreover, following from the recurrence formula in \eqref{eq:gen-fn-rec}, we decompose the second term in \eqref{eq:gen-f} into three parts, and compute each of them in turn,
\#\label{eq:b-sum}
& B_n(x,y) \triangleq \sum_{i = k+1}^n  \sum_{j = k}^{i-1} \alpha(i,j,k)\, x^i y^j \\
& = \sum_{i = k+1}^n \sum_{j = k}^{i-1} \alpha(i-1,j,k) \, x^i y^j + \sum_{i = k+1}^n \sum_{j = k}^{i-1} \alpha(i-1,j-1,k) \, x^i y^j - \sum_{i = k+1}^n \sum_{j = k}^{i-1} \alpha(i-k-1,j-k,k) \, x^i y^j. \notag
\#
First, by reindexing, we have
\#\label{eq:b-sum-1}
\sum_{i = k+1}^n \sum_{j = k}^{i-1} \alpha(i-1,j,k) \, x^i y^j = x\sum_{i = k}^{n-1} \sum_{j = k}^{i} \alpha(i,j,k) \, x^i y^j = x\sum_{i = k+1}^{n-1} \sum_{j = k}^{i-1} \alpha(i,j,k) \, x^i y^j = x B_{n-1}(x,y).
\#
In addition, for the second part, we obtain 
\#\label{eq:b-sum-2}
\sum_{i = k+1}^n \sum_{j = k}^{i-1} \alpha(i-1,j-1,k) \, x^i y^j & = xy\sum_{i = k}^{n-1} \sum_{j = k-1}^{i-1} \alpha(i,j,k) \, x^i y^j \notag\\
& = xy\sum_{i = k}^{n-1} \binom{i}{k-1} \, x^i y^{k-1} + xy\sum_{i = k+1}^{n-1} \sum_{j = k}^{i-1} \alpha(i,j,k) \, x^i y^j \notag\\
& = x^{k}y^{k} \sum_{r = 1}^{n-k} \binom{r + k-1 }{k-1} \, x^r + xyB_{n-1}(x,y) \notag\\
& = xyB_{n-1}(x,y) + D_{n-k}(x,y),
\#
where $D_{n-k}(x,y) \triangleq x^{k}y^{k} \sum_{r = 1}^{n-k} \binom{r + k-1 }{k-1} \, x^r$. When $|x|<1$, taking $n\to \infty$ gives
\$ 
D(x,y) \triangleq \lim_{n\to \infty} D_{n-k}(x,y) = x^{k}y^{k} \sum_{r = 1}^{\infty} \binom{r + k-1 }{k-1} \, x^r = x^ky^k\bigg[\frac{1}{(1-x)^{k}} - 1\bigg] = \frac{x^ky^k}{(1-x)^{k}} - x^ky^k,
\$
where the last equality follows from the generalized binomial identity. Finally, the last part equals
\#\label{eq:b-sum-3}
\sum_{i = k+1}^n \sum_{j = k}^{i-1} \alpha(i-k-1,j-k,k) \, x^i y^j = x^{k+1} y^k \sum_{i = 0}^{n-k-1} \sum_{j = 0}^{i} \alpha(i,j,k) \, x^i y^j = x^{k+1} y^k F_{n-k-1}(x,y).
\#
Substituting \eqref{eq:b-sum-1}, \eqref{eq:b-sum-2}, and \eqref{eq:b-sum-3} into \eqref{eq:b-sum}, we have for $n\ge k+1$,
\#\label{eq:recurrence-bn} 
B_n(x,y) = (x+xy)B_{n-1}(x,y) + D_{n-k}(x,y) - x^{k+1} y^k F_{n-k-1}(x,y).
\#

\paragraph{Convergence and the limit.}
We now prove that both $B_n(x,y)$ and $F_n(x,y)$ converge for fixed $x,y$ satisfying $|x|(1+ |y|) \le 1$. Since $|B_n(x,y)|\le B_n(|x|,|y|)$ and $|F_n(x,y)|\le F_n(|x|,|y|)$, we only need to consider $x,y>0$; otherwise we take the absolute value of $x,y$. When $x,y>0$, since $B_n(x,y)$ and $F_n(x,y)$ are partial sum of nonnegative terms, it suffices to show they are bounded. When $x,y>0$ and $x(1+ y) < 1$, from \eqref{eq:recurrence-bn}, we have
\$ 
(1-x-xy)B_n(x,y) \le D_{n-k}(x,y) \le D(x,y),
\$
noting that $D_{n-k}(x,y)$ is non-decreasing in $n-k$. Thus, $B_n(x,y)$ converges and so does $F_n(x,y)$. 
When $x,y>0$ and $x+xy = 1$, from \eqref{eq:recurrence-bn}, we have
\$ 
x^{k+1} y^k F_{n-k-1}(x,y) \le D_{n-k}(x,y) \le D(x,y).
\$
Thus, it holds that $F_{n-k-1}(x,y) \le D(x,y)/(x^{k+1} y^k)$ is bounded and converges, and so does $B_n(x,y)$. 

When both $B_n(x,y)$ and $F_n(x,y)$ converge, taking $n\to\infty$ on both sides of \eqref{eq:recurrence-bn}, and substituting $B(x,y) = F(x,y) - A(x,y)$, we have
\$ 
F(x,y) - \frac{1-\frac{x^ky^k}{(1-x)^k}}{1-x - xy} = (x+xy) \bigg[F(x,y) - \frac{1-\frac{x^ky^k}{(1-x)^k}}{1-x - xy}\bigg] + \frac{x^ky^k}{(1-x)^{k}} - x^ky^k - x^{k+1} y^k F(x,y).
\$
Rearranging the terms gives the closed form of $F(x,y)$ in \eqref{eq:gen-f-close}.

%\$ F_k(x,y) = \frac{1-\frac{x^ky^k}{(1-x)^k}}{1-x - xy} + \frac{1}{1-x-xy} \Big[\frac{x^ky^k}{(1-x)^{k}} - x^ky^k - x^{k+1} y^k F_k(x,y) \Big]\$

\paragraph{Convergence rate.}
%For any $(x,y)$, we define the remainder term $R_n(x,y) \triangleq F(x,y) - F_n(x,y)$.
For $(x,y)$ such that $|x| = p$ and $|y| = q/p$, we have $%|F(x,y) - F_n(x,y)| = 
|R_n(x,y)| \le R_n(|x|, |y|) = R_n(p, q/p)$. Thus, it suffices to bound $R_n(p, q/p)$. For simplicity, we write $R_n = R_n(p, q/p)$, and similarly for other functions at $(x,y) =(p,q/p)$. Rearranging \eqref{eq:recurrence-bn} and noting $p+q=1$, we have
\$
F_n - F_{n-1} + p q^k F_{n-k-1}= A_n - A_{n-1} + D_{n-k}.
\$
Note that $F = 1-q^k/(pq^k)$ and $D= 1-q^k$.
Let $c\triangleq pq^k >0$. Then $c F = D$ and we have
\#\label{eq:delta-rec}
\Delta_n \triangleq - A_n + A_{n-1} +D - D_{n-k} = R_n - R_{n-1} + c R_{n-k-1}\ge R_n - R_{n-1} + cR_n,
\#
where the inequality follows from the fact that $R_n \le R_{n-k-1}$, since $F_n$ is non-decreasing in $n$. We now show that $\Delta_n = 0$. First, observe that
\$ 
A_n - A_{n-1} = \sum_{j = 0}^{k-1} \binom{n}{j}\, p^{n-j} q^j.
\$ 
Next, we bound the difference
\$ 
D - D_{n-k} = q^{k} \sum_{r = n-k+1}^{\infty} \binom{r + k-1 }{k-1} \, p^r = q^{k} p^{n-k+1} \sum_{r = 0}^{\infty} \binom{n + r}{k-1} \, p^r.
\$
We denote by $[x^n]f(x)$ the coefficient of $x^n$ in the series $f(x)$, and have the following identity,
\$
\sum_{r=0}^\infty \binom{n+r}{k-1} p^r
&=[t^{k-1}] \sum_{r=0}^\infty (1+t)^{n+r} p^r 
= [t^{k-1}] 
\frac{(1+t)^n}{1-(1+t)p} \\
&= [t^{k-1}] (1+t)^n \frac{1}{q}  \sum_{m=0}^\infty 
\left( \frac{tp}{q}\right)^m
=
\frac{1}{q} \sum_{j=0}^{k-1}\binom{n}{j} \left( \frac{p}{q} \right)^{k-1-j}.
\$
Thus, we have
$$
D-D_{n-k}
= q^{k-1} p^{n-k+1}
\sum_{j=0}^{k-1} 
\binom{n}{j} \left( \frac{p}{q} \right)^{k-1-j} = 
\sum_{j=0}^{k-1} 
\binom{n}{j} p^{n-j} q^{j}
= A_n - A_{n-1}.
$$
Hence, $\Delta_n = 0$ in \eqref{eq:delta-rec}, which implies the recursion $R_n \le R_{n-1}/(1+c)$ for any $n\ge k+1$. Recalling that $R_k = F - F_k \le F = (1-q^k)/(pq^k)$, we obtain the final convergence rate:
\$
R_n & \le \frac{1}{(1+c)^{n-k}} R_k \le \frac{1-q^k}{pq^k(1+c)^{n-k}}.
\$

\end{proof}

We now compute the expected number of odd components, thereby proving Theorem \ref{thm:KNN}.
\begin{lemma}\label{lem:exp-x-odd}
%Fix any constants $c,\delta>0$.
%Fix any constant $c>0$. %and sufficiently large $n$. 
For any $K$ and $n$, %satisfying $n pq^K \ge c$, %\nbr{I do not understand why $c$ cannot be 0???} 
we have
\$
npq^{K}\bigg(\frac{(1-q^K)^2}{2-q^K} - \exp(-(n-3K-2)pq^K/2)\bigg) & \le \E[\odd(C_{n,K}[\supp(D)])] \\
&\le n p q^K \bigg(q^{K} + \frac{(1-q^K)^2}{2-q^K} + 1  \bigg) + 1.
\$
\end{lemma}
\begin{proof}
%Let $C_{n,K}$ be a K-chain graph on $[n]$. %, and let $D$ be a demand realization with $\supp(D)\subset [n]$. \nbr{What does this mean? Do you need to condition on $\supp(D)?$} 
We fix one vertex in $C_{n,K}$ and denote it by $v_1$, and then label the remaining vertices clockwise as $v_2, v_3, \cdots, v_n$. For convenience, the indexing is taken modulo $n$, so that $v_{n+i} = v_i$ for all $i$. We define a maximal connected component of $\supp(D)$ as a sequence of vertices $(v_s, v_{s+1}, \cdots, v_{s+i-1}) \subset V(C_{n,K})$, for $1\le i\le n-K$, satisfying the following conditions:
\begin{enumerate}
\item Endpoints in the support: $v_s, v_{s+i-1} \in \supp(D)$;
\item No $K$ consecutive gaps: Within $(v_s, \cdots, v_{s+i-1})$, there do not exist $K$ consecutive vertices that lie outside $\supp(D)$;
\item Maximality at the boundary: If $i \le n- 2K$, then the $K$ predecessors and $K$ successors are outside $\supp(D)$, that is, $(v_{s-K}, \cdots, v_{s-1}), (v_{s+i}, \cdots, v_{s+i-1+K})\subset [n] \setminus \supp(D)$. %\nbr{$v_{s+i}$ should be $v_{s+K}?$}
Otherwise, $[n]\setminus (v_s, \cdots, v_{s+i-1})\subset [n] \setminus \supp(D)$.
\end{enumerate}
Let $I_{s, i, j}$ be the indicator variable that equals $1$ if the length-$i$ sequence starting at $v_s$, $(v_s, \cdots, v_{s+i-1})$,  is a maximal component of $\supp(D)$ with $j$ vertices outside $\supp(D)$, and $0$ otherwise. 
Similarly, let $I_{n,j}$ be the indicator that equals $1$ if the entire vertex set $[n]$ forms a single connected component of $\supp(D)$, with $j$ vertices outside $\supp(D)$, and it does not contain $K$ consecutive vertices outside $\supp(D)$. Then the total number of odd components is given by
\$
\odd(C_{n,K}[\supp(D)]) =  \sum_{s=1}^{n} \sum_{i=1}^{n-K} \sum_{\substack{j=0\\j\colon i-j \text{ is odd}}}^{i} I_{s, i,j} + \sum_{\substack{j=0\\j\colon n-j \text{ is odd}}}^n I_{n,j}.
\$
%\nb{Why you need to give a formula for the number of connected components??? Shoundn't you give a formula for the number of odd components???}
For any component not equal to the entire circle (hence of length $i\le n-K$), there is a unique starting vertex $s$, a unique length $i$, and a unique $j$ (the number of vertices outside the support) such that exactly one indicator $I_{s,i,j}=1$, %$(v_s, \cdots, v_{s+i-1})$ is a maximal component; exactly one indicator $I_{s,i,j}=1$ corresponds to this component, 
ensuring no double counting. The only remaining case is when the entire vertex set $[n]$ forms a single component with no $K$ consecutive vertices outside $\supp(D)$ anywhere; this is captured by $\sum_j I_{n,j}$.

To compute $\E[I_{s,i, j}]$, we distinguish cases.
If $i=1$, the event $I_{s,1,0} = 1$ occurs when $v_s$ is an isolated node: $v_s\in \supp(D)$ and $(v_{s-K},\cdots, v_{s-1}), (v_{s+1}, \cdots, v_{s+K}) \subset [n]\setminus \supp(D)$. Thus, $\E[I_{s,1,0}] = pq^{2K}$. Observe that $I_{s,1, 1} = 0$. %\nbr{This computes the expected isolated nodes, which is a lot smaller than the expected number of (odd) components. Maybe somewhere it is worthwhile to discuss the intuition behind the bound $npq^K$ and contrast this with the expected number of isolated vertices. I think this is roughly the expected number of fully deleted $K$-segment... } 
When $2\le i\le n-K$, for $I_{s,i,j} = 1$, both endpoints must belong to $\supp(D)$. Among the remaining $i-2$ internal vertices, there are $j$ that lie outside $\supp(D)$. %, with $0\le j \le i-2$. 
We encode these internal vertices as a binary sequence of length $i-2$, where a one denotes a vertex outside the support. %The condition that no $K$ consecutive vertices lie outside $\supp(D)$ is equivalent to requiring that the binary sequence has no $K$ consecutive ones. 
Thus, by definition, $\alpha(i-2,j,K)$ counts the number of admissible realizations of $\supp(D)$ with $j$ outside the support and no $K$ consecutive outside the support. For $2\le i\le n-2K$, the probability for each such realization is $p^{i-j} q^{2K+j}$, since there are $i-j$ vertices in $\supp(D)$, and $2K+j$ vertices outside $\supp(D)$ (the $j$ internal ones plus the $K$ predecessors and $K$ successors required by the maximality). Summing over all  $i$ and $j$, we obtain
\$
\sum_{i=2}^{n-2K} \sum_{\substack{j=0\\j\colon i-j \text{ is odd}}}^{i} \E[I_{s, i,j}] = \sum_{i=2}^{n-2K} \sum_{\substack{j=0\\j\colon i-j \text{ is odd}}}^{i-2} 
\alpha(i-2,j,K)p^{i-j}q^{2K+j}.
\$
For $n-2K < i\le n-K$, the probability for each such realization is $p^{i-j} q^{n-i+j}$, since there are $n-i+j$ vertices outside $\supp(D)$ (the $j$ internal vertices plus the $n-i$ vertices not in the sequence required by the maximality). 
Thus, summing over $i$ and $j$, we obtain
\$
\sum_{i=n-2K+1}^{n-K} \sum_{\substack{j=0\\j\colon i-j \text{ is odd}}}^{i} \E[I_{s, i, j}] = \sum_{i=n-2K+1}^{n-K} \sum_{\substack{j=0\\j\colon i-j \text{ is odd}}}^{i-2} 
\alpha(i-2,j,K)p^{i-j}q^{n-i+j}.
\$
Finally, we consider $I_{n,j}$. For each $j$, $\alpha(n,j,K)$ gives the number of admissible realizations of $D$ in which $j$ vertices lie outside the support and no $K$ consecutive vertices lie outside, with each such realization having probability $p^{n-j}q^j$. Thus, we have
\$
\E[I_{n,j}] = %\sum_{j=0}^{n} 
\alpha(n,j,K)p^{n-j}q^j.
\$
By combining the results above and reindexing, we have
\#\label{eq:Xodd}
& \E[\odd(C_{n,K}[\supp(D)])] = n \bigg[ pq^{2K} + \sum_{i=0}^{n-2K-2} \sum_{ \substack{j=0 \\ j\colon {i-j} \text{ is odd}}}^{i} \alpha(i,j,K)p^{i-j+2}q^{2K+j}  \\
& \qquad \qquad \qquad \qquad + \sum_{i=n-2K-1}^{n-K-2} \sum_{\substack{j=0 \\ j\colon {i-j} \text{ is odd}}}^{i}\alpha(i,j,K)p^{i-j+2}q^{n-i + j -2} \bigg] + \sum_{\substack{j=0 \\ j\colon {n-j} \text{ is odd}}}^{n} \alpha(n,j,K)p^{n-j}q^j. \notag   
\#

\paragraph{Lower bound.}
Since all terms are non-negative, \eqref{eq:Xodd} can be lower bounded by
\#\label{eq:exp-odd-lb}
\E[\odd(C_{n,K}[\supp(D)])] & \ge n \sum_{i=0}^{n-2K-2} \sum_{ \substack{j=0 \\ j\colon {i-j} \text{ is odd}}}^{i} \alpha(i,j,K)p^{i-j+2}q^{2K+j} \notag\\
& = np^2q^{2K} (F_{K,n-2K-2}(p,q/p) - F_{K,n-2K-2}(-p,-q/p))/2,
\#
recalling the definition of $F_{k,n}(x,y)$ from \eqref{eq:def-gen-f-finite-remaining}. Here the generating function provides a convenient way to compute the sum over odd terms. Substituting the closed forms of $F_{K}(p,q/p)$ and $F_{K}(-p,-q/p)$ from Lemma \ref{lem:no-k-cons}, we have
\$
& (F_{K,n-2K-2}(p,q/p) - F_{K,n-2K-2}(-p,-q/p))/2 \\ 
& \qquad \ge (F_{K}(p,q/p) - F_{K}(-p,-q/p) - 2 R_{K,n-2K-2}(p,q/p))/2\\
& \qquad = \frac{(1-q^K)^2}{pq^{K}(2-q^K)} - R_{K,n-2K-2}(p,q/p).
\$
It remains to bound the remainder term $R_{K,n-2K-2}(p,q/p)$. Lemma \ref{lem:no-k-cons} yields that
\#\label{eq:bound-R}
R_{K,n-2K-2}(p,q/p) %& \le  \frac{(1-q^K)}{pq^{K}(1+pq^K)^{n-3K-2}} +  2n^{K}\left( \max\left\{ p,  {1}/({1+pq^K}) \right\}\right)^{n-3K-3} \\
& \le \frac{1-q^K}{pq^K(1+pq^K)^{n-3K-2}} 
\le \frac{1}{pq^K\exp((n-3K-2)pq^K/2)}, %\frac{1}{(1+c/n)^{n-3K-2}} %\notag \\& \le e^{-c/2}/pq^K,
\#
%\nbr{I think we do not need any assumption. Just use the fact that $1+x \ge \exp(x/2)$ for $x \in [0,1]$. We get $$(1+pq^k)^{n-3K-2} \ge \exp(- (n-3K-2)pq^k/2).$$}
where the second inequality holds since $1+x \ge \exp(x/2)$ for any $x \in [0,1]$. %from $pq^K \ge c/n$, and the last inequality holds since $(1+x)^{n} \ge e^{n(x-x^2/2)}$  for any $x\in(0,1)$ and sufficiently large $n$. %\nbr{There are only three inequalitites...}. 
Substituting \eqref{eq:bound-R} into \eqref{eq:exp-odd-lb}, we have
\$ 
\E[\odd(C_{n,K}[\supp(D)])] \ge npq^{K}\bigg(\frac{(1-q^K)^2}{2-q^K} - \exp(-(n-3K-2)pq^K/2)\bigg).
\$

\paragraph{Upper Bound.} For the upper bound of \eqref{eq:Xodd}, we first note that, 
\$ 
& \sum_{i=0}^{n-2K-2} \sum_{ \substack{j=0 \\ j\colon {i-j} \text{ is odd}}}^{i} \alpha(i,j,K)p^{i-j+2}q^{2K+j} \le \sum_{i=0}^{\infty} \sum_{ \substack{j=0 \\ j\colon {i-j} \text{ is odd}}}^{i} \alpha(i,j,K)p^{i-j+2}q^{2K+j} \\
&\qquad \qquad= p^2q^{2K} (F_{K}(p,q/p) - F_{K}(-p,-q/p))/2 = \frac{pq^{K}(1-q^K)^2}{2-q^K}.
\$
In addition, the remaining terms satisfy
\$
\sum_{i=n-2K-1}^{n-K-2} \sum_{\substack{j=0 \\ j\colon {i-j} \text{ is odd}}}^{i}\alpha(i,j,K)p^{i-j+2}q^{n-i + j -2} &\le   \sum_{i=n-2K-1}^{n-K-2} p^2q^{n-i-2}\sum_{\substack{j=0}}^{i}\binom{i}{j}p^{i-j}q^{j} \\
& \le \sum_{i=n-2K-1}^{n-K-2} p^2q^{n-i-2}\le pq^K,  \$
and
\$
\sum_{\substack{j=0 \\ j\colon {n-j} \text{ is odd}}}^{n} \alpha(n,j,K)p^{n-j}q^j &\le \sum_{j=0}^{n} \binom{n}{j} p^{n-j}q^j = (p+q)^n = 1.
\$ 
%where the first line follows from \eqref{eq:bound-R} for sufficiently large $n$. \nbr{where did you use this???}
Putting all these together, we have
\$
\E[\odd(C_{n,K}[\supp(D)])] & \le n p q^K \bigg(q^{K} + \frac{(1-q^K)^2}{2-q^K} + 1  \bigg) + 1. \qedhere
\$
\end{proof}

\end{appendix}

\end{document}